\documentclass[12pt]{amsart}

\usepackage[all]{xy}
\usepackage{amsaddr}
\usepackage{amssymb}
\usepackage{color}
\usepackage{mathdots}
\usepackage{hyperref}

\usepackage{titlesec}
\titleformat{name=\section}{\normalfont\Large\bfseries}{\thesection}{1em}{}

\theoremstyle{definition}
\newtheorem{theorem}{Theorem}[section]
\newtheorem{proposition}[theorem]{Proposition}
\newtheorem{lemma}[theorem]{Lemma}
\newtheorem{corollary}[theorem]{Corollary}

\newtheorem{definition}[theorem]{Definition}
\newtheorem{example}[theorem]{Example}
\newtheorem{remark}[theorem]{Remark}
\newtheorem{notation}[theorem]{Notation}

\newtheorem{conjecture}[theorem]{Conjecture}

\makeatletter
\@addtoreset{equation}{section}
\makeatother

\usepackage[top=30mm,bottom=30mm,left=25mm,right=25mm]{geometry}

\begin{document}

\title[Exact calculation of degrees for lattice equations: a singularity approach]{Exact calculation of degrees for lattice equations: a singularity approach}

\author[T. Mase]{Takafumi Mase}

\address{Graduate School of Mathematical Sciences, 
the University of Tokyo, 3-8-1 Komaba, Meguro-ku, Tokyo 153-8914, Japan.}




\begin{abstract}
	The theory of degree growth and algebraic entropy plays a crucial role in the field of discrete integrable systems.
	However, a general method for calculating degree growth for lattice equations (partial difference equations) is not yet known.
	Here we propose a method to rigorously compute the exact degree of each iterate for lattice equations.
	Halburd's method, which is a novel approach to computing the exact degree of each iterate for mappings (recurrence relations, typically from ordinary difference equations) from the singularity structure, forms the basis of our idea.
	The strategy is to extend this method to lattice equations.
	First, we illustrate, without rigorous details, how to calculate degrees for lattice equations using the lattice version of Halburd's method, and outline the issues that must be resolved to make the method rigorous.
	We then provide a framework to ensure that all calculations are accurate and rigorous.
	We further address how to detect the singularity structure in lattice equations.
	Our method is not only accurate and rigorous but can also be easily applied to a wide range of lattice equations.
\end{abstract}


\begingroup
\def\uppercasenonmath#1{} 
\let\MakeUppercase\relax 

\maketitle

\endgroup


\section{Introduction}\label{section:introduction}

It has been more than thirty years since the study of integrability tests in discrete systems started.
Integrability tests are criteria that predict whether or not a given equation is integrable.

The first integrability test for discrete systems in history is singularity confinement \cite{singularity_confinement}.
It was proposed as a discrete analogue of the so-called Painlev\'{e} property.
Roughly speaking, a discrete equation enters a singularity if it loses information on initial values, and an equation is said to pass the singularity confinement test if the lost information is recovered after a finite number of iterations.
However, it is widely known today that singularity confinement is not perfect as an integrability test.
Not only the so-called Hietarinta-Viallet equation \cite{hietarinta_viallet} but a large number of non-integrable systems that pass the singularity confinement test have been discovered and constructed \cite{bedford_kim,redemption,redeeming}.
It is also known that most linearizable equations possess non-confined singularities \cite{sc_linearizable}.

The most accurate test for discrete integrability today is algebraic entropy \cite{entropy}.
The algebraic entropy of a discrete equation (defined by a rational function) is defined as follows.
Each iterate of an equation, i.e., each term in the sequence defined by the equation, can be expressed as a rational function in the initial variables and we focus on its degree.
Algebraic entropy measures how fast the degree sequence grows as the independent variable tends to infinity.
Roughly speaking, the algebraic entropy of an equation is $0$ (resp.\ positive) if the degrees grow polynomially (resp.\ exponentially).
If the algebraic entropy of an equation is $0$, then the zero algebraic entropy criterion deems the equation integrable.
The integrability test by algebraic entropy shows high empirical accuracy.

As seen above, the algebraic entropy is defined by degree growth.
However, it is very difficult in general to compute explicitly the degrees for a given equation.
Therefore, developing a method to calculate degree growth is an important problem in the field of discrete integrable systems.

In the case of mappings of rank $2$ (i.e., $2$-dimensional recurrences on a $1$-dimensional lattice, most commonly associated with ordinary difference equations), a general theory of degree growth has been established thanks to algebro-geometric methods \cite{gizatullin,cantat,diller_favre,takenawa1}.
Here, ``rank'' refers to the number of independent constants in the general solution (degrees of freedom), as commonly used in the theory of ordinary differential and difference equations.
Calculating degree growth for a rank-$2$ mapping is not difficult if it has good singularity structure, such as to pass the singularity confinement test.
For example, if a mapping of rank $2$ passes the singularity confinement test, one can construct an algebraic surface called a space of initial conditions, and the algebraic entropy coincides with the logarithm of the maximum eigenvalue of the linear action on the Picard group.
In the case of higher-rank mappings, however, no general classification theory is known yet.
How to construct a space of initial conditions and how to calculate degree growth are now being studied \cite{bedford_kim_higher,viallet_2015,carstea_takenawa_4d,alonso_suris_wei1,alonso_suris_wei2}.

Recently, a ground-breaking idea to calculate degree growth in the case of mappings has been proposed by Halburd \cite{halburd}.
The idea is to directly calculate the exact degree sequence from singularity patterns.
Using this strategy, Halburd calculated the degree sequences for several rank-$2$ mappings from the singularity patterns.
We will review Halburd's method in Example~\ref{example:first_example_1d}.

It would be natural to think that, now that we have a general classification theory of degree growth in the case of rank-$2$ mappings, the next class of equations we should consider is higher-order mappings, not lattice equations.
One of the reasons why lattice equations are much harder to analyze in general than higher-rank mappings is that an initial value problem for a lattice equation has infinitely many initial variables.
However, as will be shown in the main part of this paper, our theory is based on this infinite nature.
Unfortunately, it is not expected that our theory can be applied to the case of mappings (see the discussion in \textsection\ref{section:conclusion}).

As for degree growth in the case of lattice equations (partial difference equations), a few previous studies are known, such as \cite{tremblay,lattice_algebraic_entropy,lattice_factorization,gubbiotti_scimiterna_levi,investigation,tran1,tran2,domain,face_centered_quad,lattice_system_entropy,hietarinta_3x3}.
For now, however, even rigorously calculating the degree growth for a concrete lattice equation is not easy.
In the case of lattice equations, a choice of domain (initial value problem) is essentially important when considering degree growth.
One might think that whether the degree growth of a lattice equation is exponential or polynomial does not depend on the choice of domain.
According to \cite{domain}, however, if considered on a pathological domain, an integrable equation such as the discrete KdV equation has exponential degree growth.
The conclusion of \cite{domain} is that a domain on which we consider an initial value problem must satisfy some conditions, which we will describe in Definition~\ref{definition:domain}, and that it is a good strategy to consider the growth of individual degrees (degrees with respect to an arbitrary single initial variable) instead of total degrees (degrees with respect to all initial variables).
For each iterate and every initial variable $z$, we have
\[
	(\text{individual degree for $z$}) \le (\text{total degree}) \le \sum_w (\text{individual degree for $w$}),
\]
where the sum on the right runs over all initial variables $w$.
Under the assumption that the number of the initial variables involved in an iterate increases at most polynomially, the total degree grows at most polynomially if and only if each individual degree does.
Therefore, developing a method to calculate individual degrees is an important problem in the study of integrability tests for lattice equations.

In this paper, we give a method to calculate the exact degree sequences for lattice equations by extending Halburd's method to lattice equations.
We will not only calculate the degrees for concrete equations but also create a general framework to guarantee without much computation that our degree calculations are all accurate and rigorous.

Let us review the basic idea and strategy of Halburd's method on the following example.
It should be noted in advance, however, that while the calculations in the example are all correct, some discussion is not rigorous because we omit some arguments.

\begin{example}\label{example:first_example_1d}
	Let us calculate the degree sequence of the $3$-point mapping
	\[
		x_n = \frac{(x_{n - 1} + 3 a) x_{n - 2} - 2 a x_{n - 1}}{x_{n - 1} - 3 a},
	\]
	where $a \ne 0$ is a constant.
	Note that a $3$-point mapping is a $2$-dimensional ordinary difference equation of the form $x_n = F(x_{n - 1}, x_{n - 2})$.
	This equation was first introduced in \cite{miura_transformations} and the geometric analysis was performed in \cite{phd}.

	This equation has two singularity patterns
	\begin{itemize}
		\item
		$(3 a, \infty, a, \infty, - a, \infty, - 3 a)$,

		\item 
		$(- 3 a, - a, a, 3 a)$,

	\end{itemize}
	which are both confining and open (not cyclic).

	First, we review the definition of singularity confinement on the first pattern.
	Suppose that $(x_{n - 1}, x_n) = (u, 3 a)$ for generic $u$.
	Note that the pair $(x_{n - 1}, x_n)$ has one degree of freedom since $u$ is generic.
	The next iterate is $x_{n + 1} = \infty$, and the pair $(x_n, x_{n + 1}) = (3 a, \infty)$ has no degree of freedom.
	That is, the information on the initial value $u$ is lost, which is the definition of the equation entering a singularity.
	To calculate further, we introduce an infinitesimal parameter $\varepsilon$ and start with $(x_{n - 1}, x_n) = (u, 3 a + \varepsilon)$ instead.
	Then, one obtains
	\begin{align}\label{equation:singularity_calculation}
		x_{n + 1} &= O \left( \varepsilon^{- 1} \right), \notag \\
		x_{n + 2} &= a + O \left( \varepsilon \right), \notag \\
		x_{n + 3} &= O \left( \varepsilon^{- 1} \right), \notag \\
		x_{n + 4} &= - a + O \left( \varepsilon \right), \\
		x_{n + 5} &= O \left( \varepsilon^{- 1} \right), \notag \\
		x_{n + 6} &= - 3 a + O \left( \varepsilon \right), \notag \\
		x_{n + 7} &= - u + O \left( \varepsilon \right), \notag
	\end{align}
	where the exponent of each $O \left( \varepsilon^{\ell} \right)$ above coincides with the leading order of the remaining part.
	The important point here is that $x_{n + 7}$ is regular, in the sense that it depends on the initial value $u$, and that the pair $(x_{n + 6}, x_{n + 7}) = (- 3 a, - u)$ possesses one degree of freedom.
	That is, the lost degree of freedom is recovered in a finite number of steps.
	This is the definition of singularity confinement, and we say that this pattern is confining.
	On the other hand, a singularity pattern is non-confining if the lost degree of freedom is not recovered.
	We say an equation is confining (or an equation passes the singularity confinement test) if all singularity patterns of the equation are confining.
	
	Even in the context of singularity confinement, what the term ``singularity'' exactly denotes is not very clear.
	In this paper, for a $3$-point mapping, we use this term to denote an iterate of a singularity pattern.
	More precisely, a singularity $x_m = b$ consists of the following data:
	\begin{itemize}
		\item
		the value $b$,

		\item
		the position $m$, i.e., the value of the independent variable,

		\item
		the multiplicity, which we will define in the next paragraph.

	\end{itemize}
	When we are interested only in the value of a singularity, we use the term ``singular value.''
	For example, the singularity pattern $(3 a, \infty, a, \infty, - a, \infty, - 3 a)$ consists of $7$ singularities but has $5$ singular values: $\pm 3 a, \pm a, \infty$.
	We will define a singularity for lattice equations in the main part of the paper (Definition~\ref{definition:constant_singularity}).

	Let us review the definition of the multiplicity of a singularity.
	Consider a singularity $x_m = b$ for $b \in \mathbb{P}^1(\mathbb{C}) = \mathbb{C} \cup \{ \infty \}$.
	If $b \in \mathbb{C}$, the expansion of $x_m$ around $\varepsilon = 0$ is
	\[
		x_m = b + f(u) \varepsilon^r + O \left( \varepsilon^{r + 1} \right)
	\]
	for some $f(u) \ne 0$ and $r > 0$.
	On the other hand, if $b = \infty$, then the expansion is
	\[
		x_m = f(u) \varepsilon^{- r} + O \left( \varepsilon^{- r + 1} \right)
	\]
	for some $f(u) \ne 0$ and $r > 0$.
	In both cases, we define the multiplicity of the singularity $x_m = b$ as $r$.
	Note that this definition coincides with the definition of the multiplicity of $x_m = b$ (as a rational function in $\varepsilon$) at $\varepsilon = 0$ (see Definition~\ref{definition:rational_function_multiplicity}).
	For example, it follows from the calculation from $x_{n + 1}$ to $x_{n + 7}$ that the multiplicities of the singularities of the pattern $(3 a, \infty, a, \infty, - a, \infty, - 3 a)$ are all $1$.
	
	Note that the first pattern does not continue with the second one since the next iterate after $(3 a, \infty, a, \infty, - a, \infty, - 3 a)$ is regular, i.e., it depends on $u$ (see the calculation of $x_{n + 7}$ in \eqref{equation:singularity_calculation}).
	It is known (but does not directly follow from the calculation \eqref{equation:singularity_calculation}) that the pattern $(3 a, \infty, a, \infty, - a, \infty, - 3 a)$ has no other singularities \cite{phd}.
	In particular, this pattern is not cyclic, in the sense that the pair $(g(u), 3 a)$ with non-constant function $g(u)$ does not appear afterward.
	A confining singularity pattern is defined to be open if it has this non-cyclic property.

	It is known that the singularity pattern $(- 3 a, - a, a, 3 a)$ is also confining and open \cite{phd}.
	Calculating the Jacobian of the map $(x_{n - 1}, x_n) \mapsto (x_n, x_{n + 1})$, one obtains that the equation has no other singularity pattern.
	Therefore, this equation passes the singularity confinement test.

	To compute the degree sequence by Halburd's method, we need to know when the equation takes singular values.
	Searching for the possibility of the equation taking the values $\pm a$, $\pm 3 a$, $\infty$, one obtains that the equation has the following periodic patterns of period $2$ and $5$, respectively:
	\begin{itemize}
		\item
		$(\cdots, \text{reg}, \infty, \text{reg}, \infty, \cdots)$,

		\item
		$(\cdots, \text{reg}, a, \text{reg}, - a, \text{reg}, \text{reg}, a, \text{reg}, - a, \text{reg}, \cdots)$,

	\end{itemize}
	where ``$\text{reg}$'' denotes some regular value.
	It should be stressed that these patterns are \emph{not} singularity patterns since the equation does not lose the dependence on the initial value.
	For example, we obtain the first pattern by starting from the initial value $(x_{n - 1}, x_n) = \left( u, \infty \right)$ with generic $u$.
	Iterating the recurrence, we have
	\[
		x_{n + 1} = \infty, \quad
		x_{n + 2} = f_2(u), \quad
		x_{n + 3} = \infty, \quad
		x_{n + 4} = f_4(u), \quad \cdots,
	\]
	where each $f_j(u)$ is a non-constant function.
	Observe that in this example, $f_j(u)$ coincides with $u$, though this is not essential.
	What is important here is that each pair $(x_m, x_{m + 1})$ has one degree of freedom and that the dependence on the initial value $u$ is preserved throughout.
	Therefore, this pattern is not a singularity pattern.
	In the second pattern, each pair contains at least one regular value.
	Therefore, this pattern is not a singularity pattern, either.

	A direct calculation shows that these two periodic patterns extend to both sides.
	Geometrically speaking, each pattern corresponds to the motion of curves
	\[
		\cdots \to \{ y = \infty \} \to \{ x = \infty \} \to \{ y = \infty \} \to \cdots
	\]
	and
	\[
		\cdots \to \{ y = a \} \to \{ x = a \} \to \{ y = - a \} \to \{ x = - a \} \to \left\{ y = - \frac{3 a (x + a)}{x - 3 a} \right\} \to \{ y = a \} \to \cdots,
	\]
	respectively, where $(x_n, x_{n + 1}) = (x, y) \in \mathbb{P}^1(\mathbb{C}) \times \mathbb{P}^1(\mathbb{C})$.
	To obtain these patterns, we must take initial values that represent one of the above curves, such as $(x_0, x_1) = (\text{reg}, \infty)$.

	Let us fix an initial condition as
	\begin{itemize}
		\item
		$x_0$ and $x_1$ are the initial values,

		\item 
		$x_0$ is a generic $\mathbb{C}$-value,

		\item
		$x_1 = z$ is an initial variable,

	\end{itemize}
	and think of $x_n$ as a rational function in $z$, i.e., $x_n = x_n(z) \in \mathbb{C}(z)$.
	Our goal is to calculate $\deg_z x_n(z)$, the degree with respect to $z$.

	The key idea of Halburd's method is to calculate $\deg_z x_n(z)$ as the number of preimages of $x_n = \text{(constant)}$ counted with multiplicity for various singular values.
	It is well known that the degree of a rational function $f \colon \mathbb{P}^1(\mathbb{C}) \to \mathbb{P}^1(\mathbb{C})$, which is defined as the maximum of the degrees of its numerator and denominator when written in lowest terms, coincides with the number of the preimages $f^{- 1}(\alpha)$ counted with multiplicity for any $\alpha \in \mathbb{P}^1(\mathbb{C})$ (Proposition~\ref{proposition:degree_preimages}).
	This property is based on the fact that any degree-$m$ polynomial over an algebraically closed field has exactly $m$ roots, counted with multiplicities.
	
	First, let us verify all the possibilities where $\pm 3 a$ or $\infty$ appears as a value.
	Since the term ``spontaneously'' plays a key role in what follows but lacks a precise definition in the literature, we clarify its meaning here:
	for $z^{*}, \alpha \in \mathbb{P}^1(\mathbb{C})$, we say that $x_n(z^{*})$ spontaneously becomes $\alpha$ if $x_0(z^{*}), \ldots, x_{n - 1}(z^{*})$ are generic and $x_n(z^{*}) = \alpha$.
	This interpretation is adopted only for the present discussion.
	
	According to the singularity analysis above, we know the following:
	\begin{itemize}
		\item
		If $x_n(z)$ spontaneously becomes $3 a$ for $z = z^{*} \in \mathbb{P}^1(\mathbb{C})$, then we have $x_{n + 1}(z^{*}) = \infty$, $x_{n + 3}(z^{*}) = \infty$, $x_{n + 5}(z^{*}) = \infty$, $x_{n + 6}(z^{*}) = - 3 a$.
		That is, if $x_{n - 1}(z^{*})$ is generic and $x_n(z^{*} + \varepsilon) = 3 a + C \varepsilon^{\ell} + O \left( \varepsilon^{\ell + 1} \right)$
		where $\varepsilon$ is infinitesimal, $C \ne 0$ and $\ell \in \mathbb{Z}_{> 0}$,
		then we have $x_{n + 1}(z^{*} + \varepsilon) \sim \varepsilon^{- \ell}$, $x_{n + 3}(z^{*} + \varepsilon) \sim \varepsilon^{- \ell}$, $x_{n + 5}(z^{*} + \varepsilon) \sim \varepsilon^{- \ell}$ and $x_{n + 6}(z^{*} + \varepsilon) + 3 a \sim \varepsilon^{\ell}$.
		Starting with $z = 3 a$ is included in this case, too.

		\item
		If $x_n$ spontaneously becomes $- 3 a$ at $z = z^{*}$, then we have $x_{n + 3} = 3 a$.

		\item
		If $z = \infty$, then we have $x_{2 n + 1} = \infty$ for $n \in \mathbb{Z}_{\ge 0}$.

		\item
		The periodic pattern of period $5$ does not generate $\pm 3 a$ or $\infty$ as a value.

	\end{itemize}
	We use the notation
	\[
		( x_n = 3 a, \infty, a, \infty, - a, \infty, - 3 a)
	\]
	to denote the singularity pattern $(3 a, \infty, a, \infty, - a, \infty, - 3 a)$ starting at index $n$.
	Let $Z^+_n$ be the number of $z^{*} \in \mathbb{P}^1(\mathbb{C})$ that gives the pattern
	$( x_n = 3 a, \infty, a, \infty, - a, \infty, - 3 a)$
	counted with multiplicity, i.e.,
	\[
		Z^+_n = \# \left\{ z^{*} \in \mathbb{P}^1(\mathbb{C}) \mid \text{$z = z^{*}$ generate this pattern} \right\}
		\quad (\text{with multiplicity}).
	\]
	Here, if the first singularity of a pattern is $x_n(z^{*}) = \alpha \in \mathbb{P}^1(\mathbb{C})$, then we define the multiplicity of the pattern as the multiplicity of $x_n = \alpha$ at $z = z^{*}$.
	For example, the multiplicity of the pattern $( x_n = 3 a, \infty, a, \infty, - a, \infty, - 3 a)$ is the exponent $\ell$ in the above list.
	We set $Z^+_n = 0$ ($n \le 0$) for convenience.
	Then, using the number of preimages of $\infty$, we obtain
	\[
		\deg_z x_n(z) = Z^+_{n - 1} + Z^+_{n - 3} + Z^+_{n - 5} + \phi^{(2)}_{n + 1}
		\quad (n \ge 0),
	\]
	where $\phi^{(2)}_{2 m} = 1$ and $\phi^{(2)}_{2 m + 1} = 0$ for $m \in \mathbb{Z}$.
	For example, the term $Z^+_{n - 1}$ corresponds to the number of $z^{*}$ that gives the pattern $( x_{n - 1} = 3 a, \infty, a, \infty, - a, \infty, - 3 a)$
	and the term $\phi^{(2)}_{n + 1}$ corresponds to the number of the periodic pattern $(x_0 = \text{reg}, z = \infty, \text{reg}, \infty, \text{reg}, \infty, \cdots)$.

	In a similar way to the above, using the number of preimages of $\pm 3 a$, one finds that
	\begin{align*}
		\deg_z x_n(z) &= Z^+_n + Z^-_{n - 3} & (\text{$\#$ of preimages of $3 a$}) \\
		&= Z^+_{n - 6} + Z^-_n & (\text{$\#$ of preimages of $- 3 a$}),
	\end{align*}
	where $Z^-_n$ is the number of $z^{*}$ that gives the pattern $( x_n = - 3 a, - a, a, 3 a)$ counted with multiplicity (we set $Z^-_n = 0$ for $n \le 0$).

	Since the degree is expressed in three ways, we have the following relations:
	\[
		Z^+_{n - 1} + Z^+_{n - 3} + Z^+_{n - 5} + \phi^{(2)}_{n + 1}
		= Z^+_n + Z^-_{n - 3}
		= Z^+_{n - 6} + Z^-_n,
	\]
	which can be written as
	\begin{align*}
		Z^+_n &= Z^+_{n - 1} + Z^+_{n - 3} - Z^-_{n - 3} + Z^+_{n - 5} + \phi^{(2)}_{n + 1}, \\
		Z^-_n &= Z^+_{n - 1} + Z^+_{n - 3} + Z^+_{n - 5} - Z^+_{n - 6} + \phi^{(2)}_{n + 1}.
	\end{align*}
	This system of linear equations with respect to $Z^{\pm}_n$ can be thought of as defining the evolution of $Z^{\pm}_n$ with the initial values
	\[
		Z^+_{- 4} = Z^+_{- 3} = Z^+_{- 2} = Z^+_{- 1} = Z^+_0 = 0, \quad
		Z^+_1 = 1, \quad
		Z^-_{- 1} = Z^-_0 = 0, \quad
		Z^-_1 = 1.
	\]
	Solving the system, one obtains
	\[
		\begin{matrix}
			n \colon & 0,& 1,& 2,& 3,& 4,& 5,& 6,& 7,& 8,& 9,& 10,& 11,& 12,& 13,& 14,& 15,& 16,& \cdots, \\
			Z^+_n \colon & 0,& 1,& 1,& 2,& 2,& 3,& 4,& 5,& 6,& 7,& 8,& 10,& 11,& 13,& 14,& 16,& 18,& \cdots, \\
			Z^-_n \colon & 0,& 1,& 1,& 2,& 3,& 4,& 6,& 7,& 9,& 11,& 13,& 16,& 18,& 21,& 24,& 27,& 31,& \cdots, \\
			\deg_z x_n(z) \colon & 0,& 1,& 1,& 2,& 3,& 4,& 6,& 8,& 10,& 13,& 15,& 19,& 22,& 26,& 30,& 34,& 39,& \cdots,
		\end{matrix}
	\]
	where the degrees are calculated by $\deg_z x_n(z) = Z^+_n + Z^-_{n - 3}$.

	Since the equations for $Z^{\pm}_n$ are linear, it is easy to study how fast $Z^{\pm}_n$ and $\deg_z x_n(z)$ grow.
	For example, eliminating $Z^-_n$, one obtains for $n \ge 3$
	\[
		Z^+_n - Z^+_{n - 1} - Z^+_{n - 3} + Z^+_{n - 4} - Z^+_{n - 5} + Z^+_{n - 6} + Z^+_{n - 8} - Z^+_{n - 9} = \phi^{(2)}_{n + 1} - \phi^{(2)}_{n}.
	\]
	Since the characteristic polynomial of the left-hand side decomposes into the product of cyclotomic polynomials, $Z^{\pm}_n$ and $\deg_z x_n(z)$ grow polynomially.

	One can show in the same way that the degree sequence with respect to the other initial variable, $x_0$, grows polynomially as well.
	Therefore, the algebraic entropy of this equation is $0$, and the integrability criterion by degree growth deems the equation integrable.
\end{example}

According to this example, Halburd's method consists of the following steps.
\begin{enumerate}
	\item
	Perform singularity analysis.
	One must calculate not only singularities but also non-singular periodic patterns.

	\item
	Let $z$ be one of the initial variables and let the others be generic numerical values.

	\item
	Choose some singular values $\{ \alpha_1, \ldots, \alpha_J \} \subset \mathbb{P}^1(\mathbb{C})$ appropriately from the singularity patterns.

	\item
	Introduce a new number sequence for each open singularity pattern that contains at least one of $\alpha_1, \ldots, \alpha_J$.
	For example, we use $Z_n(\beta_i)$ to denote the number of $z^{*} \in \mathbb{P}^1(\mathbb{C})$ with multiplicity that generates the singularity pattern starting with $x_n = \beta_i$.
	
	\item\label{enumerate:degree_relation}
	Using the singularity analysis, express by $Z_n(\beta_i)$ the number of preimages of $x_n(z) = \alpha_j$ for each $j$, which all coincides with $\deg x_n(z)$.
	If $\alpha_j$ appears in a periodic pattern, periodic terms are also necessary.
	One obtains an expression of $\deg_z x_n(z)$ for each $j = 1, \ldots, J$.
	We shall call these expressions \emph{degree relations} in this paper.

	\item\label{enumerate:linear_equations}
	Combining degree relations, obtain a system of linear recurrences for $Z_n(\beta_i)$, which can be thought of as defining the evolution for $Z_n(\beta_i)$ if the number of degree relations is sufficiently many.

	\item
	Solving this system, obtain $Z_n(\beta_i)$ and $\deg_z x_n(z)$.

\end{enumerate}

When one uses Halburd's method as a tool to easily calculate degree growth, there are two major problems.

The first problem lies in step~\ref{enumerate:linear_equations}.
As pointed out in \cite[Section~5]{express1}, there are cases where an equation has many singularity patterns but the degree relations do not determine an evolution because the number of linear equations is not sufficient.
It is known empirically that such a problem is likely to occur if singularity patterns are short.
However, the specific conditions for this phenomenon are not yet known.

The second problem, which lies in step~\ref{enumerate:degree_relation}, is more serious.
In this step, we must count the number of $z^{*} \in \mathbb{P}^1(\mathbb{C})$ with multiplicity such that $x_n(z^{*})$ becomes $\alpha_j$.
The discussion in Example~\ref{example:first_example_1d} might look rigorous but, as noted in advance, we did not guarantee that our counting of preimages is exact.
As seen before, when one tries to calculate the singularity pattern starting with $x_n = \beta_i$, the other initial value $x_{n - 1}$ is taken generic (i.e., a variable).
In Halburd's method, however, if $z = z^{*}$ generates $x_n(z^{*}) = \beta_i$ for some $z^{*} \in \mathbb{P}^1(\mathbb{C})$, then $x_{n - 1}(z^{*})$ also belongs to $\mathbb{P}^1(\mathbb{C})$.
Therefore, one must show that $x_{n + 1}(z^{*}), x_{n + 2}(z^{*}), \cdots$ actually follow the singularity pattern.
Furthermore, we must confirm that except for the open and periodic patterns we study, there is no possibility of $x_n(z^{*}) = \alpha_j$.
It is important to note, however, that this does not mean that Halburd's original calculations are not exact or rigorous.
In \cite{halburd}, Halburd thoroughly analyzed each equation to overcome this problem.
For example, in order to ensure the discussion in Example~\ref{example:first_example_1d} is rigorous in this way, we must at least ascertain that the singularity pattern $(3 a, \infty, a, \infty, - a, \infty, - 3 a)$ appears if $x_n(z^{*}) = 3 a$ and $x_{n - 1}(z^{*}) \notin \{ \pm a, \pm 3 a, \infty \}$.
Such a calculation is not very hard in the case of rank-$2$ mappings since each singularity starts with an easy condition, such as $x_n(z^{*})$ being some special value and $x_{n - 1}(z^{*})$ being generic.
However, it is much more difficult in the case of higher-rank equations or lattice equations because the degrees of freedom are large.

To resolve the second problem, some strategies are already known.
The first strategy is, as stated above, to thoroughly analyze an equation as Halburd did.
The second strategy is to abandon rigor and use the method as a convenient tool for predicting degree growth.
The express method was invented for this purpose and applied to many mappings \cite{express1,express_unconfined}.
Roughly speaking, the express method is a strategy to ignore the influence of periodic patterns and predict the degree growth.
The third strategy is to use some other theory to guarantee that the degrees (or the degree growth) calculated by Halburd's method (or by the express method) are exact.
For example, in the case of confining $3$-point mappings, the theory of spaces of initial conditions was used to show that the degree growth computed by the express method is exact \cite{express2}.

In addition to these two, there are some minor problems.
For example, the number of $z^{*}$ that generates a pattern counted with multiplicity should be defined rigorously in general cases.
Moreover, it is not very easy in general to confirm that a singularity pattern is open, which is not a problem of the method but of singularity analysis.


In this paper, we construct an easy but rigorous method to compute exact individual degrees for lattice equations.
For simplicity, in this paper, we focus on quad equations
\begin{equation}\label{equation:quad_equation}
	x_{t, n} = \Phi(x_{t - 1, n}, x_{t, n - 1}, x_{t - 1, n - 1}),
\end{equation}
where
\[
	\Phi = \Phi(B, C, D) \in \mathbb{C}(B, C, D)
\]
and $\Phi$ depends on all of $B$, $C$ and $D$.
Throughout this paper, we always think of the equation as defining the evolution in the northeast direction.
We will introduce some conditions on $\Phi$ (see Definitions~\ref{definition:d_factor_condition} and \ref{definition:basic_pattern_condition}) but do not assume other nonessential conditions.
For example, we do not assume that the equation can be solved in another direction.

The paper is organized as follows.
In \textsection\ref{section:introduction_calculation}, we introduce the lattice version of Halburd's method without rigorous discussion.
We will discuss what problems must be solved to make the method rigorous.
The main part of this paper is \textsection\ref{section:main} and \textsection\ref{section:proof}.
We will state our main theorems in \textsection\ref{section:main} and give proofs in \textsection\ref{section:proof}.
Since all we need to rigorously calculate degrees are stated in \textsection\ref{section:main}, a reader who is not interested in proofs can skip \textsection\ref{section:proof}.
In \textsection\ref{section:examples}, we analyze several examples to see how to calculate exact degrees with our theory.
For a reader who is not familiar with singularity analysis, we discuss how to detect and calculate singularity patterns, too.
In Appendix~\ref{appendix:domain}, we prove Theorem~\ref{theorem:domain_choice}, which states that when considering the degrees for a quad equation, there are only four essential choices of domain.
In Appendix~\ref{appendix:substitution_and_derivatives}, we recall some definitions and basic properties of substitution and derivative, which play an important role in \textsection\ref{section:proof}.
Appendix~\ref{appendix:algebraic_lemmas} consists of some algebraic lemmas, which we use in \textsection\ref{section:proof}.
In Appendix~\ref{appendix:divisors}, we recall the definition of divisors on $\mathbb{P}^1(K)$ and introduce some unusual notations.

\begin{notation}\label{notation:all}
	We summarize the notation that will be used in the main part of the paper.
	For clarity, we separate it into three groups.
	Notation that appears only in the introduction, is not important for understanding the main results, or occurs only within proofs is explained in the text and not included here.
	\begin{enumerate}
		\item
		Fundamental concepts and notation (excluding very basic symbols):
		\begin{itemize}
			\item
			$\operatorname{trdeg}_{K} K(y_1, \ldots, y_N)$:
			if some of $y_1, \ldots, y_N$ are $\infty$ when considering the transcendental degree of this type of field extension,
			we will think of $\infty$ as if it were an element of $K$.

			\item
			$\operatorname{div} (f = \alpha \mid T)$,
			$\operatorname{div} \left( f = \alpha \mid \text{(condition)} \right)$:
			a divisor on $\mathbb{P}^1(K)$.
			(Definition~\ref{definition:divisor})

			\item
			$T_1 \sqcup T_2$, $\bigsqcup_{\lambda} T_{\lambda}$:
			the direct sum of sets (disjoint union).

		\end{itemize}

		\item
		Notation for the setting (main symbols only):
		\begin{itemize}
			\item
			$\Phi = \Phi(B, C, D)$:
			a rational function that defines a quad equation \eqref{equation:quad_equation}.

			\item
			$H \subset \mathbb{Z}^2$:
			a domain on which we consider an equation.
			(Definition~\ref{definition:domain})

			\item
			$H_0 \subset H$:
			the initial boundary of a domain $H$.
			(Definition~\ref{definition:domain})

			\item
			$\le$:
			the product order on $\mathbb{Z}^2$, i.e., $(t, n) \le (s, m)$ if and only if $t \le s$ and $n \le m$.

			\item
			$z = x_{t_0, n_0}$:
			one of the initial variables, with respect to which we consider the degree.

			\item
			$\mathbb{K}$:
			the field of rational functions in the initial variables other than $z$.
			(Definition~\ref{definition:field_k})

		\end{itemize}

		\item
		Concepts and terminology used in this paper (key concepts only):
		\begin{itemize}
			\item
			$\text{reg}$:
			a regular value, i.e.\ a value that depends on some initial value.

			\item
			Degree Relation:
			the expression of the degree of an iterate as the number of preimages of a $\mathbb{P}^1(\mathbb{C})$-value.

			\item
			$\mathbb{Z}^2_{\le (t, n)} = \{ (t', n') \in \mathbb{Z}^2 \mid t' \le t; \ n' \le n \}$,
			$H_{\le (t, n)} = H \cap \mathbb{Z}^2_{\le (t, n)}$:
			the past light cone emanating from $(t, n)$.

			\item
			$\mathbb{Z}^2_{\ge (t, n)} = \{ (t', n') \in \mathbb{Z}^2 \mid t' \ge t; \ n' \ge n \}$,
			$H_{\ge (t, n)} = H \cap \mathbb{Z}^2_{\ge (t, n)}$:
			the future light cone emanating from $(t, n)$.

			\item
			Constant Singularity, Constant Singularity Pattern, Singular Value:
			Definition~\ref{definition:constant_singularity}.

			\item
			Movable, Fixed, Confining, and Solitary Patterns:
			Definition~\ref{definition:constant_singularity}.

			\item
			Basic Pattern:
			Definition~\ref{definition:basic_pattern}.

			\item
			$H^{\text{basic}}$:
			the domain corresponding to basic patterns.
			(Definition~\ref{definition:basic_pattern})

			\item
			$\partial$-factor condition:
			Definition~\ref{definition:d_factor_condition}.

			\item
			Basic Pattern Condition:
			Definition~\ref{definition:basic_pattern_condition}.

			\item
			First Singularity, Starting Value:
			Definition~\ref{definition:first_singularity}.

			\item
			Spontaneous Occurrence of $x_{t, n}(z^{*}) = \alpha$:
			Definition~\ref{definition:spontaneous_occurrence_rigorous}.

			\item
			$Z_{s, m}(\beta_i)$:
			the number of spontaneous occurrences of $x_{s, m}(z) = \beta_i$ with multiplicity.
			(Definition~\ref{definition:degree_divisor})

		\end{itemize}

	\end{enumerate}
\end{notation}

\section{Lattice version of Halburd's method}\label{section:introduction_calculation}

In this section, we see on an example how to extend Halburd's method to lattice equations.
It should be noted in advance, however, that while the calculations in the example are correct in the end (Example~\ref{example:first_example_2d_rigorous}), some discussion is not rigorous, as Example~\ref{example:first_example_1d}.
After the example, we will discuss what problems exist and how to fix them.

\begin{example}\label{example:first_example_2d}
	Let us consider Hirota's discrete KdV (Korteweg-de Vries) equation \cite{hirota_dkdv}:
	\[
		x_{t, n} = x_{t - 1, n - 1} - \frac{1}{x_{t - 1, n}} + \frac{1}{x_{t, n - 1}}.
	\]
	This equation is known to have the following confining singularity pattern $\text{I}_{t, n}$ \cite{singularity_confinement}:
	\[
		\text{I}_{t, n} \colon \quad
		\begin{matrix}
			x_{t, n + 1} = \infty & x_{t + 1, n + 1} = 0 \\
			x_{t, n} = 0 & x_{t + 1, n} = \infty
		\end{matrix}.
	\]
	That is, starting with the initial values
	\begin{equation}\label{equation:example2_domain_sc}
		\begin{matrix}
			{\color{red} \vdots} \\
			{\color{red} x_{t - 1, n + 2}} \\
			{\color{red} x_{t - 1, n + 1}} \\
			{\color{red} x_{t - 1, n}} & {\color{red} x_{t, n} = \varepsilon} \\
			& {\color{red} x_{t, n - 1}} & {\color{red} x_{t + 1, n - 1}} & {\color{red} x_{t + 2, n - 1}} & {\color{red} \cdots}
		\end{matrix}
	\end{equation}
	where $x_{t, n} = \varepsilon$ is infinitesimal and the others are initial variables, one obtains as $\varepsilon \to 0$
	\begin{equation}\label{equation:example2_sc_meaning}
		\begin{matrix}
			{\color{red} x_{t - 1, n + 2}} & \text{reg} & \text{reg} & \text{reg} \\
			{\color{red} x_{t - 1, n + 1}} & \infty & 0 & \text{reg} \\
			{\color{red} x_{t - 1, n}} & {\color{red} 0} & \infty & \text{reg} \\
			& {\color{red} x_{t, n - 1}} & {\color{red} x_{t + 1, n - 1}} & {\color{red} x_{t + 2, n - 1}}
		\end{matrix},
	\end{equation}
	where ``$\text{reg}$'' means that the leading order with respect to $\varepsilon$ is $0$ and the leading coefficient depends on some other initial variable.
	
	Let us fix an initial value problem as
	\[
		\begin{matrix}
			{\color{red} \vdots} & & & & & \iddots \\
			{\color{red} x_{0, 4}} & x_{1, 4} & x_{2, 4} & x_{3, 4} & x_{4, 4} \\
			{\color{red} x_{0, 3}} & x_{1, 3} & x_{2, 3} & x_{3, 3} & x_{4, 3} \\
			{\color{red} x_{0, 2}} & x_{1, 2} & x_{2, 2} & x_{3, 2} & x_{4, 2} \\
			{\color{red} x_{0, 1}} & {\color{red} z} & x_{2, 1} & x_{3, 1} & x_{4, 1} \\
			& {\color{red} x_{1, 0}} & {\color{red} x_{2, 0}} & {\color{red} x_{3, 0}} & {\color{red} x_{4, 0}} & {\color{red} \cdots}
		\end{matrix}.
	\]
	Each $x_{t, n}$ is expressed as a rational function in the initial values $x_{1, 1}$, $x_{s, 0}$ ($s \ge 1$) and $x_{0, m}$ ($m \ge 0$) and we focus on the degrees with respect to the initial variable $z = x_{1, 1}$.
	As in Example~\ref{example:first_example_1d}, we suppose that the other initial values are generic and think of each $x_{t, n}$ as a rational function in $z$ as $x_{t, n} = x_{t, n}(z) \in \mathbb{C}(z)$.
	Our basic strategy is, by substituting various values $z^{*} \in \mathbb{P}^1(\mathbb{C})$ for $z$ as in the case of mappings, to verify all the possibilities for $x_{t, n}(z)$ to become $0$ or $\infty$.

	As in Example~\ref{example:first_example_1d}, we consider the ``spontaneous occurrence'' of $x_{t, n}(z^{*}) = \alpha$.
	Although its mathematical definition will be given in Definition~\ref{definition:spontaneous_occurrence_rigorous}, we provide here a temporary interpretation:
	we say that $x_{t, n}(z^{*})$ spontaneously becomes $\alpha$ if $x_{t, n}(z^{*}) = \alpha$ and $x_{s, m}(z^{*})$ is generic for all $(s, m) \le (t, n)$ with $(s, m) \ne (t, n)$.

	According to the above pattern, if $z^{*} \in \mathbb{P}^1(\mathbb{C})$ makes $x_{t, n}(z^{*})$ spontaneously become $0$ for $t, n \ge 1$, then we have $x_{t + 1, n}(z^{*}) = \infty$, $x_{t, n + 1}(z^{*}) = \infty$ and $x_{t + 1, n + 1}(z^{*}) = 0$.
	In addition to the pattern $I_{t, n}$, the equation has the following pattern $\text{II}$ where $\infty$ appears as a value:
	\[
		\text{II} \colon \quad
		\begin{matrix}
			{\color{red} \vdots} & & & & & \iddots \\
			{\color{red} \bullet} & \circ & \circ  & \circ  & \infty \\
			{\color{red} \bullet} & \circ  & \circ  & \infty & \circ  \\
			{\color{red} \bullet} & \circ  & \infty & \circ  & \circ  \\
			{\color{red} \bullet} & {\color{red} \infty} & \circ  & \circ  & \circ \\
			& {\color{red} \bullet} & {\color{red} \bullet} & {\color{red} \bullet} & {\color{red} \bullet} & {\color{red} \cdots}
		\end{matrix}.
	\]
	While the starting point of $\text{I}_{t, n}$ varies depending on $(t, n)$ (or $z^{*}$), the pattern $\text{II}$ is fixed since it corresponds to $z^{*} = \infty$.
	Note that if one starts with $z^{*} = 0$, the pattern coincides with $\text{I}_{1, 1}$.

	It is known that, in our settings, $\text{I}_{t, n}$ and $\text{II}$ are the only possibilities where $\infty$ and $0$ appear as a value, whereas in other settings more possibilities may occur (see \cite{sc_kdv1,sc_kdv2}).
	How to prove this, which we will discuss later, is one of the main topics in this paper.

	Let $Z_{t, n}$ be the number of $z^{*} \in \mathbb{P}^1(\mathbb{C})$ counted with multiplicity that generates the pattern $\text{I}_{t, n}$ ($Z_{t, n} := 0$ for $n \le 0$ or $t \le 0$).
	For example, we have $Z_{1, 1} = 1$ since only starting with $z^{*} = 0$ generates the pattern $\text{I}_{1, 1}$.
	Then, as in Example~\ref{example:first_example_1d}, we have the following two expressions of degrees for $t, n \ge 0$, $(t, n) \ne (0, 0)$:
	\begin{align*}
		\deg_z x_{t, n}(z)
		&= Z_{t, n} + Z_{t - 1, n - 1} & (\text{$\#$ of preimages of $0$}) \\
		&= Z_{t - 1, n} + Z_{t, n - 1} + \delta_{t, n} & (\text{$\#$ of preimages of $\infty$}),
	\end{align*}
	where $\delta_{t, n}$ denotes Kronecker's delta.
	In the above relations,
	\begin{itemize}
		\item 
		$Z_{t, n}$, $Z_{t - 1, n - 1}$, $Z_{t - 1, n}$ and $Z_{t, n - 1}$ correspond to the patterns $\text{I}_{t, n}$, $\text{I}_{t - 1, n - 1}$, $\text{I}_{t - 1, n}$ and $\text{I}_{t, n - 1}$, respectively,

		\item
		$\delta_{t, n}$ corresponds to $\text{II}$.

	\end{itemize}
	From these degree relations, one obtains
	\[
		Z_{t, n} + Z_{t - 1, n - 1} = Z_{t - 1, n} + Z_{t, n - 1} + \delta_{t, n}
	\]
	for $t, n \ge 0$, $(t, n) \ne (0, 0)$,
	which can be thought of as defining the evolution of $Z_{t, n}$.
	Using this, $Z_{t, n}$ and $\deg_z x_{t, n}(z)$, which both grows polynomially, can be computed as
	\[
		Z_{t, n} \colon \quad
		\begin{matrix}
			{\color{red} 0} & 1 & 2 & 3 & 4 & 5 & 6 & 7 & 8 & 9 & 10 \\
			{\color{red} 0} & 1 & 2 & 3 & 4 & 5 & 6 & 7 & 8 & 9 & 9 \\
			{\color{red} 0} & 1 & 2 & 3 & 4 & 5 & 6 & 7 & 8 & 8 & 8 \\
			{\color{red} 0} & 1 & 2 & 3 & 4 & 5 & 6 & 7 & 7 & 7 & 7 \\
			{\color{red} 0} & 1 & 2 & 3 & 4 & 5 & 6 & 6 & 6 & 6 & 6 \\
			{\color{red} 0} & 1 & 2 & 3 & 4 & 5 & 5 & 5 & 5 & 5 & 5 \\
			{\color{red} 0} & 1 & 2 & 3 & 4 & 4 & 4 & 4 & 4 & 4 & 4 \\
			{\color{red} 0} & 1 & 2 & 3 & 3 & 3 & 3 & 3 & 3 & 3 & 3 \\
			{\color{red} 0} & 1 & 2 & 2 & 2 & 2 & 2 & 2 & 2 & 2 & 2 \\
			{\color{red} 0} & {\color{red} 1} & 1 & 1 & 1 & 1 & 1 & 1 & 1 & 1 & 1 \\
			& {\color{red} 0} & {\color{red} 0} & {\color{red} 0} & {\color{red} 0} & {\color{red} 0} & {\color{red} 0} & {\color{red} 0} & {\color{red} 0} & {\color{red} 0} & {\color{red} 0} \\
		\end{matrix}
	\]
	and
	\[
		\deg_z x_{t, n}(z) \colon \quad
		\begin{matrix}
			{\color{red} 0} & 1 & 3 & 5 & 7 & 9 & 11 & 13 & 15 & 17 & 19 \\
			{\color{red} 0} & 1 & 3 & 5 & 7 & 9 & 11 & 13 & 15 & 17 & 17 \\
			{\color{red} 0} & 1 & 3 & 5 & 7 & 9 & 11 & 13 & 15 & 15 & 15 \\
			{\color{red} 0} & 1 & 3 & 5 & 7 & 9 & 11 & 13 & 13 & 13 & 13 \\
			{\color{red} 0} & 1 & 3 & 5 & 7 & 9 & 11 & 11 & 11 & 11 & 11 \\
			{\color{red} 0} & 1 & 3 & 5 & 7 & 9 & 9 & 9 & 9 & 9 & 9 \\
			{\color{red} 0} & 1 & 3 & 5 & 7 & 7 & 7 & 7 & 7 & 7 & 7 \\
			{\color{red} 0} & 1 & 3 & 5 & 5 & 5 & 5 & 5 & 5 & 5 & 5 \\
			{\color{red} 0} & 1 & 3 & 3 & 3 & 3 & 3 & 3 & 3 & 3 & 3 \\
			{\color{red} 0} & {\color{red} 1} & 1 & 1 & 1 & 1 & 1 & 1 & 1 & 1 & 1 \\
			& {\color{red} 0} & {\color{red} 0} & {\color{red} 0} & {\color{red} 0} & {\color{red} 0} & {\color{red} 0} & {\color{red} 0} & {\color{red} 0} & {\color{red} 0} & {\color{red} 0}
		\end{matrix}.
	\]
	The degrees calculated here are in perfect agreement with those calculated in \cite{domain}, and the exact solutions for $Z_{t, n}$ and $\deg_z x_{t, n}(z)$ are
	\[
		Z_{t, n} = \min(t, n), \quad
		\deg_z x_{t, n}(z) = 2 \min(t, n) - 1.
	\]
\end{example}

Even in the case of lattice equations, the basic procedure of the method is almost the same as in the case of mappings.
Therefore, the problems we have are similar, too.

The first problem is the definition of singularities in the case of lattice equations.
Although singularity confinement for lattice equations was applied in the original paper \cite{singularity_confinement}, the search for its rigorous definition is still going on \cite{sc_kdv1,sc_kdv2,sc_mkdv2}.
Moreover, roughly speaking, all the singularities in our settings should have codimension one since all the initial values other than $z$ are taken generic.
Therefore, we must start by rigorously defining what is a singularity we consider in this paper.

The second problem is how an equation enters a singularity.
For example, when we calculated the singularity pattern starting with $x_{t, n} = 0$ in the singularity confinement calculation, the initial values are taken as \eqref{equation:example2_domain_sc}.
However, we must confirm that the pattern starting with a spontaneous occurrence of $x_{t, n}(z^{*}) = 0$ in fact coincides with $\text{I}_{t, n}$, which is calculated by singularity confinement.
The definition of the spontaneous occurrence of $x_{t, n}(z^{*}) = 0$ and its multiplicity should be clarified, too.

The third problem is the global structure of singularity patterns.
Singularity confinement is a local property.
For instance, when we calculated the singularity pattern starting with $x_{t, n} = 0$ in the example, we did not discuss the outside region of \eqref{equation:example2_sc_meaning}.
Therefore, singularity confinement does not guarantee that for any $z = z^{*} \in \mathbb{P}^1(\mathbb{C})$, an equation does not have any strange pattern.
For example, we must confirm that one singularity pattern does not have two or more singularity blocks (see Figures~\ref{figure:inexistent_pattern_1} and \ref{figure:inexistent_pattern_2}).
Here, a singularity block denotes a maximal subset of the singular part in which any two points can be joined by a finite sequence of steps of Euclidean length $1$ or $\sqrt{2}$.
This problem is slightly similar to how to distinguish open and periodic patterns, which we mentioned as a minor problem in the case of mappings.
In the case of lattice equations, however, this is not a minor problem anymore because there are too many degrees of freedom in initial values.
Moreover, we must show that each pattern has only one starting point (see Figure~\ref{figure:inexistent_pattern_3}).
Even in the case of Hirota's discrete KdV equation, many types of complicated singularity patterns have been reported \cite{sc_kdv1,sc_kdv2}.
Therefore, we must verify that assigning a value to $z$, while keeping all other initial values generic, does not lead to such complicated singularity patterns.

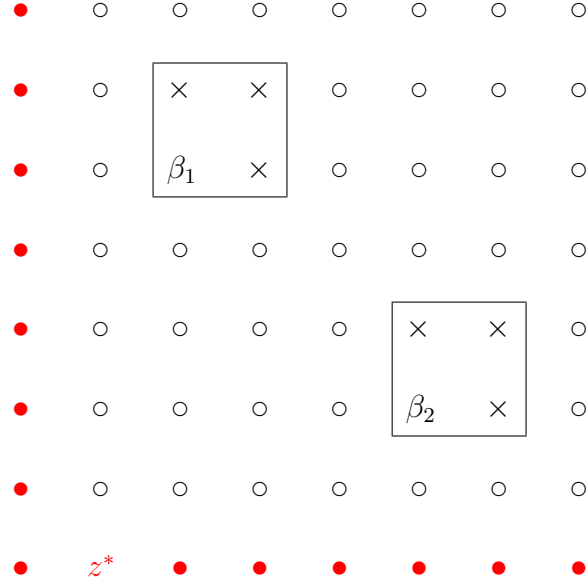
\begin{figure}
	\begin{center}
		\begin{picture}(240, 240)

			\multiput(40, 220)(30, 0){7}{\circle{5}}
			\multiput(40, 190)(30, 0){1}{\circle{5}}
			\multiput(40, 160)(30, 0){1}{\circle{5}}
			\multiput(40, 130)(30, 0){7}{\circle{5}}
			\multiput(40, 100)(30, 0){4}{\circle{5}}
			\multiput(40, 70)(30, 0){4}{\circle{5}}
			\multiput(40, 40)(30, 0){7}{\circle{5}}

			\multiput(130, 190)(30, 0){4}{\circle{5}}
			\multiput(130, 160)(30, 0){4}{\circle{5}}
			\multiput(220, 100)(30, 0){1}{\circle{5}}
			\multiput(220, 70)(30, 0){1}{\circle{5}}

			\multiput(65, 187)(30, 0){2}{$\times$}
			\multiput(95, 157)(30, 0){1}{$\times$}
			\put(65, 157){$\beta_1$}
			\put(60, 150){\framebox(50, 50)}

			\multiput(155, 97)(30, 0){2}{$\times$}
			\multiput(185, 67)(30, 0){1}{$\times$}
			\put(155, 67){$\beta_2$}
			\put(150, 60){\framebox(50, 50)}

			\color{red}
			\multiput(10, 10)(0, 30){8}{\circle*{5}}
			\multiput(70, 10)(30, 0){6}{\circle*{5}}

			\put(35, 7){$z^{*}$}

		\end{picture}
	\end{center}
	\caption{
		Example of an undesirable singularity pattern, where a single value $z^{*}$ generates two or more independent singularity blocks.
		Here, we call two singularity blocks independent if they cannot be compared under the product order on $\mathbb{Z}^2$.
		Such a phenomenon can never occur in the case of mappings, since all points are totally ordered on $\mathbb{Z}$.
	}
	\label{figure:inexistent_pattern_1}
\end{figure}

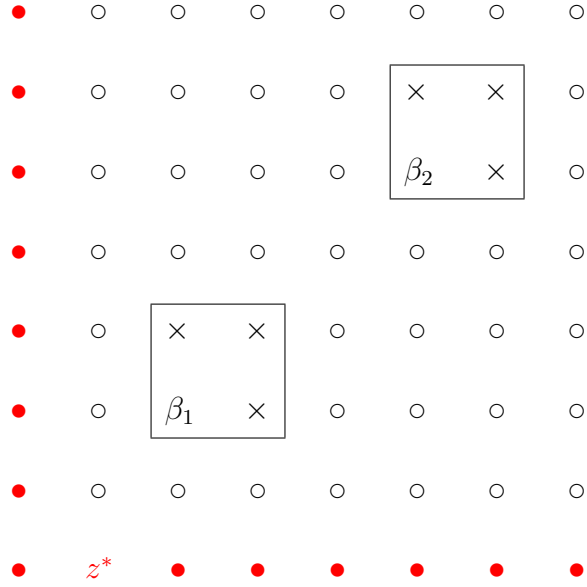
\begin{figure}
	\begin{center}
		\begin{picture}(240, 240)

			\multiput(40, 220)(30, 0){7}{\circle{5}}
			\multiput(40, 190)(30, 0){4}{\circle{5}}
			\multiput(40, 160)(30, 0){4}{\circle{5}}
			\multiput(40, 130)(30, 0){7}{\circle{5}}
			\multiput(40, 100)(30, 0){1}{\circle{5}}
			\multiput(40, 70)(30, 0){1}{\circle{5}}
			\multiput(40, 40)(30, 0){7}{\circle{5}}

			\multiput(220, 190)(30, 0){1}{\circle{5}}
			\multiput(220, 160)(30, 0){1}{\circle{5}}
			\multiput(130, 100)(30, 0){4}{\circle{5}}
			\multiput(130, 70)(30, 0){4}{\circle{5}}

			\multiput(65, 97)(30, 0){2}{$\times$}
			\multiput(95, 67)(30, 0){1}{$\times$}
			\put(65, 67){$\beta_1$}
			\put(60, 60){\framebox(50, 50)}

			\multiput(155, 187)(30, 0){2}{$\times$}
			\multiput(185, 157)(30, 0){1}{$\times$}
			\put(155, 157){$\beta_2$}
			\put(150, 150){\framebox(50, 50)}

			\color{red}
			\multiput(10, 10)(0, 30){8}{\circle*{5}}
			\multiput(70, 10)(30, 0){6}{\circle*{5}}

			\put(35, 7){$z^{*}$}

		\end{picture}
	\end{center}
	\caption{
		Similar to Figure~\ref{figure:inexistent_pattern_1}, but two singularity blocks are not independent, as after a singularity block is confined, another block appears.
		Such a phenomenon is common in the case of mappings since two or more patterns are sometimes combined and generate this type of pattern.
		A cyclic pattern of a mapping falls into this category, too.
	}
	\label{figure:inexistent_pattern_2}
\end{figure}

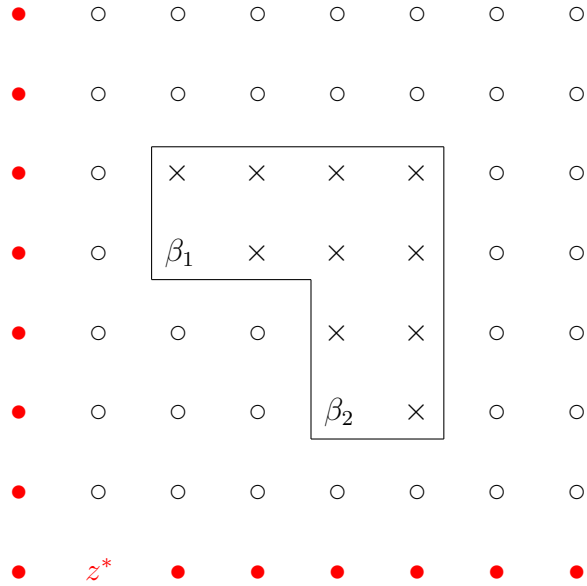
\begin{figure}
	\begin{center}
		\begin{picture}(240, 240)

			\multiput(40, 220)(30, 0){7}{\circle{5}}
			\multiput(40, 190)(30, 0){7}{\circle{5}}
			\multiput(40, 160)(30, 0){1}{\circle{5}}
			\multiput(40, 130)(30, 0){1}{\circle{5}}
			\multiput(40, 100)(30, 0){3}{\circle{5}}
			\multiput(40, 70)(30, 0){3}{\circle{5}}
			\multiput(40, 40)(30, 0){7}{\circle{5}}

			\multiput(190, 160)(30, 0){2}{\circle{5}}
			\multiput(190, 130)(30, 0){2}{\circle{5}}
			\multiput(190, 100)(30, 0){2}{\circle{5}}
			\multiput(190, 70)(30, 0){2}{\circle{5}}

			\multiput(65, 157)(30, 0){4}{$\times$}
			\multiput(95, 127)(30, 0){3}{$\times$}
			\put(65, 127){$\beta_1$}

			\multiput(125, 97)(30, 0){2}{$\times$}
			\multiput(155, 67)(30, 0){1}{$\times$}
			\put(125, 67){$\beta_2$}
			
			\put(60, 120){\line(1, 0){60}}
			\put(60, 120){\line(0, 1){50}}
			\put(60, 170){\line(1, 0){110}}
			\put(170, 170){\line(0, -1){110}}
			\put(170, 60){\line(-1, 0){50}}
			\put(120, 60){\line(0, 1){60}}

			\color{red}
			\multiput(10, 10)(0, 30){8}{\circle*{5}}
			\multiput(70, 10)(30, 0){6}{\circle*{5}}

			\put(35, 7){$z^{*}$}

		\end{picture}
	\end{center}
	\caption{
		Similar to Figure~\ref{figure:inexistent_pattern_1}, but one singularity block has two or more starting points.
		Such a phenomenon can never occur in the case of mappings.
	}
	\label{figure:inexistent_pattern_3}
\end{figure}

The fourth problem is, as mentioned in the above example, how to confirm that there is no possibility for $x_{t, n}(z)$ to become $0$ or $\infty$ except for the patterns $\text{I}_{t, n}$ and $\text{II}$.
For example, we must at least check that a spontaneous occurrence of $x_{t, n}(z^{*}) = \beta \notin \{ 0, \infty \}$ does not generate $0$ or $\infty$ as a value.
As illustrated in Figure~\ref{figure:inexistent_pattern_2}, a pattern starting with $\beta_1 \notin \{ 0, \infty \}$ and later involving $\beta_2 = 0$ would be problematic.
We must ensure that such a situation does not occur.
This is the most difficult part of our theory and we must develop some techniques to resolve this problem.

The only remedy known today for these problems would be the coprimeness.
This approach was developed as an algebraic reinterpretation of singularity confinement and can deal with global singularity structure \cite{coprimeness}.
For example, in the case of Hirota's discrete KdV equation, one can show that each iterate $x_{t, n}$ decomposes into
\[
	x_{t, n} = \rho_{t, n} \frac{\tau_{t, n} \tau_{t - 1, n - 1}}{\tau_{t - 1, n} \tau_{t, n - 1}},
\]
where $\tau_{t, n}$ is an irreducible polynomial with no monomial factor in the initial variables, $\rho_{t, n}$ is a monic Laurent monomial, and $\tau_{t, n}$ and $\tau_{s, m}$ are coprime unless $(t, n) = (s, m)$.
This decomposition, which corresponds to the pattern $\text{I}_{t, n}$, tells us everything we need to know about singularities.
In this illustrative example, $x_{t, n}(z)$ is considered to have a singularity at $z = z^{*}$ if $x_{t, n}(z^{*}) \in \{ 0, \infty \}$;
this convention applies only here, and the general definition will appear in Definition~\ref{definition:constant_singularity}.
For $(t, n) \ne (1, 1)$, the spontaneous occurrence of $x_{t, n}(z^{*}) = 0$ can be defined as $\tau_{t, n}(z^{*}) = 0$ ($z^{*} \ne 0, \infty$) and thus $Z_{t, n}$ is defined as $\deg_z \tau_{t, n}(z)$.
Since each $\tau_{s, m}(z)$ is coprime to $\tau_{t, n}(z)$ unless $(s, m) = (t, n)$, such $z^{*}$ does not generate any $0$ or $\infty$ other than the pattern $\text{I}_{t, n}$.
The patterns $\text{I}_{1, 1}$ and $\text{II}$ correspond to $z = 0$ and $\infty$, respectively.
Therefore, the decomposition guarantees that there is no possibility for $x_{t, n}(z^{*})$ to become $0$ or $\infty$ except for the patterns $\text{I}_{t, n}$ and $\text{II}$.

From the above discussion, it can be seen that showing the coprimeness for an equation immediately justifies the calculation by Halburd's method.
As of now, however, the idea of using the theory of coprimeness as a tool to calculate degrees for concrete equations has crucial flaws.

First, it requires too much computation to prove that an equation satisfies the coprimeness.
Apparently, the coprimeness property is related to Laurentification, which is a technique to transform an equation into one with the Laurent property by transforming dependent variables \cite{hone2007,coprimeness}.
Although some researchers (the present author included) invented several techniques to show the coprimeness, no general method is known as of now.
Therefore, in previous studies, some part of the discussion almost always relies on the concrete form of each equation \cite{coprimeness,exhv,factorize,exkdv}.
In addition, as seen in the above example, the theory of coprimeness is most compatible with singularity analysis when the singular values are $0$ and $\infty$.
If singularities take other values, we need multiple tau functions to express the original dependent variables, which makes Laurentification much more complicated \cite{hamad_kamp,laurentification_qrt,extended_laurent}.

In this paper, we do not use the theory of coprimeness to justify our calculations.
Our approach is based on the fact that an initial value problem for a lattice equation has infinitely many initial variables.
For example, consider the same domain as in Example~\ref{example:first_example_2d} and suppose that $x_{t, n}(z^{*}) = \alpha$ is a singularity.
Then, on which initial variables does $z^{*}$ depend?
Clearly, $z^{*}$ does not depend on $x_{t + 1, 0}$ or $x_{0, n + 1}$ since $x_{t, n}(z)$ is independent of these initial variables.
This independence, which is not valid in the case of mappings, enables us to calculate all singularity patterns from basic computations such as \eqref{equation:example2_sc_meaning}.

From here on, to clarify which $z^{*}$ depends on which initial variables, we stop substituting generic values for the initial variables other than $z$.
The rigorous definition is given in Definition~\ref{definition:field_k}.

\section{How to calculate degrees for lattice equations}\label{section:main}

In this section, we state our main theorems (Theorems~\ref{theorem:first_singularity}, \ref{theorem:main_theorem}, \ref{theorem:main_degree} and \ref{theorem:degree_solvable}) and give some propositions, which is helpful to detect the singularity structure of an equation.
All proofs will be given in the next section.

First, we define the past and future light cones.

\begin{definition}[Past and Future Light Cones]
	\begin{enumerate}
		\item
		For $(t, n) \in \mathbb{Z}^2$, we define the past light cone emanating from $(t, n)$ as
		\[
			\mathbb{Z}^2_{\le (t, n)} = \{ (t', n') \in \mathbb{Z}^2 \mid t' \le t; \ n' \le n \}.
		\]

		\item
		For $(t, n) \in \mathbb{Z}^2$, we define the future light cone emanating from $(t, n)$ as
		\[
			\mathbb{Z}^2_{\ge (t, n)} = \{ (t', n') \in \mathbb{Z}^2 \mid t' \ge t; \ n' \ge n \}.
		\]

	\end{enumerate}
\end{definition}

Next, we define the conditions on a domain we consider in this paper.

\begin{definition}\label{definition:domain}
	\begin{enumerate}
		\item 
		A nonempty set $H \subset \mathbb{Z}^2$ is called a domain if the following two conditions are satisfied:
		\begin{itemize}
			\item
			For any $(t, n) \in H$, the intersection between the corresponding past light cone and the domain $H$, i.e.,
			\[
				H_{\le (t, n)} := \mathbb{Z}^2_{\le (t, n)} \cap H = \{ (t', n') \in H \mid t' \le t; \ n' \le n \},
			\]
			is a finite set.

			\item
			For any $(t, n) \in H$, the corresponding future light cone
			\[
				\mathbb{Z}^2_{\ge (t, n)} = \{ (t', n') \in \mathbb{Z}^2 \mid t' \ge t; \ n' \ge n \}
			\]
			is contained in $H$.
			In particular, $H_{\ge (t, n)} := \mathbb{Z}^2_{\ge (t, n)} \cap H$ coincides with $\mathbb{Z}^2_{\ge (t, n)}$.

		\end{itemize}

		\item
		For a domain $H$, we define its initial boundary $H_0 \subset H$ as
		\[
			H_0 = \{ (t, n) \in H \mid (t - 1, n - 1) \notin H \}.
		\]

	\end{enumerate}
\end{definition}

\begin{remark}
	One can check that $(t, n) \in H$ does not belong to the initial boundary if and only if $(t - 1, n)$, $(t, n - 1)$ and $(t - 1, n - 1)$ all belong to $H$.
\end{remark}

The reason why we impose the above conditions when considering degree growth is discussed in \cite{domain}.
Domains satisfying these conditions were termed \emph{light-cone regular domains} in \cite{laurent_domain}.

From here on, we fix a quad equation \eqref{equation:quad_equation}, a domain $H \subset \mathbb{Z}$ and one of the initial variables $z = x_{t_0, n_0}$.
We think of each $x_{t, n}$ as a rational function in $z$, i.e., $x_{t, n} = x_{t, n}(z)$ and consider the degree with respect to $z$.
Note that among infinitely many choices of $H$ (a domain) and $z$ (one of the initial variables), only four of them are essential, in the sense that for any choice of $H$ and $z$, the degree computation can be reduced to one of these four cases.
We will leave such discussion in Appendix~\ref{appendix:domain} (Theorem~\ref{theorem:domain_choice}) as it is only of secondary importance in this paper.

\begin{definition}[$\mathbb{K}$]\label{definition:field_k}
	We define a field $\mathbb{K}$ as
	\[
		\mathbb{K} = \mathbb{C} \left( x_{t, n} \mid (t, n) \in H_0; \ x_{t, n} \ne z \right).
	\]
	That is, $\mathbb{K}$ is the field of rational functions in the initial variables other than $z$.
\end{definition}

In the examples in \textsection\ref{section:introduction} and \textsection\ref{section:introduction_calculation}, the initial variables other than $z$ are all taken generic.
In our strategy, however, we do not use this approach.
We think of each initial value other than $z$ just as a variable.
Therefore, $x_{t, n}(z)$ is an element of $\mathbb{K}(z)$.
Instead of taking generic initial values, we will substitute a $\mathbb{P}^1 \left( \overline{\mathbb{K}} \right)$-value $z^{*}$ for $z$, where $\overline{\mathbb{K}}$ is the algebraic closure of $\mathbb{K}$.

As seen in the rest of the paper, this approach makes all definitions and discussion much clearer.
A key idea in this paper is to clarify on which initial variables a value $z^{*}$ we substitute for $z$ depends.
Some definitions, lemmas and propositions entirely rely on this approach.
The definition of a singularity we use in this paper, which is given below, is such an example.

\begin{definition}[Constant Singularity, Constant Singularity Pattern]\label{definition:constant_singularity}
	For $(t, n) \in H$ and $z^{*} \in \mathbb{P}^1 \left( \overline{\mathbb{K}} \right)$, we say that $x_{t, n}(z)$ has a \emph{constant singularity} at $z = z^{*}$ if $x_{t, n}(z^{*}) \in \mathbb{P}^1(\mathbb{C})$.
	The \emph{multiplicity} of a constant singularity $x_{t, n}(z^{*}) = x^{*}$ ($\in \mathbb{P}^1(\mathbb{C})$) is the multiplicity of $x_{t, n}(z) = x^{*}$ at $z = z^{*}$ as a rational function in $z$.
	We define the multiplicity of $x_{t, n}(z^{*}) = x^{*}$ as $0$ if $x_{t, n}(z^{*}) \ne x^{*}$ (including the case where $x_{t, n}(z)$ does not have a constant singularity at $z = z^{*}$).
	Here, a constant singularity $x_{t, n}(z^{*}) = x^{*}$ consists of its position $(t, n)$, the corresponding value $x^{*}$, and its multiplicity.
	As in Example~\ref{example:first_example_1d}, when we are interested only in the value of a constant singularity, we use the term ``singular value.''

	For $z^{*} \in \mathbb{P}^1 \left( \overline{\mathbb{K}} \right)$, the \emph{constant singularity pattern} corresponding to $z = z^{*}$ consists of the following data:
	\begin{itemize}
		\item 
		the set of the points $(t, n) \in H$ where $x_{t, n}(z)$ has a constant singularity at $z = z^{*}$,

		\item
		their values, i.e., $x_{t, n}(z^{*}) \in \mathbb{P}^1(\mathbb{C})$ for each point $(t, n) \in H$ where $x_{t, n}(z^{*})$ has a constant singularity,

		\item
		their multiplicities.

	\end{itemize}
	However, we do not use the term ``constant singularity pattern'' if it contains no constant singularity.

	We say that a constant singularity pattern corresponding to $z = z^{*}$ is \emph{fixed} (resp.\ \emph{movable}) if $z^{*} \in \mathbb{P}^1(\mathbb{C})$ (resp.\ $z^{*} \in \mathbb{P}^1 \left( \overline{\mathbb{K}} \right) \setminus \mathbb{P}^1(\mathbb{C})$).
	A constant singularity pattern is said to be \emph{confining} (resp.\ \emph{non-confining}) if it contains only finitely many (resp.\ infinitely many) constant singularities.
	We say that a constant singularity pattern is \emph{solitary} if it consists of only one constant singularity.
\end{definition}

Since all the singularity patterns we consider in the remaining part of this paper are constant singularity patterns, we sometimes use a singularity pattern (or simply a pattern) to denote a constant singularity pattern.
We will discuss non-constant singularity patterns in \textsection\ref{section:conclusion}.

\begin{remark}
	In the context of singularity confinement, a singularity is defined by loss of information on the initial values, as seen in Example~\ref{example:first_example_1d}.
	However, the definition of a constant singularity does not involve any loss of information on the initial values, although such a loss may occur in significant examples.
	If we were to require such a loss as part of the definition, it would be necessary to define what loss of information means for lattice equations.
	In fact, at the stage where the first singularity appears, no loss of information on the initial values seems to occur.
	This observation will later be confirmed in the proof of Theorem~\ref{theorem:main_theorem}, although the theorem itself serves a different purpose.
	As this paper does not aim to establish an ideal definition of singularity confinement for lattice equations, we will not pursue this discussion further.
\end{remark}

\begin{definition}[Basic Pattern]\label{definition:basic_pattern}
	Let
	\[
		H^{\text{basic}}
		= \mathbb{Z}^2_{\ge (- 1, - 1)} \setminus \{ (- 1, - 1) \}
	\]
	and consider the equation
	\[
		x^{\text{basic}}_{t, n} = \Phi \left( x^{\text{basic}}_{t - 1, n}, x^{\text{basic}}_{t, n - 1}, x^{\text{basic}}_{t - 1, n - 1} \right)
	\]
	on $H^{\text{basic}}$
	(see Figure~\ref{figure:basic_domain}).
	For each $(t, n) \in H^{\text{basic}}$, think of $x^{\text{basic}}_{t, n}$ as a rational function in $x^{\text{basic}}_{0, 0}$,
	i.e., $x^{\text{basic}}_{t, n} = x^{\text{basic}}_{t, n}(x^{\text{basic}}_{0, 0})$.
	For $x^{*} \in \mathbb{P}^1(\mathbb{C})$, the \emph{basic pattern} corresponding to $x^{\text{basic}}_{0, 0} = x^{*}$ is the fixed constant singularity pattern corresponding to $x^{\text{basic}}_{0, 0} = x^{*}$ on $H^{\text{basic}}$.
\end{definition}

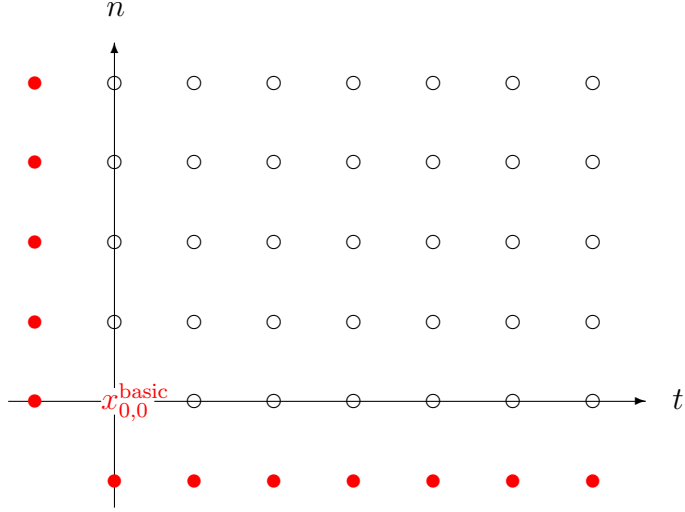
\begin{figure}
	\begin{center}
		\begin{picture}(240, 200)

			\put(55, 40){\vector(1, 0){185}}
			\put(40, 45){\vector(0, 1){130}}

			\put(0, 40){\line(1, 0){35}}
			\put(40, 0){\line(0, 1){35}}

			\multiput(40, 160)(30, 0){7}{\circle{5}}
			\multiput(40, 130)(30, 0){7}{\circle{5}}
			\multiput(40, 100)(30, 0){7}{\circle{5}}
			\multiput(40, 70)(30, 0){7}{\circle{5}}
			\multiput(70, 40)(30, 0){6}{\circle{5}}

			\put(250, 37){$t$}
			\put(37, 185){$n$}
			
			\color{red}
			\multiput(10, 40)(0, 30){5}{\circle*{5}}
			\multiput(40, 10)(30, 0){7}{\circle*{5}}

			\put(35, 37){$x^{\text{basic}}_{0, 0}$}

		\end{picture}
	\end{center}
	\caption{Domain on which we consider basic patterns.}
	\label{figure:basic_domain}
\end{figure}

\begin{remark}
	Note that for generic $x^{*} \in \mathbb{P}^1(\mathbb{C})$, the corresponding basic pattern is solitary.
	That is, there are only finitely many values of $x^{*}$ that produce a non-solitary pattern.
	This fact will be justified later (see Proposition~\ref{proposition:solitary}).
\end{remark}

Unfortunately, our strategy is not always valid.
We introduce two conditions on a quad equation.

\begin{definition}[$\partial$-factor Condition]\label{definition:d_factor_condition}
	We say that a quad equation $\Phi$ satisfies the \emph{$\partial$-factor condition} if it has the following two properties:
	\begin{itemize}
		\item
		Consider the factorizations of the numerators of $\partial_B \Phi$ and $\partial_B \frac{1}{\Phi}$, where $\partial_X = \frac{\partial}{\partial X}$ denotes the (partial) derivative with respect to a variable $X$.
		Let $F = F(B, C, D) \in \mathbb{C}[B, C, D]$ be an arbitrary irreducible polynomial occurring in one of these factorizations.
		Then, $F$ satisfies one of the following:
		\begin{itemize}
			\item
			$\partial_B F \ne 0$,
	
			\item
			$F = C - \alpha$ or $F = D - \alpha$ for some $\alpha \in \mathbb{C}$.
	
		\end{itemize}
		That is, the numerators of $\partial_B \Phi$ and $\partial_B \frac{1}{\Phi}$ do not have an irreducible factor that depends on both $C$ and $D$ but not on $B$.

		\item
		$\Phi$ satisfies the above condition with $B$ and $C$ interchanged.

	\end{itemize}
\end{definition}

\begin{definition}
	Let $x^{*} \in \mathbb{P}^1(\mathbb{C})$ and let
	\[
		\Phi(B, C, D) = \frac{\Phi_1(B, C, D)}{\Phi_2(B, C, D)}
	\]
	where $\Phi_1, \Phi_2 \in \mathbb{C}[B, C, D]$ are coprime.
	If $x^{*} \ne \infty$ (resp.\ $x^{*} = \infty$), an irreducible factor of the relation $\Phi(B, C, D) = x^{*}$ is an irreducible factor of the polynomial
	\[
		\Phi_1(B, C, D) - x^{*} \Phi_2(B, C, D)
	\]
	(resp.\ of the polynomial $\Phi_2(B, C, D)$).
\end{definition}

\begin{definition}[Basic Pattern Condition]\label{definition:basic_pattern_condition}
	Let $x^{*} \in \mathbb{P}^1(\mathbb{C})$ and consider the relation
	\[
		\Phi(B, C, D) = x^{*}.
	\]
	We say that the equation satisfies the \emph{basic pattern condition} for $x^{*}$ if any irreducible factor $F = F(B, C, D) \in \mathbb{C}[B, C, D]$ of the above relation satisfies one of the following:
	\begin{itemize}
		\item
		$\partial_D F \ne 0$,

		\item
		$F = B - \alpha$ or $F = C - \alpha$ for some $\alpha \in \mathbb{C}$.

	\end{itemize}
	That is, the equation does not satisfy the basic pattern condition for $x^{*}$ if the relation $\Phi(B, C, D) = x^{*}$ has an irreducible factor that contains both $B$ and $C$ but not $D$.
\end{definition}

The name ``basic pattern condition'' is derived from Theorem~\ref{theorem:main_theorem}.
If an equation does not satisfy this condition, a movable pattern does not necessarily coincide with the corresponding basic pattern (see Example~\ref{example:counterexample_2}).

\begin{remark}
	If the polynomial $\Phi_1(B, C, D) - x^{*} \Phi_2(B, C, D)$ is irreducible, then $x^{*}$ satisfies the basic pattern condition.
	In many practical cases, this irreducibility holds for generic $x^{*} \in \mathbb{C}$ since $\Phi_1$ and $\Phi_2$ are coprime.
	For instance, if $\deg_D \Phi(B, C, D) = 1$ (equivalent to the usual reversibility condition in discrete integrable systems, i.e., that the equation can be solved in the opposite direction by a rational function), then $\Phi_1(B, C, D) - x^{*} \Phi_2(B, C, D)$ is irreducible for generic $x^{*}$, as it is of degree $1$ in $D$.
	Similarly, if $\deg_B \Phi(B, C, D) = 1$ or $\deg_C \Phi(B, C, D) = 1$, the polynomial is irreducible for generic $x^{*} \in \mathbb{C}$.

	However, there are some cases where $\Phi_1(B, C, D) - x^{*} \Phi_2(B, C, D)$ is reducible for generic $x^{*}$.

	The first case occurs when 
	\[
		\operatorname{trdeg}_{\mathbb{C}} \mathbb{C}\left( \Phi_1(B, C, D), \Phi_2(B, C, D) \right) = 1,
	\]
	with $\operatorname{trdeg}_{K_1} K_2$ denoting the transcendental degree of a field extension $K_1 \subset K_2$.
	In this situation, there exists a polynomial $Y = Y(B, C, D)$ such that $\Phi_1$ and $\Phi_2$ can be regarded as polynomials in $Y$ only: $\Phi_1 = \Phi_1(Y)$, $\Phi_2 = \Phi_2(Y)$.
	If either $\Phi_1(Y)$ or $\Phi_2(Y)$ has degree greater than $1$, then the polynomial $\Phi_1(Y) - x^{*} \Phi_2(Y)$ factorizes for generic $x^{*}$.
	However, any irreducible factorization must take the following form:
	\[
		\Phi_1(Y) - x^{*} \Phi_2(Y) = a (Y - b_1) \cdots (Y - b_m).
	\]
	Since $\Phi$ depends on each of $B$, $C$, and $D$, so does $Y$.
	Therefore, the basic pattern condition still holds for generic $x^{*}$ in this case.

	The second case occurs when there exist $\Psi_1, \Psi_2 \in \mathbb{C}[B, C, D]$ and $p \in \mathbb{Z}_{\ge 2}$ such that
	\[
		\Phi_1 = \Psi^p_1, \quad
		\Phi_2 = \Psi^p_2.
	\]
	In this situation, the polynomial $\Phi_1(B, C, D) - x^{*} \Phi_2(B, C, D)$ trivially factorizes as
	\[
		\Phi_1 - x^{*} \Phi_2
		= \Psi^p_1 - x^{*} \Psi^p_2
		= \prod^{p - 1}_{j = 0} \left( \Psi_1 - \zeta^j_p \sqrt[p]{x^{*}} \Psi_2 \right),
	\]
	where $\zeta_p$ denotes a primitive $p$-th root of unity.
	Although each factor may further decompose, the analysis reduces to the case where $\Phi_1$ and $\Phi_2$ are replaced by $\Psi_1$ and $\Psi_2$, i.e., after removing the $p$-th power.
	Therefore, the basic pattern condition can be verified by considering the situation without such a $p$-th power structure.
\end{remark}

These observations lead us to the following conjecture.

\begin{conjecture}
	For any quad equation, there exist at most finitely many values in $\mathbb{P}^1(\mathbb{C})$ for which the basic pattern condition fails.
\end{conjecture}

Let us state our main theorems in this paper.
The first theorem guarantees that each constant singularity pattern has only one starting point.

\begin{theorem}\label{theorem:first_singularity}
	Suppose that a quad equation $\Phi$ satisfies the $\partial$-factor condition (Definition~\ref{definition:d_factor_condition}).
	Let $H \subset \mathbb{Z}^2$ be a domain and let $z$ be one of the initial variables.
	Let $z^{*} \in \mathbb{P}^1 \left( \overline{\mathbb{K}} \right)$ and consider the constant singularity pattern corresponding to $z = z^{*}$.
	Then, the set
	\[
		\{ (t', n') \in H \mid x_{t', n'}(z^{*}) \text{ is a constant singularity} \}
	\]
	has a minimum element with respect to the product order $\le$ on $\mathbb{Z}^2$.
	In particular, the equation has no undesirable pattern as in Figures~\ref{figure:inexistent_pattern_1} and \ref{figure:inexistent_pattern_3}.
\end{theorem}

Each basic pattern has only one starting point.
A pattern with two or more starting points, such as shown in Figures~\ref{figure:inexistent_pattern_1} or \ref{figure:inexistent_pattern_3}, is undesirable in our study because it does not correspond to a basic pattern.

\begin{definition}[First Singularity, Starting Value]\label{definition:first_singularity}
	If $(t, n)$ is the minimum element of the set in Theorem~\ref{theorem:first_singularity}, we say that $x_{t, n}(z^{*})$ is the \emph{first singularity} (or first constant singularity) of the pattern.
	We call the value of the first singularity of a pattern the \emph{starting value} of the pattern.
\end{definition}

The second theorem ascertains that the global structure of a movable pattern coincides with that of the corresponding basic pattern.

\begin{theorem}\label{theorem:main_theorem}
	Suppose that a quad equation $\Phi$ satisfies the $\partial$-factor condition.
	Let $H \subset \mathbb{Z}^2$ be a domain and let $z$ be one of the initial variables.
	Let $z^{*} \in \mathbb{P}^1 \left( \overline{\mathbb{K}} \right)$ and consider the constant singularity pattern corresponding to $z = z^{*}$.
	Let $x_{t, n}(z^{*}) = x^{*} \in \mathbb{P}^1(\mathbb{C})$ be the first constant singularity of this pattern and let $r$ be the multiplicity of $x_{t, n}(z^{*}) = x^{*}$.
	Suppose that the equation satisfies the basic pattern condition for $x^{*}$ (Definition~\ref{definition:basic_pattern_condition}).
	Then, in the region $H_{\ge (t, n)}$, the constant singularity pattern corresponding to $z = z^{*}$ coincides with the $(t, n)$-translation of the basic pattern corresponding to $x^{\text{basic}}_{0, 0} = x^{*}$ with all the multiplicities multiplied by $r$.
	That is, if $x^{\text{basic}}_{s, m}(x^{*}) = \alpha$ is a constant singularity of multiplicity $\ell$ on $H^{\text{basic}}$, then $x_{t + s, n + m}(z^{*}) = \alpha$ is a constant singularity of multiplicity $r \ell$ on $H$.
\end{theorem}

The third theorem guarantees that given all the fixed and basic patterns we need, the degree counting in Halburd's method is exact.

\begin{theorem}\label{theorem:main_degree}
	Suppose that a quad equation $\Phi$ satisfies the $\partial$-factor condition.
	Let $H \subset \mathbb{Z}^2$ be a domain and let $z$ be one of the initial variables.
	\begin{itemize}
		\item
		Let
		$S = \{ \alpha_1, \ldots, \alpha_J \} \subset \mathbb{P}^1(\mathbb{C})$,
		$S \ne \emptyset$.
		We focus on these singular values when calculating degrees.

		\item
		Assume that we know all the singularity patterns that contain a singular value $\alpha_j$, and assume in addition that the set of starting values of these patterns is finite.
		Let $\beta_1, \ldots, \beta_I \in \mathbb{P}^1(\mathbb{C})$
		(resp.\ $\gamma_1, \ldots, \gamma_L \in \mathbb{P}^1(\mathbb{C})$)
		be the starting values of the movable (resp.\ fixed) patterns we consider here.
	
		\item
		Suppose that the equation satisfies the basic pattern condition for $\beta_1, \ldots, \beta_I$.

		\item
		Consider the basic pattern corresponding to $x^{\text{basic}}_{0, 0} = \beta_i$ on $H^{\text{basic}}$.
		Let
		\[
			\operatorname{mult}^{\text{basic}, \beta_i}_{s, m}(\alpha_j)
		\]
		be the multiplicity of $x^{\text{basic}}_{s, m}(\beta_i) = \alpha_j$ in this pattern.
		The multiplicity $\operatorname{mult}^{\text{basic}, \beta_i}_{s, m}(\alpha_j)$ is considered to be $0$ if $x^{\text{basic}}_{s, m}(\beta_i) \ne \alpha_j$ or $(s, m) \notin H^{\text{basic}}$.
	
		\item
		Consider the fixed pattern corresponding to $z = \gamma_{\ell}$ on $H$.
		Let
		\[
			\operatorname{mult}^{\gamma_{\ell}}_{t, n}(\alpha_j)
		\]
		be the multiplicity of $x_{t, n}(\gamma_{\ell}) = \alpha_j$ in this pattern.
		The multiplicity is considered to be $0$ if $x_{t, n}(\gamma_{\ell}) \ne \alpha_j$.
	
	\end{itemize}
	Then, for each $j$, we have the following degree relation:
	\[
		\deg_z x_{t, n}(z)
		= \sum^I_{i = 1} \sum_{(s, m) \in H} \operatorname{mult}^{\text{basic}, \beta_i}_{t - s, n - m}(\alpha_j) Z_{s, m}(\beta_i)
		+ \sum^L_{\ell = 1} \operatorname{mult}^{\gamma_{\ell}}_{t, n}(\alpha_j),
	\]
	where $Z_{s, m}(\beta_i)$ is ``the number of spontaneous occurrences of $x_{s, m}(z) = \beta_i$ with multiplicity,''
	whose rigorous definition is given in Definition~\ref{definition:degree_divisor}.
	Since the left-hand side is independent of $j$, so is the right-hand side.
\end{theorem}

\begin{remark}
	Although $Z_{s, m}(\beta_i)$ is used in Theorems~\ref{theorem:main_degree} and \ref{theorem:degree_solvable}, its definition is not given in this section.
	To calculate degrees with our theorems, however, the definition of $Z_{s, m}(\beta_i)$ is not relevant.
	What is important here is the fact that, for each $s, m, i$, $Z_{s, m}(\beta_i)$ is a nonnegative integer defined in some rigorous way.
\end{remark}

\begin{remark}
	If $t - s < 0$ or $n - m < 0$, then $\operatorname{mult}^{\text{basic}, \beta_i}_{t - s, n - m}(\alpha_j)$ is always $0$.
	Therefore, the $\sum_{(s, m) \in H}$ in the degree relation can be replaced with $\sum_{(s, m) \in H_{\le (t, n)}}$, which is always a finite sum.
\end{remark}

The fourth theorem tells us whether or not we can calculate degrees from degree relations.

\begin{theorem}\label{theorem:degree_solvable}
	Under the same assumptions and notations as in Theorem~\ref{theorem:main_degree}, we have the following:
	\begin{enumerate}
		\item
		If $I \ge J$, i.e., the number of singular values is equal to or greater than that of starting values, then it is impossible to compute degrees only from the degree relations.

		\item
		If $\{ \beta_1, \ldots, \beta_I \} \subset S$ and $I < J$, then the degree relations determine all the degrees.
		That is, by solving a system of linear equations, one can compute degrees only from the degree relations.

	\end{enumerate}
\end{theorem}

\begin{remark}
	If the starting value of a movable pattern is $\beta$, then the pattern contains $\beta$ as a value.
	Therefore, $\alpha \in S \setminus \{ \beta_1, \ldots, \beta_I \}$ implies that the equation does not have a movable pattern starting with $\alpha$.
	In short, to calculate degrees in our method, an equation needs to have a $\mathbb{P}^1(\mathbb{C})$-value that does not appear as the starting value of any movable pattern.
\end{remark}

The following propositions, which are taken from \textsection\ref{section:proof}, are often useful in analyzing the singularity structure of a concrete equation.

\begin{proposition}\label{proposition:key_lemma}
	Suppose that $\Phi$ satisfies the $\partial$-factor condition.
	Let $(t, n) \in H$ and suppose that $z^{*} \in \mathbb{P}^1 \left( \overline{\mathbb{K}} \right)$ generates a constant singularity at $(t, n)$.
	If $(s, m) \in H$ satisfies $s > t$ and
	\[
		x_{s - 1, m}(z^{*}), x_{s - 1, m - 1}(z^{*}), x_{s - 1, m - 2}(z^{*}), \cdots
	\]
	are not constant singularities, then
	\[
		x_{s, m}(z^{*}), x_{s, m - 1}(z^{*}), x_{s, m - 2}(z^{*}), \cdots
	\]
	all depend on the initial variable $w = x_{s, m - m'}$ where $m'$ is the least nonnegative integer satisfying $(s, m - m') \in H_0$.
	In particular, they are not constant singularities, either (see Figure~\ref{figure:proposition_key_lemma}).

	The statement with the $t$- and $n$-axes interchanged holds, too.
\end{proposition}

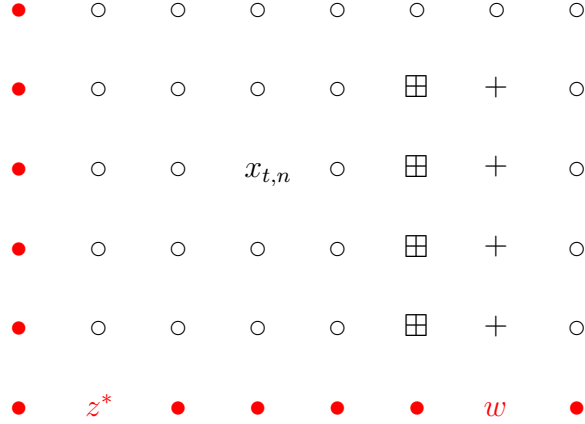
\begin{figure}
	\begin{center}
		\begin{picture}(240, 200)

			\multiput(40, 160)(30, 0){7}{\circle{5}}
			\multiput(40, 130)(30, 0){4}{\circle{5}}
			\multiput(40, 100)(30, 0){2}{\circle{5}}
			\multiput(40, 70)(30, 0){4}{\circle{5}}
			\multiput(40, 40)(30, 0){4}{\circle{5}}

			\multiput(130, 100)(30, 0){1}{\circle{5}}

			\put(95, 97){$x_{t, n}$}

			\multiput(155, 37)(0, 30){4}{$\boxplus$}
			\multiput(185, 38)(0, 30){4}{$+$}
			\multiput(220, 40)(0, 30){4}{\circle{5}}

			\color{red}
			\multiput(10, 10)(0, 30){6}{\circle*{5}}
			\multiput(70, 10)(30, 0){4}{\circle*{5}}
			\multiput(220, 10)(30, 0){1}{\circle*{5}}

			\put(35, 7){$z^{*}$}

			\put(185, 7){$w$}

		\end{picture}
	\end{center}
	\caption{
		Situation in Proposition~\ref{proposition:key_lemma}.
		If the points marked with ``$\boxplus$'' are not constant singularities, then those marked with ``$+$'' depend on $w$ and are not constant singularities, either.
		This procedure can easily be repeated and thus none of the points east to but not higher than the $\boxplus$-wall is a constant singularity. 
		To prove this proposition, it is essential that the points marked with ``$+$'' lie outside the past light cone emanating from $(t, n)$.
	}
	\label{figure:proposition_key_lemma}
\end{figure}

\begin{proposition}\label{proposition:confining}
	Suppose that $\Phi$ satisfies the $\partial$-factor condition.
	Let $x_{t, n}(z^{*})$ be the first constant singularity of a singularity pattern and let $(s, m) \in H_{\ge (t, n)}$.
	If
	\[
		x_{t, m + 1}(z^{*}), x_{t + 1, m + 1}(z^{*}), \cdots, x_{s, m + 1}(z^{*}), x_{s + 1, m + 1}(z^{*}), x_{s + 1, m}(z^{*}), \ldots, x_{s + 1, n + 1}(z^{*}), x_{s + 1, n}(z^{*})
	\]
	are not constant singularities, then all the constant singularities of this pattern lie in the region
	\[
		H_{\ge (t, n)} \cap H_{\le (s, m)}
	\]
	(see Figure~\ref{figure:confining_pattern}).
	In particular, this constant singularity pattern is confining.
	Therefore, a quad equation with the $\partial$-factor condition has no undesirable pattern as in Figure~\ref{figure:inexistent_pattern_2}.
\end{proposition}
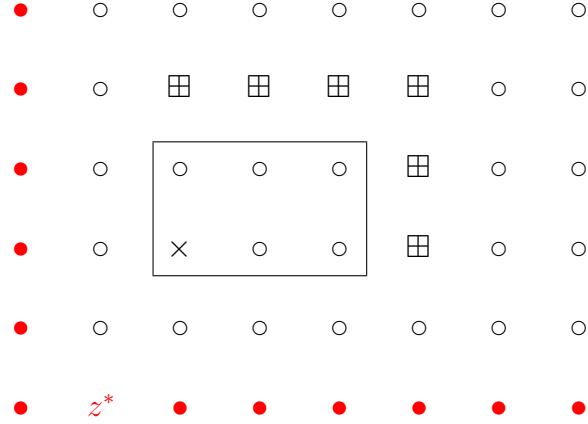
\begin{figure}
	\begin{center}
		\begin{picture}(240, 200)

			\multiput(70, 160)(30, 0){6}{\circle{5}}
			\multiput(190, 130)(30, 0){2}{\circle{5}}
			\multiput(190, 100)(30, 0){2}{\circle{5}}
			\multiput(190, 70)(30, 0){2}{\circle{5}}
			\multiput(70, 40)(30, 0){6}{\circle{5}}
			\multiput(40, 40)(0, 30){5}{\circle{5}}

			\multiput(65, 127)(30, 0){4}{$\boxplus$}
			\multiput(155, 67)(0, 30){2}{$\boxplus$}
			
			\put(65, 67){$\times$}
			\multiput(100, 70)(30, 0){2}{\circle{5}}
			\multiput(70, 100)(30, 0){3}{\circle{5}}
			\put(60, 60){\framebox(80, 50)}

			\color{red}
			\multiput(10, 10)(0, 30){6}{\circle*{5}}
			\multiput(70, 10)(30, 0){6}{\circle*{5}}

			\put(35, 7){$z^{*}$}

		\end{picture}
	\end{center}
	\caption{
		Situation in Proposition~\ref{proposition:confining}.
		The point marked with ``$\times$'' is the first constant singularity.
		If none of the points marked with ``$\boxplus$'' has a constant singularity, then all the constant singularities of this pattern lie in the boxed region.
		If we consider a basic pattern, the assumption of the proposition coincides with what we check when we perform singularity confinement as in Example~\ref{example:first_example_2d}.
	}
	\label{figure:confining_pattern}
\end{figure}

The following proposition allows us to find all the candidates for the starting values that can generate a non-solitary pattern.

\begin{proposition}\label{proposition:solitary}
	Suppose that the equation satisfies the $\partial$-factor condition and let $x^{*} \in \mathbb{P}^1(\mathbb{C})$.
	If none of 
	\[
		\Phi(x^{*}, C, D), \quad
		\Phi(B, x^{*}, D), \quad
		\Phi(B, C, x^{*})
	\]
	belongs to $\mathbb{P}^1(\mathbb{C})$, then any singularity pattern starting from the value $x^{*}$ is solitary.
\end{proposition}

\begin{remark}
	Given a quad equation $\Phi$, one can find all $x^{*}$ that could generate non-solitary patterns by computing derivatives of $\Phi$ (and of $\frac{1}{\Phi}$ if necessary).
\end{remark}

The following proposition, used in conjunction with Proposition~\ref{proposition:solitary}, enables us to find all the candidates for the $\mathbb{P}^1(\mathbb{C})$-values that do not appear as the starting value of any movable pattern.

\begin{proposition}\label{proposition:find_hidden_value}
	Let $\alpha \in \mathbb{P}^1(\mathbb{C})$ and suppose that $\alpha$ does not appear as the starting value of any movable pattern.
	Let $(t, n) \in H \setminus H_0$.
	If $x_{t, n}(z)$ depends on $z$, then the value $\alpha$ must appear at $(t, n)$ in some non-solitary pattern.
\end{proposition}

The following proposition determines whether each candidate value obtained by Propositions~~\ref{proposition:solitary} and \ref{proposition:find_hidden_value} appears as the starting value of some movable pattern.

\begin{proposition}\label{proposition:check_hidden_value}
	Suppose that the equation satisfies the $\partial$-factor condition and let $x^{*} \in \mathbb{P}^1(\mathbb{C})$.
	Then, the value $x^{*}$ appears at the starting point of some movable pattern if and only if the relation 
	\[
		\Phi(B, C, D) = x^{*}
	\]
	has an irreducible factor other than
	\[
		B - \alpha, \quad
		C - \alpha, \quad
		D - \alpha \quad
		(\alpha \in \mathbb{C}).
	\]
\end{proposition}

\section{Proofs}\label{section:proof}

In this section, we give proofs of our main theorems.
The lemmas the author thinks might be useful to clarify the singularity structure for a concrete equation are called propositions and are included in \textsection\ref{section:main}.

\begin{lemma}\label{lemma:main_lemma}
	Let $(t, n) \in H$ and suppose that $(t - 1, n)$ and $(t, n - 1)$ both belong to $H$.
	Let $x_{t, n}(z^{*}) = x^{*} \in \mathbb{P}^1(\mathbb{C})$ be a constant singularity of multiplicity $r \in \mathbb{Z}_{> 0}$ and suppose that $x_{t', n'}(z^{*})$ is not a constant singularity for any $(t', n') \in H_{\le (t, n)} \setminus \{ (t, n) \}$.
	If
	\[
		x_{t - 1, n}(z^{*}), x_{t - 1, n + 1}(z^{*}), x_{t - 1, n + 2}(z^{*}), \cdots \quad
	\]
	and
	\[
		x_{t, n - 1}(z^{*}), x_{t + 1, n - 1}(z^{*}), x_{t + 2, n - 1}(z^{*}), \cdots
	\]
	are algebraically independent over $\mathbb{C}$, then the constant singularity pattern corresponding to $z = z^{*}$ coincides with the $(t, n)$-translation of the basic pattern corresponding to $x^{\text{basic}}_{0, 0} = x^{*}$ with all the multiplicities multiplied by $r$.
\end{lemma}
\begin{proof}
	As in Figure~\ref{figure:proposition_main_proposition}, the key idea to the proof is to apply Lemma~\ref{lemma:expansion_lemma}.

	Let $(s, m) \in H^{\text{basic}} \setminus H^{\text{basic}}_0$ and let $\ell$ be the multiplicity of the constant singularity $x^{\text{basic}}_{s, m}(x^{*}) = \alpha$ of the basic pattern (if it is not a constant singularity, we think of $\ell$ as $0$).
	We show that the multiplicity of the constant singularity $x_{t + s, n + m}(z^{*}) = \alpha$ is $r \ell$.
	We omit the proof for the case $x^{*} = \infty$ or $\alpha = \infty$, since in this case only the notation becomes complicated.

	Let
	\[
		H^{\text{basic}}_{\le (s, m)} \cap H^{\text{basic}}_0 = \{ (0, 0), (s_1, m_1), \ldots, (s_N, m_N) \}
	\]
	and let
	\[
		X = x^{\text{basic}}_{0, 0}, Y_1 = x^{\text{basic}}_{s_1, m_1}, \ldots, Y_N = x^{\text{basic}}_{s_N, m_N}.
	\]
	Since $x^{\text{basic}}_{s, m}$ is a rational function in
	$X, Y_1, \ldots, Y_N$ over $\mathbb{C}$,
	there exists
	\[
		F = F(X; Y_1, \ldots, Y_N) \in \mathbb{C}(X; Y_1, \ldots, Y_N)
	\]
	such that
	\[
		x^{\text{basic}}_{s, m} = \alpha + F(X; Y_1, \ldots, Y_N).
	\]
	Let $\varepsilon$ be an infinitesimal parameter.
	Since $x^{\text{basic}}_{s, m}(x^{*}) = \alpha$ is a constant singularity of multiplicity $\ell$,
	the expansion of $F(x^{*} + \varepsilon; Y_1, \ldots, Y_N)$ is
	\[
		F(x^{*} + \varepsilon; Y_1, \ldots, Y_N) = F_{\ell}(Y_1, \ldots, Y_N) \varepsilon^{\ell} + O \left( \varepsilon^{\ell + 1} \right),
	\]
	where $F_{\ell}(Y_1, \ldots, Y_N)$ is a non-zero rational function.

	Using $F$, we can express $x_{t + s, n + m}(z)$ by
	$x_{t, n}(z), x_{t + s_1, n + m_1}(z), \ldots, x_{t + s_N, n + m_N}(z)$
	as
	\[
		x_{t + s, n + m}(z) = \alpha + F(x_{t, n}(z); x_{t + s_1, n + m_1}(z), \ldots, x_{t + s_N, n + m_N}(z)).
	\]
	Since
	\[
		x_{t, n}(z^{*} + \varepsilon) = x^{*} + a \varepsilon^r + O \left( \varepsilon^{r + 1} \right)
	\]
	for some $a \in \overline{\mathbb{K}} \setminus \{ 0 \}$
	and $x_{t + s_1, n + m_1}(z^{*}), \ldots, x_{t + s_N, n + m_N}(z^{*})$ are algebraically independent over $\mathbb{C}$,
	it follows from Lemma~\ref{lemma:expansion_lemma} that
	\[
		F(x_{t, n}(z^{*} + \varepsilon); x_{t + s_1, n + m_1}(z^{*} + \varepsilon), \ldots, x_{t + s_N, n + m_N}(z^{*} + \varepsilon)) = b \varepsilon^{r \ell} + O \left( \varepsilon^{r \ell + 1} \right)
	\]
	for some $b \in \overline{\mathbb{K}} \setminus \{ 0 \}$, which completes the proof.
\end{proof}

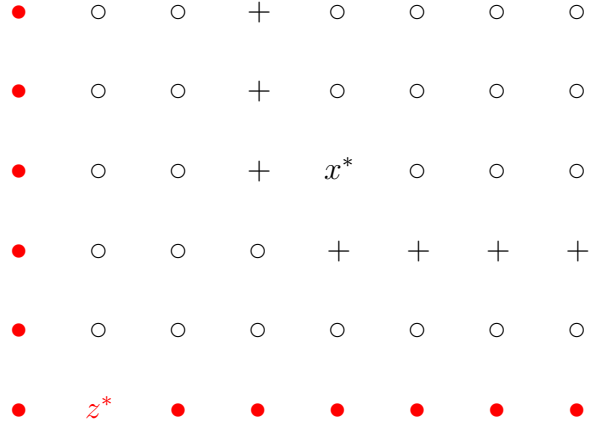
\begin{figure}
	\begin{center}
		\begin{picture}(240, 200)

			\multiput(40, 160)(30, 0){2}{\circle{5}}
			\multiput(40, 130)(30, 0){2}{\circle{5}}
			\multiput(40, 100)(30, 0){2}{\circle{5}}
			\multiput(40, 70)(30, 0){3}{\circle{5}}
			\multiput(40, 40)(30, 0){7}{\circle{5}}

			\multiput(130, 160)(30, 0){4}{\circle{5}}
			\multiput(130, 130)(30, 0){4}{\circle{5}}
			\multiput(160, 100)(30, 0){3}{\circle{5}}

			\multiput(96, 97)(0, 30){3}{$+$}
			\multiput(126, 67)(30, 0){4}{$+$}

			\put(125, 97){$x^{*}$}

			\color{red}
			\multiput(10, 10)(0, 30){6}{\circle*{5}}
			\multiput(70, 10)(30, 0){6}{\circle*{5}}

			\put(35, 7){$z^{*}$}

		\end{picture}
	\end{center}
	\caption{Situation in Lemma~\ref{lemma:main_lemma}.
	The points marked with ``$+$'' can be thought of as if they were initial variables for a basic pattern since those points are algebraically independent over $\mathbb{C}$.}
	\label{figure:proposition_main_proposition}
\end{figure}

\begin{definition}[$\mathbb{K}_{t, n}$]\label{definition:k_tn}
	For $(t, n) \in H$, we define a subfield $\mathbb{K}_{t, n} \subset \mathbb{K}$ as
	\[
		\mathbb{K}_{t, n} = \mathbb{C} \left( x_{t', n'} \mid (t', n') \in H_0; \ (t', n') \le (t, n); \ x_{t', n'} \ne z \right)
	\]
	(see Figure~\ref{figure:definition_ktn}).
\end{definition}

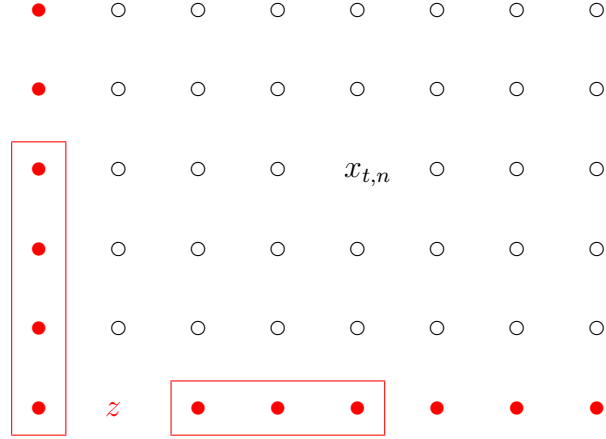
\begin{figure}
	\begin{center}
		\begin{picture}(240, 200)

			\multiput(40, 160)(30, 0){7}{\circle{5}}
			\multiput(40, 130)(30, 0){7}{\circle{5}}
			\multiput(40, 100)(30, 0){3}{\circle{5}}
			\multiput(40, 70)(30, 0){7}{\circle{5}}
			\multiput(40, 40)(30, 0){7}{\circle{5}}

			\multiput(160, 100)(30, 0){3}{\circle{5}}

			\put(125, 97){$x_{t, n}$}

			\color{red}
			\multiput(10, 10)(0, 30){6}{\circle*{5}}
			\multiput(70, 10)(30, 0){6}{\circle*{5}}

			\put(35, 7){$z$}

			\put(0, 0){\framebox(20, 110)}
			\put(60, 0){\framebox(80, 20)}

		\end{picture}
	\end{center}
	\caption{
		Situation in Definition~\ref{definition:k_tn}.
		The initial variables in the boxed regions, which lie in the past light cone emanating from $(t, n)$ but are different from $z$, belong to $\mathbb{K}_{t, n}$.
	}
	\label{figure:definition_ktn}
\end{figure}

The following lemma immediately follows from the definition of $\mathbb{K}_{t, n}$.

\begin{lemma}
	\begin{enumerate}
		\item
		$\mathbb{K} = \bigcup_{(t, n) \in H} \mathbb{K}_{t, n}$,
		$\overline{\mathbb{K}} = \bigcup_{(t, n) \in H} \overline{\mathbb{K}_{t, n}}$.

		\item
		If $(t, n) \le (s, m)$, then
		$\mathbb{K}_{t, n} \subset \mathbb{K}_{s, m}$
		and
		$\overline{\mathbb{K}_{t, n}} \subset \overline{\mathbb{K}_{s, m}}$.

		\item
		$x_{t, n} (z) \in \mathbb{K}_{t, n} (z)$.

	\end{enumerate}
\end{lemma}

\begin{lemma}\label{lemma:trdeg_two}
	Let $z^{*} \in \mathbb{P}^1 \left( \overline{\mathbb{K}} \right)$ and let $(t, n) \in H \setminus H_0$.
	If
	\[
		\operatorname{trdeg}_{\mathbb{C}} \mathbb{C}(x_{t - 1, n}(z^{*}), x_{t, n - 1}(z^{*}), x_{t - 1, n - 1}(z^{*})) \le 2,
	\]
	then
	\[
		(t, n) \ge (t_0, n_0)
	\]
	and
	\[
		z^{*} \in \mathbb{P}^1 \left( \overline{\mathbb{K}_{t, n}} \right).
	\]
	In particular, if $x_{t, n}(z^{*})$ is a constant singularity, then $(t, n) \ge (t_0, n_0)$ and $z^{*} \in \mathbb{P}^1 \left( \overline{\mathbb{K}_{t, n}} \right)$.
\end{lemma}
\begin{proof}
	Since $x_{t - 1, n}(z)$, $x_{t, n - 1}(z)$ and $x_{t - 1, n - 1}(z)$ are algebraically independent over $\mathbb{C}$ and
	\[
		\operatorname{trdeg}_{\mathbb{C}} \mathbb{C}(x_{t - 1, n}(z^{*}), x_{t, n - 1}(z^{*}), x_{t - 1, n - 1}(z^{*})) \le 2,
	\]
	at least one of $x_{t - 1, n}(z)$, $x_{t, n - 1}(z)$ and $x_{t - 1, n - 1}(z)$ must depend on $z$.
	Therefore, we have
	\[
		(t, n) \ge (t_0, n_0).
	\]
	
	Let us show that $z^{*} \in \mathbb{P}^1 \left( \overline{\mathbb{K}_{t, n}} \right)$.
	We may assume that $z^{*} \ne \infty$.
	First, we consider the case where some of $x_{t - 1, n}(z^{*})$, $x_{t, n - 1}(z^{*})$ and $x_{t - 1, n - 1}(z^{*})$ is $\infty$.
	Suppose that
	\[
		h \in \{ (t - 1, n), \ (t, n - 1), \ (t - 1, n - 1) \}
	\]
	satisfies $x_h(z^{*}) = \infty$.
	Since $x_h(z) \in \mathbb{K}_{t, n}(z)$, we can express $x_h(z)$ as
	\[
		x_h(z) = \frac{p(z)}{q(z)},
	\]
	where $p(z), q(z) \in \mathbb{K}_{t, n}[z]$ are coprime to each other.
	Since $x_h(z^{*}) = \infty$, $z^{*}$ satisfies $q(z^{*}) = 0$, which implies $z^{*} \in \overline{\mathbb{K}_{t, n}}$.

	Next, we consider the remaining case.
	Since
	\[
		\operatorname{trdeg}_{\mathbb{C}} \mathbb{C}(x_{t - 1, n}(z^{*}), x_{t, n - 1}(z^{*}), x_{t - 1, n - 1}(z^{*})) \le 2,
	\]
	there exists $F = F(B, C, D) \in \mathbb{C}[B, C, D]$, $F \ne 0$, such that
	\begin{equation}\label{equation:z_star_nontrivial}
		F(x_{t - 1, n}(z^{*}), x_{t, n - 1}(z^{*}), x_{t - 1, n - 1}(z^{*})) = 0.
	\end{equation}
	Since $x_{t - 1, n}(z)$, $x_{t, n - 1}(z)$ and $x_{t - 1, n - 1}(z)$ are algebraically independent over $\mathbb{C}$, we have
	\[
		F(x_{t - 1, n}(z), x_{t, n - 1}(z), x_{t - 1, n - 1}(z)) \ne 0.
	\]
	Therefore, the relation \eqref{equation:z_star_nontrivial}, whose coefficients belong to $\mathbb{K}_{t, n}$, is not trivial with respect to $z^{*}$.
	Hence, we have $z^{*} \in \overline{\mathbb{K}_{t, n}}$.
\end{proof}

\begin{definition}[$\Delta_{t, n}$]\label{definition:delta}
	For $(t, n) \in H_{\ge (t_0, n_0)}$, we define a subset $\Delta_{t, n} \subset \mathbb{P}^1 \left( \overline{\mathbb{K}} \right)$ as
	\[
		\Delta_{t, n} = \mathbb{P}^1 \left( \overline{\mathbb{K}_{t, n}} \right) \setminus \left( \mathbb{P}^1(\mathbb{C}) \cup \bigcup_{(s, m) \in \left( H_{\ge (t_0, n_0)} \cap H_{\le (t, n)} \right) \setminus \{ (t, n) \}} \mathbb{P}^1 \left( \overline{\mathbb{K}_{s, m}} \right) \right),
	\]
	where $z = x_{t_0, n_0}$.
	For $(t, n) \in H \setminus H_{\ge (t_0, n_0)}$, we set
	\[
		\Delta_{t, n} = \emptyset
	\]
	for convenience.
\end{definition}

\begin{remark}
	The $\mathbb{P}^1(\mathbb{C})$ on the right-hand side is contained in the $\bigcup$-part unless that is the empty set, which happens only when $(t, n) = (t_0, n_0)$.
\end{remark}

\begin{remark}\label{remark:partition}
	It is clear by definition that $\mathbb{P}^1 \left( \overline{\mathbb{K}} \right)$ is covered as the direct sum of $\mathbb{P}^1(\mathbb{C})$ and $\Delta_{t, n}$:
	\[
		\mathbb{P}^1 \left( \overline{\mathbb{K}} \right) = \mathbb{P}^1(\mathbb{C}) \sqcup \bigsqcup_{(t, n) \in H} \Delta_{t, n},
	\]
	with $\sqcup$ and $\bigsqcup$ denoting direct sum of sets (disjoint union).
	It will be revealed in Proposition~\ref{proposition:p1_decomposition} that this decomposition corresponds to where a constant singularity pattern starts.
\end{remark}

\begin{lemma}
	Let $z^{*} \in \Delta_{t, n}$.
	If $(s, m) \in H \setminus H_0$ satisfies
	\[
		\operatorname{trdeg}_{\mathbb{C}} \mathbb{C} \left( x_{s - 1, m}(z^{*}), x_{s, m - 1}(z^{*}), x_{s - 1, m - 1}(z^{*}) \right) \le 2,
	\]
	then $(t, n) \le (s, m)$.
	In particular, if $x_{s, m}(z)$ has a constant singularity at $z = z^{*} \in \Delta_{t, n}$, then $(t, n) \le (s, m)$.
\end{lemma}
\begin{proof}
	Note that $(t, n) \ge (t_0, n_0)$ since $\Delta_{t, n} \ne \emptyset$.
	Since the transcendental degree is less than $3$, it follows from Lemma~\ref{lemma:trdeg_two} that $(s, m) \ge (t_0, n_0)$ and $z^{*} \in \mathbb{P}^1 \left( \overline{\mathbb{K}_{s, m}} \right)$.

	Let us show that $t \le s$.
	If $t = t_0$, then we have $t = t_0 \le s$.
	If not, then $z^{*}$ belongs to $\mathbb{P}^1 \left( \overline{\mathbb{K}_{t, n}} \right) \setminus \mathbb{P}^1 \left( \overline{\mathbb{K}_{t - 1, n}} \right)$.
	That is, there exists $n'$ such that $z^{*}$ depends on the initial variable $x_{t, n'}$.
	Therefore, it follows from $z^{*} \in \mathbb{P}^1 \left( \overline{\mathbb{K}_{s, m}} \right)$ that $s \ge t$.
	We can show that $n \le m$ in the same way.
\end{proof}

\begin{lemma}\label{lemma:trdeg_at_least_two}
	Let $(t, n) \in H \setminus H_0$.
	If $z^{*} \in \Delta_{t, n}$, then
	\[
		\operatorname{trdeg}_{\mathbb{C}} \mathbb{C} \left( x_{t - 1, n}(z^{*}), x_{t, n - 1}(z^{*}), x_{t - 1, n - 1}(z^{*}) \right) \ge 2.
	\]
	In particular, if $x_{t, n}(z)$ has a constant singularity at $z = z^{*} \in \Delta_{t, n}$, then the above transcendental degree is exactly $2$.
\end{lemma}
\begin{proof}
	Note that $(t, n) \ge (t_0, n_0)$ since $\Delta_{t, n} \ne \emptyset$.
	We show that $x_{t - 1, n}(z^{*})$ and $x_{t - 1, n - 1}(z^{*})$ are algebraically independent over $\mathbb{C}$.
	Note that $x_{t - 1, n}(z), x_{t - 1, n - 1}(z) \in \mathbb{K}_{t - 1, n}(z)$ are algebraically independent over $\mathbb{C}$.

	If $t = t_0$, then $x_{t - 1, n}(z)$ and $x_{t - 1, n - 1}(z)$ are independent of $z$.
	Therefore, $x_{t - 1, n}(z^{*})$ and $x_{t - 1, n - 1}(z^{*})$ are algebraically independent over $\mathbb{C}$.

	Let us consider the remaining case, i.e., $t > t_0$.
	It is sufficient to show that $z^{*}$ is transcendental over $\mathbb{K}_{t - 1, n}$.
	Let $n'$ be the least positive integer that satisfies $(t, n - n') \in H_0$ and let $w = x_{t, n - n'}$.
	It follows from $t > t_0$ that $w \ne z$.
	Since $z^{*} \in \Delta_{t, n}$, $z^{*}$ depends on the initial variable $w$.
	Hence, $z^{*}$ is transcendental over $\mathbb{K}_{t - 1, n}$, which completes the proof.
\end{proof}

From here on in this section, we always assume that $\Phi$ satisfies the $\partial$-factor condition (Definition~\ref{definition:d_factor_condition}).

\begin{lemma}\label{lemma:key_lemma}
	Let $(t, n) \in H$ and $z^{*} \in \mathbb{P}^1 \left( \overline{\mathbb{K}_{t, n}} \right)$.
	Suppose that $(s, m) \in H$ satisfies $s > t$ and $s > t_0$ and let $m'$ be the least nonnegative integer satisfying $(s, m - m') \in H_0$.
	Let $w = x_{s, m - m'}$ and think of $x_{s, m}$ and $x_{s, m - 1}$ as rational functions in $z$ and $w$,
	i.e., $x_{s, m} = x_{s, m}(z, w)$ and $x_{s, m - 1} = x_{s, m - 1}(z, w)$.
	If
	\begin{itemize}
		\item 
		$\partial_w x_{s, m - 1}(z^{*}, w) \ne 0$, $\partial_w \frac{1}{x_{s, m - 1}(z^{*}, w)} \ne 0$,

		\item
		$x_{s - 1, m}(z^{*})$ and $x_{s - 1, m - 1}(z^{*})$ are not constant singularities,

	\end{itemize}
	then
	\[
		\partial_w x_{s, m}(z^{*}, w) \ne 0, \quad
		\partial_w \frac{1}{x_{s, m}(z^{*}, w)} \ne 0.
	\]
	In particular, $x_{s, m}(z^{*}, w)$ depends on $w$ and is not a constant singularity.

	The statement with the $t$- and $n$-axes interchanged holds, too.
\end{lemma}
\begin{remark}
	The assumption $s > t_0$ is necessary to guarantee that $z \ne w$.
\end{remark}
\begin{proof}[Proof of Lemma~\ref{lemma:key_lemma}]
	We only consider the case $z^{*} \ne \infty$.
	Since $z^{*} \in \overline{\mathbb{K}_{t, n}}$ and $t < s$, $z^{*}$ does not depend on $w = x_{s, m - m'}$.
	Moreover, it follows from $\partial_w x_{s, m - 1}(z^{*}, w) \ne 0$ and $\partial_w \frac{1}{x_{s, m - 1}(z^{*}, w)} \ne 0$ that $x_{s, m - 1}(z^{*}, w)$ is not a constant singularity.

	First, we show that $\partial_w x_{s, m}(z^{*}) \ne 0$.
	Since $x_{s - 1, m}$ and $x_{s - 1, m - 1}$ are both independent of $w$ and
	\[
		x_{s, m}(z, w) = \Phi(x_{s - 1, m}(z), x_{s, m - 1}(z, w), x_{s - 1, m - 1}(z)),
	\]
	we have
	\[
		\partial_w x_{s, m}(z, w) = \partial_C \Phi(x_{s - 1, m}(z), x_{s, m - 1}(z, w), x_{s - 1, m - 1}(z)) \times \partial_w x_{s, m - 1}(z, w).
	\]
	The second factor of the right-hand side does not vanish at $z = z^{*}$ by the assumption.
	Let us show that $\partial_C \Phi(x_{s - 1, m}(z^{*}), x_{s, m - 1}(z^{*}, w), x_{s - 1, m - 1}(z^{*})) \ne 0$.

	Since none of $x_{s - 1, m}(z^{*})$, $x_{s, m - 1}(z^{*}, w)$ and $x_{s - 1, m - 1}(z^{*})$ is $\infty$, the first factor
	\[
		\partial_C \Phi(x_{s - 1, m}(z^{*}), x_{s, m - 1}(z^{*}, w), x_{s - 1, m - 1}(z^{*}))
	\]
	cannot be zero in the form $\frac{1}{\infty}$.
	Therefore, the first factor becomes $0$ only if its numerator vanishes.
	Let $F = F(B, C, D) \in \mathbb{C}(B, C, D)$ be an arbitrary irreducible factor of the numerator of the rational function $\partial_C \Phi (B, C, D)$.
	By the $\partial$-factor condition, $F$ satisfies $F = B - \alpha$, $F = D - \alpha$ (for some $\alpha \in \mathbb{C}$) or $\partial_C F \ne 0$.
	If $F = B - \alpha$ or $F = D - \alpha$, then we have
	\[
		F(x_{s - 1, m}(z^{*}), x_{s, m - 1}(z^{*}, w), x_{s - 1, m - 1}(z^{*})) \ne 0
	\]
	since neither $x_{s - 1, m}(z^{*})$ nor $x_{s - 1, m - 1}(z^{*})$ is a constant singularity.
	Let us consider the case $\partial_C F \ne 0$.
	Since $x_{s - 1, m}(z^{*})$ and $x_{s - 1, m - 1}(z^{*})$ are not constant singularities, we have
	\[
		\operatorname{trdeg}_{\mathbb{C}} \mathbb{C}(x_{s - 1, m}(z^{*}), x_{s - 1, m - 1}(z^{*})) \ge 1.
	\]
	Moreover, $x_{s, m - 1}(z^{*}, w)$ is transcendental over $\mathbb{C}(x_{s - 1, m}(z^{*}), x_{s - 1, m - 1}(z^{*}))$ since $x_{s, m - 1}(z^{*}, w)$ depends on $w$ but the other two do not.
	Therefore, it follows from Lemma~\ref{lemma:polynomial_nonvanishing} that
	\[
		F(x_{s - 1, m}(z^{*}), x_{s, m - 1}(z^{*}, w), x_{s - 1, m - 1}(z^{*})) \ne 0.
	\]
	Since $F$ is an arbitrary irreducible factor of the numerator, we have
	\[
		\partial_C \Phi(x_{s - 1, m}(z^{*}), x_{s, m - 1}(z^{*}, w), x_{s - 1, m - 1}(z^{*})) \ne 0.
	\]

	Next, we show that $\partial_w \frac{1}{x_{s, m}(z^{*})} \ne 0$.
	Since
	\[
		\partial_w \frac{1}{x_{s, m}(z, w)} = \left( \partial_C \frac{1}{\Phi(B, C, D)} \right) \Big|_{(B, C, D) = (x_{s - 1, m}(z), x_{s, m - 1}(z, w), x_{s - 1, m - 1}(z))} \times \partial_w x_{s, m - 1}(z, w),
	\]
	one can prove in the same way as above that the right-hand side does not vanish at $z = z^{*}$.
\end{proof}

\begin{proof}[Proof of Proposition~\ref{proposition:key_lemma}]
	By Lemma~\ref{lemma:trdeg_two}, we have $(t, n) \ge (t_0, n_0)$ and $z^{*} \in \mathbb{P}^1 \left( \overline{\mathbb{K}_{t, n}} \right)$.
	Since $s > t \ge t_0$, the statement follows from Lemma~\ref{lemma:key_lemma}.
\end{proof}

\begin{remark}
	It is clear by the above proof that Proposition~\ref{proposition:key_lemma} still holds even if the assumption that $z^{*}$ generates a constant singularity at $(t, n)$ is replaced with $z^{*} \in \mathbb{P}^1 \left( \overline{\mathbb{K}_{t, n}} \right)$ and $s > t_0$.
\end{remark}

\begin{proof}[Proof of Theorem~\ref{theorem:first_singularity}]
	Let $(t, n)$ be a minimal element of the set.
	It is sufficient to show that the set has no other minimal element.
	Because of the minimality, there does not exist a constant singularity in the region $H_{\le (t, n)}$ except for $x_{t, n}(z^{*})$.
	Since
	\[
		x_{t, n - 1}(z^{*}), x_{t, n - 2}(z^{*}), \cdots
	\]
	are not constant singularities, it follows from Proposition~\ref{proposition:key_lemma} that
	\[
		x_{t + 1, n - 1}(z^{*}), x_{t + 1, n - 2}(z^{*}), \cdots
	\]
	are not constant singularities, either.
	Repeating this procedure, one finds that there does not exist a constant singularity in the region
	\[
		\{ (s, m) \in H \mid s > t; \ m < n \}.
	\]
	Interchanging the $t$- and $n$- axes, one can show in the same way that there does not exist a constant singularity in the region
	\[
		\{ (s, m) \in H \mid s < t; \ m > n \}.
	\]
	Hence, $(t, n)$ is a unique minimal element of the set.
\end{proof}

\begin{proposition}\label{proposition:p1_decomposition}
	If $x_{t, n}(z^{*})$ is the first singularity of a movable pattern, then $z^{*}$ belongs to $\Delta_{t, n}$.
	That is, the partition of $\mathbb{P}^1 \left( \overline{\mathbb{K}} \right)$ in Remark~\ref{remark:partition} corresponds to where a constant singularity pattern starts.
\end{proposition}
\begin{proof}
	Note that $(t, n)$ belongs to $H_{\ge (t_0, n_0)} \setminus \{ (t_0, n_0) \}$ since the pattern is movable.

	First, consider the case where $t > t_0$ and $n > n_0$.
	Since
	\[
		\Delta_{t, n} = \mathbb{P}^1 \left( \overline{\mathbb{K}_{t, n}} \right) \setminus \left( \mathbb{P}^1 \left( \overline{\mathbb{K}_{t - 1, n}} \right) \cup \mathbb{P}^1 \left( \overline{\mathbb{K}_{t, n - 1}} \right) \right),
	\]
	it is sufficient to show that $z^{*}$ belongs to neither $\mathbb{P}^1 \left( \overline{\mathbb{K}_{t - 1, n}} \right)$ nor $\mathbb{P}^1 \left( \overline{\mathbb{K}_{t, n - 1}} \right)$.
	Assume that $z^{*} \in \mathbb{P}^1 \left( \overline{\mathbb{K}_{t - 1, n}} \right)$.
	Since
	\[
		x_{t - 1, n}(z^{*}), x_{t - 1, n - 1}(z^{*}), \cdots
	\]
	are not constant singularities, it follows from Proposition~\ref{proposition:key_lemma} that
	$x_{t, n}(z^{*})$ is not a constant singularity, either, which is a contradiction.
	Hence, we have $z^{*} \notin \mathbb{P}^1 \left( \overline{\mathbb{K}_{t - 1, n}} \right)$.
	One can show that $z^{*} \notin \mathbb{P}^1 \left( \overline{\mathbb{K}_{t, n - 1}} \right)$ in the same way.

	Next, consider the case where $n = n_0$ or $t = t_0$.
	Since $(t, n) \ne (t_0, n_0)$, we may assume that $t = t_0$ and $n > n_0$.
	Note that $(t, n - 1)$ belongs to $H_{\ge (t_0, n_0)} \cap H_{\le (t, n)}$ but $(t - 1, n)$ does not.
	Therefore, $\Delta_{t, n}$ in this case is written as
	\[
		\Delta_{t, n} = \mathbb{P}^1 \left( \overline{\mathbb{K}_{t, n}} \right) \setminus \mathbb{P}^1 \left( \overline{\mathbb{K}_{t, n - 1}} \right).
	\]
	One can show in the same way as above that $(t, n) \notin \mathbb{P}^1 \left( \overline{\mathbb{K}_{t, n - 1}} \right)$, which completes the proof.
\end{proof}

\begin{definition}[Spontaneous Occurrence of $x_{t, n}(z^{*}) = \alpha$]\label{definition:spontaneous_occurrence_rigorous}
	For $z^{*} \in \mathbb{P}^1 \left( \overline{\mathbb{K}} \right)$ and $\alpha \in \mathbb{P}^1$, we say that $x_{t, n}(z^{*})$ spontaneously becomes $\alpha$ if $x_{t, n}(z^{*}) = \alpha$ and $z^{*} \in \Delta_{t, n}$.
	By Proposition~\ref{proposition:p1_decomposition}, the above condition holds if and only if $x_{t, n}(z^{*}) = \alpha$ is the first singularity of the pattern.
\end{definition}

\begin{corollary}\label{corollary:not_zero_over_zero}
	If $x_{t, n}(z^{*})$ is the first constant singularity of a movable pattern, then
	\[
		\operatorname{trdeg}_{\mathbb{C}} \mathbb{C} \left( x_{t - 1, n}(z^{*}), x_{t, n - 1}(z^{*}), x_{t - 1, n - 1}(z^{*}) \right) = 2.
	\]
	In particular, the first singularity $x_{t, n}(z^{*})$ of any movable pattern can be directly calculated from $x_{t - 1, n}(z^{*})$, $x_{t, n - 1}(z^{*})$ and $x_{t - 1, n - 1}(z^{*})$ as
	\[
		x_{t, n}(z^{*}) = \Phi(x_{t - 1, n}(z^{*}), x_{t, n - 1}(z^{*}), x_{t - 1, n - 1}(z^{*})).
	\]
\end{corollary}
\begin{proof}
	The first part follows from Lemmas~\ref{lemma:trdeg_two}, \ref{lemma:trdeg_at_least_two} and Proposition~\ref{proposition:p1_decomposition}.
	The second part follows from Lemma~\ref{lemma:substitution_trdeg_codimension_one}.
\end{proof}

\begin{remark}
	In order to calculate $x_{t, n}(z^{*})$, one usually calculates $x_{t, n}(z)$ first and then substitutes $z = z^{*}$, or introduces an infinitesimal parameter such as $z^{*} + \varepsilon$ first and then substitutes $\varepsilon = 0$.
	Using this corollary, however, one can calculate $x_{t, n}(z^{*})$ directly from $x_{t - 1, n}(z^{*})$, $x_{t, n - 1}(z^{*})$ and $x_{t - 1, n - 1}(z^{*})$.
	In particular, the first singularity of a pattern is not the result of $\frac{0}{0}$.
	This is an important property when finding a value that does not appear as the starting value of any movable pattern (Proposition~\ref{proposition:check_hidden_value}).
\end{remark}

\begin{proof}[Proof of Proposition~\ref{proposition:confining}]
	Since $x_{t, n}(z^{*})$ is the first singularity, it follows from Theorem~\ref{theorem:first_singularity} that the pattern has no constant singularity in the region
	\[
		\{ (t', n') \in H \mid t' > t; \ n' < n \} \cup \{ (t', n') \in H \mid t' < t; \ n' > n \}.
	\]
	Since
	\[
		x_{s + 1, m + 1}(z^{*}), x_{s + 1, m}(z^{*}), \ldots, x_{s + 1, n}(z^{*}), x_{s + 1, n - 1}(z^{*}), \cdots
	\]
	are not constant singularities, it follows from Proposition~\ref{proposition:key_lemma} that
	\[
		x_{s + 2, m + 1}(z^{*}), x_{s + 2, m}(z^{*}), \cdots
	\]
	are not constant singularities, either.
	Repeating this procedure, one finds that the pattern has no constant singularity in the region
	\[
		\{ (t', n') \in H \mid t' > s; \ n' \le m + 1 \}.
	\]
	Since the pattern has no constant singularity on the line
	\[
		\{ (t', m + 1) \in H \mid t' \in \mathbb{Z} \},
	\]
	it follows from Proposition~\ref{proposition:key_lemma} that there does not exist a constant singularity on the line
	\[
		\{ (t', m + 2) \in H \mid t' \in \mathbb{Z} \},
	\]
	either.
	Repeating this procedure, one finds that the pattern has no constant singularity in the region
	\[
		\{ (t', n') \in H \mid t' \in \mathbb{Z}; \ n' \ge m + 1 \},
	\]
	which completes the proof.
\end{proof}

\begin{lemma}\label{lemma:algebraic_independence_fresh}
	Let $(t, n) \in H_{\ge (t_0, n_0)}$ and let $z^{*} \in \mathbb{P}^1 \left( \overline{\mathbb{K}_{t, n}} \right)$.
	Then,
	\[
		x_{t + 1, n - 1}(z^{*}), x_{t + 2, n - 1}(z^{*}), x_{t + 3, n - 1}(z^{*}), \cdots \quad
	\]
	and
	\[
		x_{t - 1, n + 1}(z^{*}), x_{t - 1, n + 2}(z^{*}), x_{t - 1, n + 2}(z^{*}), \cdots
	\]
	are algebraically independent over $\mathbb{K}_{t, n}$.
\end{lemma}
\begin{proof}
	For each $j \in \mathbb{Z}_{\ge 1}$, let $n'_j$ be the least nonnegative integer that satisfy
	\[
		(t + j, n - 1 - n'_j) \in H_0
	\]
	and let
	\[
		w_j = x_{t + j, n - 1 - n'_j}.
	\]
	Note that $w_j$ is an initial variable different from $z$ and does not belong to $\mathbb{K}_{t, n}$.
	Then, it follows from Proposition~\ref{proposition:key_lemma} that $x_{t + 1, n - 1}(z^{*})$ depends on $w_1$.
	Therefore, $x_{t + 1, n - 1}(z^{*})$ is transcendental over $\mathbb{K}_{t, n}$.
	Using Proposition~\ref{proposition:key_lemma} again, one can show that $x_{t + 2, n - 1}(z^{*})$ is transcendental over $\mathbb{K}_{t, n}(x_{t + 1, n - 1}(z^{*}))$.
	Repeating this procedure, one obtains that
	\[
		x_{t + 1, n - 1}(z^{*}), x_{t + 2, n - 1}(z^{*}), x_{t + 3, n - 1}(z^{*}), \cdots
	\]
	are algebraically independent over $\mathbb{K}_{t, n}$.

	Let
	\[
		K = \mathbb{K}_{t, n}(x_{t + 1, n - 1}(z^{*}), x_{t + 2, n - 1}(z^{*}), x_{t + 3, n - 1}(z^{*}), \cdots).
	\]
	For each $\ell \in \mathbb{Z}_{\ge 1}$, let $t'_{\ell}$ be the least nonnegative integer that satisfies
	\[
		(t - 1 - t'_{\ell}, n + \ell) \in H_0
	\]
	and let
	\[
		w'_{\ell} = x_{t - 1 - t'_{\ell}, n + \ell}.
	\]
	Then, $w'_j$ is an initial variable different from $z$ and does not belong to $K$.
	Therefore, one can show in a similar way to the above that
	\[
		x_{t - 1, n + 1}(z^{*}), x_{t - 1, n + 2}(z^{*}), x_{t - 1, n + 2}(z^{*}), \cdots
	\]
	are algebraically independent over $K$,
	which completes the proof.
\end{proof}

\begin{lemma}\label{lemma:algebraic_independence_hard}
	Let $z^{*} \in \mathbb{P}^1 \left( \overline{\mathbb{K}} \right)$ and suppose that $x_{t, n}(z^{*}) = x^{*} \in \mathbb{P}^1(\mathbb{C})$ is the first singularity of a movable pattern.
	Suppose that the equation satisfies the basic pattern condition for $x^{*}$.
	Then, $x_{t - 1, n - 1}(z^{*})$ is algebraic over $\mathbb{C}(x_{t - 1, n}(z^{*}), x_{t, n - 1}(z^{*}))$.
	In particular, $x_{t - 1, n}(z^{*})$ and $x_{t, n - 1}(z^{*})$ are algebraically independent over $\mathbb{C}$.
\end{lemma}
\begin{proof}
	By Corollary~\ref{corollary:not_zero_over_zero}, we have
	\[
		\operatorname{trdeg}_{\mathbb{C}} \mathbb{C}(x_{t - 1, n}(z^{*}), x_{t, n - 1}(z^{*}), x_{t - 1, n - 1}(z^{*})) = 2.
	\]
	Since none of $x_{t - 1, n}(z^{*})$, $x_{t, n - 1}(z^{*})$ and $x_{t - 1, n - 1}(z^{*})$ is $\infty$ and
	\[
		\Phi(x_{t - 1, n}(z^{*}), x_{t, n - 1}(z^{*}), x_{t - 1, n - 1}(z^{*})) = x^{*},
	\]
	there exists an irreducible factor $F = F(B, C, D) \in \mathbb{C}[B, C, D]$ of the relation $\Phi(B, C, D) = x^{*}$ such that
	\[
		F(x_{t - 1, n}(z^{*}), x_{t, n - 1}(z^{*}), x_{t - 1, n - 1}(z^{*})) = 0.
	\]
	According to the assumption, $x_{t - 1, n}(z^{*})$ and $x_{t, n - 1}(z^{*})$ are not constant singularities.
	Therefore, by the basic pattern condition for $x^{*}$, we have
	\[
		\partial_D F \ne 0.
	\]
	Since $\operatorname{trdeg}_{\mathbb{C}} \mathbb{C}(x_{t - 1, n}(z^{*}), x_{t, n - 1}(z^{*})) \ge 1$,
	it follows from Lemma~\ref{lemma:polynomial_nonvanishing} that $x_{t - 1, n - 1}(z^{*})$ is algebraic over $\mathbb{C}(x_{t - 1, n}(z^{*}), x_{t, n - 1}(z^{*}))$.
	The second statement follows from
	\[
		\operatorname{trdeg}_{\mathbb{C}} \mathbb{C}(x_{t - 1, n}(z^{*}), x_{t, n - 1}(z^{*}), x_{t - 1, n - 1}(z^{*})) = 2.
	\]
\end{proof}

\begin{proof}[Proof of Theorem~\ref{theorem:main_theorem}]
	By Lemmas~\ref{lemma:algebraic_independence_hard} and \ref{lemma:algebraic_independence_fresh},
	\[
		x_{t - 1, n}(z^{*}), x_{t - 1, n + 1}(z^{*}), x_{t - 1, n + 2}(z^{*}), \cdots \quad
	\]
	and
	\[
		x_{t, n - 1}(z^{*}), x_{t + 1, n - 1}(z^{*}), x_{t + 2, n - 1}(z^{*}), \cdots
	\]
	are algebraically independent over $\mathbb{C}$.
	Therefore, the statement follows from Lemma~\ref{lemma:main_lemma}.
\end{proof}

\begin{definition}[$D_{s, m}(\beta)$, $Z_{s, m}(\beta)$, $D(\gamma)$]\label{definition:degree_divisor}
	For $(s, m) \in H$ and $\beta \in \mathbb{P}^1(\mathbb{C})$, we define a divisor $D_{s, m}(\beta) \in \operatorname{Div} \mathbb{P}^1 \left( \overline{\mathbb{K}} \right)$ as
	\[
		D_{s, m}(\beta) = \operatorname{div} \left( x_{s, m} = \beta \mid \Delta_{s, m} \right),
	\]
	where the definition of the right-hand side is given in Definition~\ref{definition:divisor} and we think of
	\[
		x_{s, m} \colon \mathbb{P}^1 \left( \overline{\mathbb{K}} \right) \to \mathbb{P}^1 \left( \overline{\mathbb{K}} \right); \
		z \mapsto x_{s, m}(z)
	\]
	as a rational function in $z$.
	Let $Z_{s, m}(\beta)$ be its degree, i.e.,
	\[
		Z_{s, m}(\beta) = \deg D_{s, m}(\beta),
	\]
	which is a rigorous definition of ``the number of spontaneous occurrences of $x_{s, m} = \beta$ counted with multiplicity.''

	For $\gamma \in \mathbb{P}^1(\mathbb{C})$, we define a prime divisor $D(\gamma) \in \operatorname{Div} \mathbb{P}^1 \left( \overline{\mathbb{K}} \right)$ as
	\[
		D(\gamma) = \gamma,
	\]
	which corresponds to the fixed pattern $z = \gamma$.
\end{definition}

\begin{lemma}\label{lemma:divisor_meaning}
	Let $z^{*} \in \mathbb{P}^1 \left( \overline{\mathbb{K}} \right)$ and let $\beta \in \mathbb{P}^1(\mathbb{C})$.
	Then, $z^{*} \in \operatorname{supp} D_{t, n}(\beta)$ if and only if $x_{t, n}(z^{*}) = \beta$ and it is the first singularity of the pattern corresponding to $z = z^{*}$.
	Moreover, the coefficient of $z^{*}$ in $D_{t, n}(\beta)$ coincides with the multiplicity of the constant singularity $x_{t, n}(z^{*}) = \beta$.
\end{lemma}
\begin{proof}
	Clear by Remark~\ref{remark:partition} and Proposition~\ref{proposition:p1_decomposition}.
\end{proof}

\begin{lemma}\label{lemma:divisor_equality}
	Under the same assumptions and notations as in Theorem~\ref{theorem:main_degree}, the following equalities of divisors hold:
	\begin{align*}
		& \operatorname{div} \left( x_{t, n} = \alpha_j \mid \Delta_{s, m}, x_{s, m} = \beta_i \right)
		= \operatorname{mult}^{\text{basic}, \beta_i}_{t - s, n - m}(\alpha_j) D_{s, m}(\beta_i), \\
		& \operatorname{div} \left( x_{t, n} = \alpha_j \mid \Delta_{s, m} \right)
		= \sum_i \operatorname{mult}^{\text{basic}, \beta_i}_{t - s, n - m}(\alpha_j) D_{s, m}(\beta_i), \\
		& \operatorname{div} \left( x_{t, n} = \alpha_j \mid \mathbb{P}^1 \left( \overline{\mathbb{K}} \right) \setminus \mathbb{P}^1(\mathbb{C}) \right)
		= \sum_{(s, m) \in H} \sum_i \operatorname{mult}^{\text{basic}, \beta_i}_{t - s, n - m}(\alpha_j) D_{s, m}(\beta_i), \\
		& \operatorname{div} \left( x_{t, n} = \alpha_j \mid \mathbb{P}^1(\mathbb{C}) \right)
		= \sum_{\ell} \operatorname{mult}^{\gamma_{\ell}}_{t, n}(\alpha_j) D(\gamma_{\ell}), \\
		& \operatorname{div} \left( x_{t, n} = \alpha_j \right)
		= \sum_i \sum_{(s, m) \in H} \operatorname{mult}^{\text{basic}, \beta_i}_{t - s, n - m}(\alpha_j) D_{s, m}(\beta_i)
		+ \sum_{\ell} \operatorname{mult}^{\gamma_{\ell}}_{t, n}(\alpha_j) D(\gamma_{\ell}).
	\end{align*}
\end{lemma}
\begin{proof}
	The first equality follows from Theorem~\ref{theorem:main_theorem} and Lemma~\ref{lemma:divisor_meaning}.
	Summing the first equality for $i$, we have the second equality since $x_{s, m}$ cannot become $\alpha_j$ unless $x_{t, n} = \beta_i$ for some $i$.
	Using
	\[
		\mathbb{P}^1 \left( \overline{\mathbb{K}} \right) \setminus \mathbb{P}^1(\mathbb{C}) = \bigsqcup_{(s, m) \in H} \Delta_{s, m},
	\]
	we have
	\begin{align*}
		\operatorname{div} \left( x_{t, n} = \alpha_j \mid \mathbb{P}^1 \left( \overline{\mathbb{K}} \right) \setminus \mathbb{P}^1(\mathbb{C}) \right)
		&= \sum_{(s, m) \in H} \operatorname{div} \left( x_{t, n} = \alpha_j \mid \Delta_{s, m} \right) \\
		&= \sum_{(s, m) \in H} \sum_i \operatorname{mult}^{\text{basic}, \beta_i}_{t - s, n - m}(\alpha_j) D_{s, m}(\beta_i),
	\end{align*}
	which is the third equality.
	The fourth equality immediately follows from the definition of $D(\gamma_{\ell})$ and $\operatorname{mult}^{\gamma_{\ell}}_{t, n}(\alpha_j)$.
	The last equality is the sum of the third and fourth equalities.
\end{proof}

\begin{proof}[Proof of Theorem~\ref{theorem:main_degree}]
	The theorem immediately follows from the last equality of Lemma~\ref{lemma:divisor_equality} since
	\[
		\deg_z x_{t, n}(z) = \deg (\operatorname{div} (x_{t, n} = \alpha_j)), \quad
		\deg D_{s, m}(\beta_i) = Z_{s, m}(\beta_i), \quad
		\deg D(\gamma_{\ell}) = 1.
	\]
\end{proof}

\begin{proof}[Proof of Theorem~\ref{theorem:degree_solvable}]
	If $(s, m) \in H_0$, then we have $Z_{s, m}(\beta_i) = 0$ since $Z_{s, m}(\beta_i)$ corresponds to a movable pattern.
	Note that for each $(t, n) \in H \setminus H_0$, the degree relations at $(t, n)$ consist of $J - 1$ linear equations on $Z_{s, m}(\beta_i)$ ($(s, m) \in H_{\le (t, n)} \setminus H_0$ and $i = 1, \ldots, I$).
	In this proof, we think of $Z_{s, m}(\beta_i)$ (resp.\ the degree relations) just as unknown variables (resp.\ linear equations).
	We consider whether or not the system of linear equations on these variables has a unique solution.

	First, we suppose that $I \ge J$.
	Let $(t, n) \in H \setminus H_0$ and let $r = \# (H_{\le (t, n)} \setminus H_0)$.
	Then, the number of the variables $Z_{s, m}(\beta_i)$ ($(s, m) \in H_{\le (t, n)} \setminus H_0$ and $i = 1, \ldots, I$) is $r I$.
	On the other hand, the number of the linear equations in this region is $r (J - 1)$ since we have $J - 1$ equations for each $(s, m)$.
	Therefore, the system of the linear equations in this region is under-determining and thus the solution is not unique.

	Next, we suppose that $\{ \beta_1, \ldots, \beta_I \} \subset S$ and $I < J$.
	Excluding some elements of $S$ if necessary, we may assume that $J = I + 1$.
	Since $\{ \beta_1, \ldots, \beta_{J - 1} \} \subset S = \{ \alpha_1, \ldots, \alpha_J \}$, we may also assume that
	\[
		\alpha_1 = \beta_1, \ldots, \alpha_{J - 1} = \beta_{J - 1}.
	\]
	That is,
	\begin{itemize}
		\item 
		we focus on the singular values $\alpha_1, \ldots, \alpha_{J - 1}, \alpha_J$,

		\item
		a movable pattern contains none of these singular values unless its starting value is one of $\alpha_1, \ldots, \alpha_{J - 1}$,

		\item
		the equation has no movable pattern starting with $\alpha_J$.

	\end{itemize}	
	For $j < J$, the degree relation at $(t, n)$ between the values $\alpha_j$ and $\alpha_J$ is
	\begin{multline}\label{equation:degree_relation_solvable}
		\sum^{J - 1}_{i = 1} \sum_{(s, m) \in H} \operatorname{mult}^{\text{basic}, \alpha_i}_{t - s, n - m}(\alpha_j) Z_{s, m}(\alpha_i)
		+ \sum^L_{\ell = 1} \operatorname{mult}^{\gamma_{\ell}}_{t, n}(\alpha_j) \\
		= \sum^{J - 1}_{i = 1} \sum_{(s, m) \in H} \operatorname{mult}^{\text{basic}, \alpha_i}_{t - s, n - m}(\alpha_J) Z_{s, m}(\alpha_i)
		+ \sum^L_{\ell} \operatorname{mult}^{\gamma_{\ell}}_{t, n}(\alpha_J).
	\end{multline}
	Note that
	\[
		\operatorname{mult}^{\text{basic}, \alpha_i}_{0, 0}(\alpha_j) = \delta_{i, j}
	\]
	for $i = 1, \ldots, J - 1$ and $j = 1, \ldots, J$.
	Thus, while the right-hand side of \eqref{equation:degree_relation_solvable} contains none of $Z_{t, n}(\alpha_1), \ldots, Z_{t, n}(\alpha_{J - 1})$,
	the left-hand side is
	\[
		Z_{t, n}(\alpha_j) + \left( \text{linear combination of $Z_{s, m}(\alpha_i)$ ($(s, m) \in H_{\le (t, n)} \setminus \left( H_0 \cup \{ (t, n) \} \right) $)} \right).
	\]
	Therefore, the degree relation at $(t, n)$ between $\alpha_j$ ($j < J$) and $\alpha_J$ can be thought of as defining an evolution of $Z_{t, n}(\alpha_j)$.
	Hence, the system of linear equations has a unique solution.
\end{proof}

\begin{remark}
	Since $Z_{t, n}(\beta_i)$ are defined by divisors, the degree relations as a system of linear equations always have a solution.
	This is why we only considered the uniqueness of the solution in the proof.
\end{remark}

\begin{proof}[Proof of Proposition~\ref{proposition:solitary}]
	Let $x_{t, n}(z^{*}) = x^{*} \in \mathbb{C}$ be the first singularity of a pattern.
	
	First, we show that $x_{t, n + 1}(z^{*})$ depends on at least one of the variables $x_{t - t', n + 1}$ and $x_{t - 1 - t'', n}$, where $t'$ (resp.\ $t''$) is the least nonnegative integer satisfying $(t - t', n + 1) \in H_0$ (resp.\ $(t - 1 - t', n) \in H_0$).
	If $(t, n + 1) \in H_0$ (which can occur only if $(t, n) \in H_0$), then $x_{t, n + 1}$ itself is an initial variable (i.e., $t' = 0$) and there is nothing to prove.
	From here on, we consider the case where $(t, n + 1) \in H \setminus H_0$.
	By Proposition~\ref{proposition:key_lemma}, $x_{t - 1, n + 1}(z^{*})$ depends on $x_{t - t', n + 1}$.
	On the other hand, $x_{t - 1, n}(z) \in \mathbb{K}_{t - 1, n}(z)$ depends on $x_{t - t'', n}$, and Proposition~\ref{proposition:p1_decomposition} implies that $z^{*}$ is transcendental over $\mathbb{K}_{t - 1, n}$.
	Therefore, $x_{t - 1, n}(z^{*})$ depends on $x_{t - t'', n}$.
	Since $x_{t - 1, n + 1}(z^{*})$ and $x_{t - 1, n}(z^{*})$ depend on different initial variables, they are algebraically independent and we have
	\[
		x_{t, n + 1}(z^{*}) = \Phi(x_{t - 1, n + 1}(z^{*}), x^{*}, x_{t - 1, n}(z^{*})).
	\]
	Since it does not belong to $\mathbb{P}^1(\mathbb{C})$, $\Phi(B, x^{*}, D)$ depends on either $B$ or $D$.
	Substituting $x_{t - 1, n + 1}(z^{*})$ and $x_{t - 1, n}(z^{*})$ for $B$ and $D$ respectively, we conclude that $x_{t, n + 1}(z^{*})$ depends on either $x_{t - t', n + 1}$ or $x_{t - 1 - t'', n}$.
	Interchanging the $t$- and $n$- axes, one can show that $x_{t + 1, n}(z^{*})$ depends on either $x_{t + 1, n - n'}$ or $x_{t, n - 1 - n''}$, where $n'$ (resp.\ $n''$) is the least nonnegative integer satisfying $(t + 1, n - n') \in H_0$ (resp.\ $(t, n - 1 - n'') \in H_0$).

	Next, we show that $x_{t + 1, n + 1}(z^{*})$ is not a constant singularity.	
	Since $x_{t, n + 1}(z^{*})$ and $x_{t + 1, n}(z^{*})$ depend on different initial variables, they are algebraically independent over $\mathbb{C}$.
	Therefore, using $\Phi(B, C, x^{*}) \notin \mathbb{P}^1(\mathbb{C})$, we have
	\[
		x_{t + 1, n + 1}(z^{*}) = \Phi(x_{t, n + 1}(z^{*}), x_{t + 1, n}(z^{*}), x^{*}) \notin \mathbb{P}^1(\mathbb{C}).
	\]
	That is, $x_{t + 1, n + 1}(z^{*})$ is not a constant singularity.

	Since none of $x_{t, n + 1}(z^{*})$, $x_{t + 1, n + 1}(z^{*})$ and $x_{t + 1, n}(z^{*})$ is a constant singularity, it follows from Proposition~\ref{proposition:confining} that this pattern has no constant singularity other than $x_{t, n}(z^{*}) = x^{*}$, which completes the proof.
\end{proof}

\begin{proof}[Proof of Proposition~\ref{proposition:find_hidden_value}]
	Since $\deg_z x_{t, n}(z) \ge 1$, there exists $z^{*} \in \mathbb{P}^1 \left( \overline{\mathbb{K}} \right)$ such that $x_{t, n}(z^{*}) = \alpha$.
	Consider the pattern corresponding to $z = z^{*}$.
	If it is movable, then its starting point is not $(t, n)$ since its starting value is not $\alpha$.
	If the pattern is fixed, then the starting point is not $(t, n)$ since $(t, n) \in H \setminus H_0$.
	In both cases, the pattern has at least two constant singularities, $x_{t, n}(z^{*}) = \alpha$ and the starting point, and is not solitary.
\end{proof}

\begin{proof}[Proof of Proposition~\ref{proposition:check_hidden_value}]
	First, suppose that $x_{t, n}(z^{*}) = x^{*}$ is the first constant singularity of a movable pattern.
	We show that the relation $\Phi(B, C, D) = x^{*}$ has a non-trivial factor.
	By Corollary~\ref{corollary:not_zero_over_zero}, it is impossible for $x_{t, n}(z^{*})$ to become $x^{*}$ as $\frac{0}{0}$.
	Since none of $x_{t - 1, n}(z^{*})$, $x_{t, n - 1}(z^{*})$ and $x_{t - 1, n - 1}(z^{*})$ is $\infty$, there exists an irreducible factor $F \in \mathbb{C}[B, C, D]$ of the relation
	\[
		\Phi(B, C, D) = x^{*}
	\]
	such that
	\[
		F \left( x_{t - 1, n}(z^{*}), x_{t, n - 1}(z^{*}), x_{t - 1, n - 1}(z^{*}) \right) = 0.
	\]
	Since none of $x_{t - 1, n}(z^{*})$, $x_{t, n - 1}(z^{*})$ and $x_{t - 1, n - 1}(z^{*})$ belongs to $\mathbb{C}$, we have
	\[
		F \ne 
		B - \alpha, \
		C - \alpha, \
		D - \alpha \quad
		(\alpha \in \mathbb{C}).
	\]
	Therefore, $F$ is a non-trivial factor.
	
	Next, suppose that the relation $\Phi(B, C, D) = x^{*}$ has a non-trivial factor $F = F(B, C, D) \in \mathbb{C}[B, C, D]$.
	We show that the equation has a movable pattern starting with $x^{*}$.
	It follows from the $\partial$-factor condition that $F$ depends on both $B$ and $C$ since $F$ appears in the numerator of $\partial_B \Phi$ or $\partial_B \frac{1}{\Phi}$ (and $\partial_C \Phi$ or $\partial_C \frac{1}{\Phi}$) as a factor.
	Take a lattice point $(t, n) \in H \setminus H_0$ that satisfies
	\[
		\deg_z x_{t - 1, n}(z) \ge 1, \quad
		\deg_z x_{t_0, n - 1}(z) = \deg_z x_{t_0 + 1, n - 1}(z) = \cdots = \deg_z x_{t, n - 1}(z) = 0.
	\]
	Note that such a lattice point always exists and satisfies $(t, n) \ge (t_0, n_0)$.
	For example, we can take
	\[
		(t, n) = \begin{cases}
			(t_0 + 2, n_0 + 1) & ((t_0 + 1, n_0) \in H_0), \\
			(t_0 + 1, n_0) & ((t_0 + 1, n_0) \in H \setminus H_0).
		\end{cases}
	\]
	Since $x_{t, n - 1}(z)$ and $x_{t - 1, n - 1}(z)$ are independent of $z$, we simply write $x_{t, n - 1}(z^{*}) = x_{t, n - 1}$ and $x_{t - 1, n - 1}(z^{*}) = x_{t - 1, n - 1}$.

	Let $n'$ be the least nonnegative integer satisfying $(t, n - n') \in H_0$ and let $w = x_{t, n - n'}$.
	We think of $x_{t, n - 1}$ as a rational function in $w$, i.e., $x_{t, n - 1} = x_{t, n - 1}(w)$.
	Note that $x_{t - 1, n}$ and $x_{t - 1, n - 1}$ are independent of $w$.
	
	Consider the algebraic equation
	\[
		F(B, x_{t, n - 1}(w), x_{t - 1, n - 1}) = 0
	\]
	for $B$.
	Since $F(B, C, D)$ depends on both $B$ and $C$, at least one of the solutions $B = B^{*}$ must depend on $x_{t, n - 1}(w)$.
	That is, $B^{*}$ satisfies $F(B^{*}, x_{t, n - 1}(w), x_{t - 1, n - 1}) = 0$ and
	\[
		B^{*} \in \overline{\mathbb{C}(x_{t, n - 1}(w), x_{t - 1, n - 1})} \setminus \overline{\mathbb{C}(x_{t - 1, n - 1})}.
	\]
	In particular, $B^{*}$ depends on $w$ and thus is transcendental over $\mathbb{K}_{t - 1, n}$.

	Consider the algebraic equation
	\[
		x_{t - 1, n}(z) = B^{*}
	\]
	for $z$.
	Since $x_{t - 1, n}(z) \in \mathbb{K}_{t - 1, n}(z)$, $\deg_z x_{t - 1, n}(z) \ge 1$ and $B^{*}$ is transcendental over $\mathbb{K}_{t - 1, n}$, all the solutions to this algebraic equation depend on $B^{*}$.
	Since $B^{*}$ depends on $w$, so does the solution $z^{*}$.

	Let us show that $x_{t, n}(z^{*}) = x^{*}$.
	By construction, we have
	\[
		F(x_{t - 1, n}(z^{*}), x_{t, n - 1}, x_{t - 1, n - 1})
		= F(B^{*}, x_{t, n - 1}, x_{t - 1, n - 1})
		= 0.
	\]
	Since $F$ is an irreducible factor of the relation $\Phi(B, C, D) = x^{*}$, we have
	\[
		x_{t, n}(z^{*})
		= \Phi(x_{t - 1, n}(z^{*}), x_{t, n - 1}, x_{t - 1, n - 1})
		= x^{*}.
	\]

	Let us show that the first singularity of the pattern $z = z^{*}$ is $x_{t, n}(z^{*}) = x^{*}$.
	Since
	\[
		\deg_z x_{t_0, n - 1}(z) = \deg_z x_{t_0 + 1, n - 1}(z) = \cdots = \deg_z x_{t, n - 1}(z) = 0,
	\]
	it is sufficient to show that none of
	\[
		x_{t_0, n}(z^{*}), x_{t_0 + 1, n}(z^{*}), \ldots, x_{t - 1, n}(z^{*})
	\]
	is a constant singularity.
	Let $s < t$.
	Since $x_{s, n}(z) \in \mathbb{K}_{t - 1, n}(z)$ and $z^{*}$ is transcendental over $\mathbb{K}_{t - 1, n}$, we conclude that $x_{s, n}(z^{*})$ is not a constant singularity.
	Hence, the equation has a movable pattern whose first singularity is $x_{t, n}(z^{*}) = x^{*}$.
\end{proof}

\section{Examples}\label{section:examples}

In this section, we analyze several examples to see how to compute exact degrees with our theory.

\begin{example}\label{example:polynomial}
	The simplest class of examples are polynomial equations, i.e.,
	\[
		x_{t, n} = \Phi(x_{t - 1, n}, x_{t, n - 1}, x_{t - 1, n - 1})
	\]
	where $\Phi(B, C, D)$ is a polynomial in $B$, $C$ and $D$.
	In this case, $\infty$ cannot appear in any movable pattern.
	Moreover, only the fixed pattern corresponding to $z = \infty$ contains $\infty$ as a value.
	Therefore, by letting $S = \{ \infty \}$ and calculating the multiplicity of $x_{t, n} = \infty$ for this fixed pattern, one can compute the degrees.
	This result does not rely on the form of the equation except that it is defined by a polynomial.
	Neither the $\partial$-factor condition nor the basic pattern condition is required.
	Note that computing exact degrees for polynomial mappings by Halburd's method is already mentioned in the original paper \cite{halburd}.
\end{example}

\begin{example}\label{example:polynomial_concrete}
	Let us consider the polynomial equation
	\[
		x_{t, n} = \Phi \left( x_{t - 1, n}, x_{t, n - 1}, x_{t - 1, n - 1} \right)
		= x_{t, n - 1} (x_{t - 1, n} - x_{t - 1, n - 1}).
	\]
	Let $X_{t, n}$ be the multiplicity of $x_{t, n}(\infty) = \infty$, i.e.,
	\[
		x_{t, n} \left( \varepsilon^{- 1} \right) = y_{t, n} \varepsilon^{- X_{t, n}} + O \left( \varepsilon^{- X_{t, n} + 1} \right)
	\]
	with $y_{t, n} \ne 0$.
	It follows from the discussion in Example~\ref{example:polynomial} that
	\[
		\deg_z x_{t, n}(z) = X_{t, n}.
	\]
	
	First, we show that $X_{t, n}$ satisfies the ultra-discrete version of the equation:
	\[
		X_{t, n} = X_{t, n - 1} + \max \left( X_{t - 1, n}, X_{t - 1, n - 1} \right).
	\]
	The fact that $X_{t, n}$ satisfies this ultra-discrete form is not immediate, because the original equation is not subtraction-free.
	We now explain why this holds.

	If $X_{t - 1, n} \ne X_{t - 1, n - 1}$, then the ultra-discrete equation holds for $(t, n)$ since their leading orders differ.
	Furthermore, since $x_{t, n} \left( \varepsilon^{- 1} \right)$ is a polynomial in $\varepsilon^{- 1}$, we have
	\[
		X_{t, n} \ge X_{t, n - 1}.
	\]
	Therefore, using the coefficient $y_{t, n}$ appearing in the above expansion, we have
	\[
		x_{t, n} \left( \varepsilon^{- 1} \right) =
		\begin{cases}
			y_{t, n - 1} y_{t - 1, n} \varepsilon^{- X_{t, n - 1} - X_{t - 1, n}} + (\text{lower}) & (X_{t, n} > X_{t, n - 1}), \\
			y_{t, n - 1} (y_{t - 1, n} - y_{t - 1, n - 1}) \varepsilon^{- X_{t, n - 1} - X_{t - 1, n}} + (\text{lower}) & (X_{t, n} = X_{t, n - 1}),
		\end{cases}
	\]
	where ``$(\text{lower})$'' refers to terms of lower order in the expansion with respect to $\varepsilon^{- 1}$.
	It can be shown by induction on $(t, n)$ that
	\begin{itemize}
		\item
		$y_{t, n}$ depends on the initial variables $x_{t - t', n}$ and $x_{t, n - n'}$ where $t'$ (resp.\ $n'$) is the least positive integer satisfying $(t - t', n) \in H_0$ (resp.\ $(t, n - n') \in H_0$), i.e., $y_{t, n} \in \Delta_{t, n}$ (see Definition~\ref{definition:delta}),

		\item
		$y_{t, n} =
		\begin{cases}
			y_{t, n - 1} y_{t - 1, n} & (X_{t, n} > X_{t, n - 1}), \\
			y_{t, n - 1} (y_{t - 1, n} - y_{t - 1, n - 1}) & (X_{t, n} = X_{t, n - 1}).
		\end{cases}$

	\end{itemize}
	The key point is that $y_{t - 1, n} - y_{t - 1, n - 1} \ne 0$ since $y_{t - 1, n}$ depends on the initial variable $x_{t - t', n}$ but $y_{t - 1, n - 1}$ does not.
	Therefore, $X_{t, n}$ indeed satisfies the ultra-discrete form.
	Moreover, because $X_{t - 1, n} \ge X_{t - 1, n - 1}$ for all $(t, n)$, the ultra-discrete equation simplifies to
	\[
		X_{t, n} = X_{t, n - 1} + X_{t - 1, n}.
	\]

	Next, we give the solution to the ultra-discrete equation for the following choice of $H$ and $z$:
	\[
		H = \mathbb{Z}^2_{\ge 0}, \quad
		z = x_{1, 0}.
	\]
	Since $X_{t, n} = X_{t, n - 1} + X_{t - 1, n}$, the exact solution is given by Pascal's triangle and is expressed as
	\[
		X_{t, n} = \frac{(t + n - 1) !}{(t - 1) ! \ n !} \quad
		(t, n \ge 1),
	\]
	which yields the following values:
	\[
		X_{t, n} = \deg_z x_{t, n}(z) \colon \quad
		\begin{matrix}
			{\textcolor{red} 0} & 1 & 7 & 28 & 84 & 210 & 462 & 924 \\
			{\textcolor{red} 0} & 1 & 6 & 21 & 56 & 126 & 252 & 462 \\
			{\textcolor{red} 0} & 1 & 5 & 15 & 35 & 70 & 126 & 210 \\
			{\textcolor{red} 0} & 1 & 4 & 10 & 20 & 35 & 56 & 84 \\
			{\textcolor{red} 0} & 1 & 3 & 6 & 10 & 15 & 21 & 28 \\
			{\textcolor{red} 0} & 1 & 2 & 3 & 4 & 5 & 6 & 7 \\
			{\textcolor{red} 0} & 1 & 1 & 1 & 1 & 1 & 1 & 1 \\
			{\textcolor{red} 0} & {\textcolor{red} 1} & {\textcolor{red} 0} & {\textcolor{red} 0} & {\textcolor{red} 0} & {\textcolor{red} 0} & {\textcolor{red} 0} & {\textcolor{red} 0}
		\end{matrix}.
	\]
	The degree exhibits exponential growth, as indicated by the binomial coefficient formula.
\end{example}

\begin{example}\label{example:laurent_property}
	The second simplest class of examples would be equations with the Laurent property.
	An equation (or an initial value problem for an equation) is said to have the Laurent property if every iterate can be expressed as a Laurent polynomial in the initial variables \cite{laurent_phenomenon}.
	Since the denominator of each iterate is a monomial, $\infty$ cannot appear in any movable pattern and only the fixed patterns corresponding to $z = 0, \infty$ can contain $\infty$ as a value.
	Therefore, by letting $S = \{ \infty \}$ and calculating the multiplicity of $x_{t, n} = \infty$ for these two fixed patterns, one can compute the degrees.
	As in the case of polynomial mappings, this result does not rely on the form of the equation, either.
\end{example}

\begin{example}\label{example:laurent_property_concrete}
	Let us consider the discrete Liouville equation
	\[
		x_{t, n} = \Phi \left( x_{t - 1, n}, x_{t, n - 1}, x_{t - 1, n - 1} \right)
		= \frac{x_{t - 1, n} x_{t, n - 1} + 1}{x_{t - 1, n - 1}}.
	\]
	As discussed in Example~\ref{example:laurent_property}, because the equation has the Laurent property, we can compute the degrees by calculating the fixed patterns corresponding to $z = 0, \infty$.

	Since the equation is subtraction-free, these patterns can be obtained via ultra-discretization.
	That is, the leading order of the expansion of $x_{t, n} \left( \varepsilon \right)$ (or $x_{t, n} \left( \varepsilon^{- 1} \right)$) with respect to $\varepsilon^{- 1}$, denoted by $X_{t, n}$, satisfies the ultra-discrete form
	\[
		X_{t, n} = \max \left( X_{t - 1, n} + X_{t, n - 1}, 0 \right) - X_{t - 1, n - 1}.
	\]

	Let us fix $H$ and $z$ as
	\[
		H = \mathbb{Z}^2_{\ge 0} \setminus \{ (0, 0) \}, \quad
		z = x_{1, 1}.
	\]
	Solving the ultra-discrete equation yields the following patterns:
	\[
		\begin{matrix}
			{\color{red} \vdots} & \vdots & \vdots & \vdots & \iddots \\
			{\color{red} x_{0, 3}} & \text{reg} & \infty & \infty & \ldots \\
			{\color{red} x_{0, 2}} & \text{reg} & \infty & \infty & \ldots \\
			{\color{red} x_{0, 1}} & {\color{red} x_{1, 1} = 0} & \text{reg} & \text{reg} & \ldots \\
			& {\color{red} x_{1, 0}} & {\color{red} x_{2, 0}} & {\color{red} x_{3, 0}} & {\color{red} \ldots}
		\end{matrix}
	\]
	and
	\[
		\begin{matrix}
			{\color{red} \vdots} & \vdots & \vdots & \vdots & \iddots \\
			{\color{red} x_{0, 3}} & \infty & \infty & \infty & \ldots \\
			{\color{red} x_{0, 2}} & \infty & \infty & \infty & \ldots \\
			{\color{red} x_{0, 1}} & {\color{red} x_{1, 1} = \infty} & \infty & \infty & \ldots \\
			& {\color{red} x_{1, 0}} & {\color{red} x_{2, 0}} & {\color{red} x_{3, 0}} & {\color{red} \ldots}
		\end{matrix}.
	\]
	Since the degree coincides with the sum of the multiplicities of $x_{t, n} = \infty$ in these two patterns, we have
	\[
		\deg_z x_{t, n}(z) = \begin{cases}
			0 & (\min (t, n) = 0), \\
			1 & (\min (t, n) = 1), \\
			2 & (\min (t, n) \ge 2).
		\end{cases}
	\]
	This result is in perfect agreement with the degrees previously calculated in \cite{domain}.
\end{example}

\begin{example}\label{example:nonintegrable_confining}
	Let us calculate the degrees of the equation
	\[
		x_{t, n} = - x_{t - 1, n - 1} + \frac{a}{x^k_{t - 1, n}} + \frac{b}{x^k_{t, n - 1}},
	\]
	where $a$ and $b$ are non-zero parameters and $k$ is a positive even integer.
	This equation was discovered by the present author and reported in detail with collaborators as the first-ever example of a non-integrable lattice equation that passes the singularity confinement test \cite{2dchaos,exkdv}.
	Let $z = x_{t_0, n_0}$ be one of the initial variables and let
	\[
		\Phi = \Phi(B, C, D) = - D + \frac{a}{B^k} + \frac{b}{C^k}.
	\]
	Directly calculating $\partial_B \Phi$, $\partial_B \frac{1}{\Phi}$, $\partial_C \Phi$ and $\partial_C \frac{1}{\Phi}$, one can easily show that $\Phi$ satisfies the $\partial$-factor condition.
	
	Looking at the equation carefully, we notice that it is not easy for $x_{t, n}$ to become $\infty$.
	In fact, the relation $\Phi(B, C, D) = \infty$ can be written as
	\[
		B^k C^k = 0.
	\]
	Therefore, by Proposition~\ref{proposition:check_hidden_value}, the value $\infty$ does not appear as the first singularity of any movable pattern.

	Let us check the basic pattern condition.
	Since the equation does not have a movable pattern starting with $\infty$, we do not need to check the basic pattern condition for $\infty$.
	Let $x^{*} \in \mathbb{C}$.
	The relation $\Phi(B, C, D) = x^{*}$ is written as
	\[
		- B^k C^k D + a C^k + b B^k - x^{*} B^k C^k = 0.
	\]
	The left-hand side is irreducible since it is degree $1$ with respect to $D$ and the coefficients of $D^1$ and $D^0$ share no irreducible factor.
	Therefore, the equation satisfies the basic pattern condition for $x^{*}$.

	To compute the degrees, we need to calculate all the constant singularity patterns that contain $\infty$ as a value.
	It follows from Theorem~\ref{theorem:main_theorem} that the movable pattern starting with $x_{t, n}(z^{*}) = x^{*} \in \mathbb{C}$ coincides with the $(t, n)$-translation of the basic pattern corresponding to $x^{\text{basic}}_{0, 0} = x^{*}$.
	Since
	\[
		\partial_B \Phi = - \frac{k a}{B^{k + 1}}, \quad
		\partial_C \Phi = - \frac{k b}{C^{k + 1}}, \quad
		\partial_D \Phi = - 1,
	\]
	it follows from Proposition~\ref{proposition:solitary} that any pattern whose starting value is neither $0$ nor $\infty$ is solitary.
	Therefore, it is sufficient to calculate the basic pattern corresponding to $x^{\text{basic}}_{0, 0} = 0$ and the fixed patterns corresponding to $z = 0, \infty$.

	First, we calculate the basic pattern corresponding to $x^{\text{basic}}_{0, 0} = 0$.
	Let $x^{\text{basic}}_{0, 0} = \varepsilon$.
	Then, a direct calculation leads to
	\[
		\begin{matrix}
			{\color{red} x^{\text{basic}}_{- 1, 2}} & \text{reg} & \text{reg} & \text{reg} \\
			{\color{red} x^{\text{basic}}_{- 1, 1}} & \sim \varepsilon^{- k} & \sim \varepsilon & \text{reg} \\
			{\color{red} x^{\text{basic}}_{- 1, 0}} & {\color{red} x^{\text{basic}}_{0, 0} = \varepsilon} & \sim \varepsilon^{- k} & \text{reg} \\
			& {\color{red} x^{\text{basic}}_{0, - 1}} & {\color{red} x^{\text{basic}}_{1, - 1}} & {\color{red} x^{\text{basic}}_{2, - 1}}
		\end{matrix},
	\]
	where ``$\text{reg}$'' means that it does not have a constant singularity at $\varepsilon = 0$.
	By Proposition~\ref{proposition:confining}, this pattern does not have a constant singularity outside the above region.
	Therefore, by Theorem~\ref{theorem:main_theorem}, the movable pattern starting with $x_{t, n}(z^{*}) = 0$ is
	\[
		\begin{matrix}
			x_{t, n + 1} = \infty^k & x_{t + 1, n + 1} = 0 \\
			x_{t, n} = 0 & x_{t + 1, n} = \infty^k.
		\end{matrix}
	\]
	
	Next, we calculate the fixed pattern corresponding to $z = \infty$.
	Regardless of the shape of the domain, the fixed pattern corresponding to $z = \infty$ is
	\[
		\begin{matrix}
			& & & \iddots \\
			& & \infty \\
			& \infty \\
			{\color{red} z = \infty}
		\end{matrix}.
	\]
	Let us check this calculation.
	A key idea is to show that if $x_{t, n}(z^{*}) = \infty$, then $x_{t + 1, n}(z^{*})$ does depend on the initial variable $w = x_{t + 1, n - n'}$ where $n'$ is the least nonnegative integer that satisfies $(t + 1, n - n') \in H_0$:
	\[
		\begin{matrix}
			& & & & \infty & x_{t + 1, n} \\
			& & & \infty & x_{t, n - 1} & x_{t + 1, n - 1} \\
			& & \iddots & & & \vdots \\
			& \infty & & & & \vdots \\
			{\color{red} z = \infty} & & & & & {\color{red} w}
		\end{matrix}.
	\]
	It follows from Proposition~\ref{proposition:key_lemma} that $x_{t + 1, n - 1}(z^{*})$ depends on $w$.
	Since $x_{t, n - 1}(z^{*})$ is not a constant singularity but $x_{t + 1, n - 1}(z^{*})$ depends on $w$, $x_{t, n - 1}(z^{*})$ and $x_{t + 1, n - 1}(z^{*})$ are algebraically independent over $\mathbb{C}$.
	Therefore, it follows from Lemma~\ref{lemma:substitution_trdeg_codimension_one} that $x_{t + 1, n}(z^{*})$ can be directly calculated from $x_{t, n}(z^{*})$, $x_{t + 1, n - 1}(z^{*})$ and $x_{t, n - 1}(z^{*})$ as
	\[
		x_{t + 1, n}(z^{*}) = \Phi(x_{t, n}(z^{*}), x_{t + 1, n - 1}(z^{*}), x_{t, n - 1}(z^{*})).
	\]
	A direct calculation shows that $x_{t + 1, n}(z^{*})$ depends on $w$, too.
	Interchanging the $t$- and $n$- axes, one can show that $x_{t, n + 1}(z^{*})$ depends on some initial variable $w' = x_{t - t', n + 1}$.
	Since $x_{t + 1, n}(z^{*})$ (resp.\ $x_{t, n + 1}(z^{*})$) is independent of $w'$ (resp.\ $w$), $x_{t + 1, n}(z^{*})$ and $x_{t, n + 1}(z^{*})$ are algebraically independent over $\mathbb{C}$ and can be thought of as if they were initial variables when calculating $x_{t + 1, n + 1}(z^{*})$ from $x_{t, n + 1}(z^{*})$, $x_{t + 1, n}(z^{*})$ and $x_{t, n}(z^{*})$.
	Therefore, a direct calculation leads to $x_{t + 1, n + 1}(z^{*}) = \infty$.
	Repeating this procedure, one obtains the whole pattern.
	
	On the other hand, the fixed pattern corresponding to $z = 0$ varies depending on the shape of the domain, especially on whether or not $(t_0 + 1, n_0)$ and $(t_0, n_0 + 1)$ belong to $H \setminus H_0$, respectively.
	If neither $(t_0 + 1, n_0)$ nor $(t_0, n_0 + 1)$ belongs to $H_0$, then the pattern is
	\[
		\begin{matrix}
			\infty^k & 0 \\
			{\color{red} z = 0} & \infty^k
		\end{matrix},
	\]
	which has the same shape as a movable pattern and is confining.
	If only one of $(t_0 + 1, n_0)$ and $(t_0, n_0 + 1)$ belongs to $H_0$, then the pattern is
	\[
		\begin{matrix}
			& & & \iddots \\
			& & \infty^k \\
			& \infty^k & \\
			\infty^k & & \\
			{\color{red} z = 0} & {\color{red} x_{t_0 + 1, n_0}}
		\end{matrix}
	\]
	or
	\[
		\begin{matrix}
			& & & & \iddots \\
			& & & \infty^k \\
			{\color{red} x_{t_0, n_0 + 1}} & & \infty^k \\
			{\color{red} z = 0} & \infty^k
		\end{matrix},
	\]
	respectively.
	If both $(t_0 + 1, n_0)$ and $(t_0, n_0 + 1)$ belong to $H_0$, then the pattern is solitary.
	One can calculate these patterns rigorously in a similar way to the above.

	Before computing the degrees, let us fix a domain $H$, for example, as follows:
	\[
		H = \mathbb{Z}^2_{\ge (0, 0)}, \quad
		z = x_{1, 0}.
	\]
	Then, by Theorem~\ref{theorem:main_degree}, we obtain the following degree relations for $(t, n) \in H$:
	\begin{align*}
		\deg_z x_{t, n}(z) &= Z_{t, n}(0) + Z_{t - 1, n - 1}(0) + \delta^{t, 1}_{n, 0} & (\text{$\#$ of preimages of $0$}) \\
		&= k Z_{t - 1, n}(0) + k Z_{t, n - 1}(0) + \delta_{t - 1, n} + k (\delta_{t, n} - \delta^{t, 0}_{n, 0}) & (\text{$\#$ of preimages of $\infty$}),
	\end{align*}
	where $\delta^{a, b}_{c, d} = \delta_{a, b} \delta_{c, d}$ and $Z_{t, n}(0) = 0$ for $(t, n) \in \mathbb{Z}^2 \setminus (H \setminus H_0)$.
	From the analysis of the linear equation
	\[
		Z_{t, n}(0) = k Z_{t - 1, n}(0) + k Z_{t, n - 1}(0) - Z_{t - 1, n - 1}(0) - \delta^{t, 1}_{n, 0} + \delta_{t - 1, n} + k (\delta_{t, n} - \delta^{t, 0}_{n, 0}),
	\]
	one finds that the degree growth is exponential in any direction in the first quadrant.
	If $k = 2$, for example, the following are the results of iterating the recursion over a finite range:
	\[
		Z_{t, n}(0) \colon \quad
		\begin{matrix}
			{\textcolor{red} 0}& 128 & 1536 & 10816 & 58336 & 266592 & 1086392 & 4067660 \\
			{\textcolor{red} 0} & 64 & 672 & 4208 & 20456 & 85188 & 319198 & 1107036 \\
			{\textcolor{red} 0} & 32 & 288 & 1576 & 6808 & 25542 & 87181 & 277910 \\
			{\textcolor{red} 0} & 16 & 120 & 560 & 2108 & 7016 & 21556 & 62552 \\
			{\textcolor{red} 0} & 8 & 48 & 184 & 585 & 1692 & 4608 & 12024 \\
			{\textcolor{red} 0} & 4 & 18 & 52 & 134 & 328 & 776 & 1792 \\
			{\textcolor{red} 0} & 2 & 5 & 10 & 20 & 40 & 80 & 160 \\
			{\textcolor{red} 0} & {\textcolor{red} 0} & {\textcolor{red} 0} & {\textcolor{red} 0} & {\textcolor{red} 0} & {\textcolor{red} 0} & {\textcolor{red} 0} & {\textcolor{red} 0}
		\end{matrix},
	\]
	\[
		\deg_z x_{t, n}(z) \colon \quad
		\begin{matrix}
			{\textcolor{red} 0} & 128 & 1600 & 11488 & 62544 & 287048 & 1171580 & 4386858 \\
			{\textcolor{red} 0} & 64 & 704 & 4496 & 22032 & 91996 & 344740 & 1194217 \\
			{\textcolor{red} 0} & 32 & 304 & 1696 & 7368 & 27650 & 94197 & 299466 \\
			{\textcolor{red} 0} & 16 & 128 & 608 & 2292 & 7601 & 23248 & 67160 \\
			{\textcolor{red} 0} & 8 & 52 & 202 & 637 & 1826 & 4936 & 12800 \\
			{\textcolor{red} 0} & 4 & 20 & 57 & 144 & 348 & 816 & 1872 \\
			{\textcolor{red} 0} & 2 & 5 & 10 & 20 & 40 & 80 & 160 \\
			{\textcolor{red} 0} & {\textcolor{red} 1} & {\textcolor{red} 0} & {\textcolor{red} 0} & {\textcolor{red} 0} & {\textcolor{red} 0} & {\textcolor{red} 0} & {\textcolor{red} 0}
		\end{matrix}.
	\]
	In this paper, we do not address the problem of solving the linear recurrence for $Z_{t, n}(0)$, i.e., finding its exact solution, because this question is not the main focus of the present work and lies within the domain of linear partial difference equations.
\end{example}

\begin{example}\label{example:first_example_2d_rigorous}
	Let us see the calculations in Example~\ref{example:first_example_2d} from a point of view of our theory.
	One can check that all the calculations are justified by our theory.
	For example, using our theory, even those without any knowledge of singularity analysis can rigorously calculate the fixed pattern corresponding to $z = \infty$.
	Since the discussion is the same as in Example~\ref{example:nonintegrable_confining}, we omit the details here.
	
	It should be noted, however, that there is a small difference between the $Z_{t, n}$ in Example~\ref{example:first_example_2d} and the $Z_{t, n}(0)$ defined in \textsection\ref{section:main} and \textsection\ref{section:proof}.
	The difference arises in the fixed pattern starting with $x_{1, 1}(z^{*}) = 0$, i.e., corresponding to $z = 0$.
	While $Z_{1, 1}$ contains the information on this pattern, $Z_{1, 1}(0)$ does not.
	This pattern behaves as if it were a movable pattern but is labeled fixed by the strict definition of singularity patterns (Definition~\ref{definition:constant_singularity}).
	Therefore, with the notation in Theorem~\ref{theorem:main_degree}, the degree relations are now written for $(t, n) \in H$ as
	\begin{align*}
		\deg_z x_{t, n}(z) &= Z_{t, n}(0) + Z_{t - 1, n - 1}(0) + \delta^{t, 1}_{n, 1} + \delta^{t, 2}_{n, 2} \\
		&= Z_{t - 1, n}(0) + Z_{t, n - 1}(0) + \delta^{t, 1}_{n, 2} + \delta^{t, 2}_{n, 1} + \delta_{t, n},
	\end{align*}
	where $Z_{t, n}(0) = 0$ for $(t, n) \in \mathbb{Z} \setminus (H \setminus H_0)$.
	Although these expressions look more complicated than those in Example~\ref{example:first_example_2d}, the initial values for $Z_{t, n}(0)$ are simpler in this notation.
\end{example}

\begin{example}
	Let us consider the discrete modified KdV equation \cite{levi_yamilov_mkdv}
	\[
		(1 - x_{t - 1, n - 1} x_{t - 1, n}) \left( k x_{t, n} - \frac{x_{t, n -1}}{k} \right) = (1 - x_{t, n - 1} x_{t - 1, n - 1}) \left( k x_{t - 1, n - 1} - \frac{x_{t - 1, n}}{k} \right),
	\]
	where $k \in \mathbb{C}$ satisfies $k \ne 0$ and $k^4 \ne 1$.
	Let
	\[
		\Phi(B, C, D) = \frac{D (B C - k^2) + B - C}{k^2 D (B - C) + B C - k^2}
	\]
	and let $z = x_{t_0, n_0}$ be one of the initial variables.
	Using $\Phi$, we can write the equation as
	\[
		x_{t, n} = \Phi(x_{t - 1, n}, x_{t, n - 1}, x_{t - 1, n - 1}).
	\]
	The singularity structure of this equation has been analyzed in \cite{sc_mkdv1,sc_mkdv2}, which implies that $\pm k$ and $\pm \frac{1}{k}$ are important values.
	However, for a reader who is not familiar with singularity analysis, we start with detecting singularities.

	First, we calculate the derivatives:
	\begin{align*}
		\partial_B \Phi(B, C, D) &= - \frac{(k D - 1) (k D + 1) (C - k) (C + k)}{(k^2 D (B - C) + B C - k^2)^2}, \\
		\partial_B \frac{1}{\Phi(B, C, D)} &= \frac{(k D - 1) (k D + 1) (C - k) (C + k)}{(D (B C - k^2) + B - C)^2}, \\
		\partial_C \Phi(B, C, D) &= \frac{(k D - 1) (k D + 1) (B - k) (B + k)}{(k^2 D (B - C) + B C - k^2)^2}, \\
		\partial_C \frac{1}{\Phi(B, C, D)} &= - \frac{(k D - 1) (k D + 1) (B - k) (B + k)}{(D (B C - k^2) + B - C)^2}, \\
		\partial_D \Phi(B, C, D) &= \frac{(B - k) (B + k) (C - k) (C + k)}{(k^2 D (B - C) + B C - k^2)^2}, \\
		\partial_D \frac{1}{\Phi(B, C, D)} &= - \frac{(B - k) (B + k) (C - k) (C + k)}{(D (B C - k^2) + B - C)^2},
	\end{align*}
	which implies that the equation satisfies the $\partial$-factor condition.
	Moreover, by Proposition~\ref{proposition:solitary}, one finds that any fixed pattern whose starting value is neither $\pm k$ nor $\pm \frac{1}{k}$ is solitary, which is why we focus on these four values.
	The relation $\Phi(B, C, D) = \pm \frac{1}{k}$ can be written as
	\[
		\frac{(B \pm k) (C \mp k) ( k D \mp 1)}{k} = 0,
	\]
	respectively.
	Therefore, by Proposition~\ref{proposition:check_hidden_value}, the values $\pm \frac{1}{k}$ do not appear as the starting value of any movable pattern.

	Let us check the basic pattern condition for $x^{*} \in \mathbb{P}^1(\mathbb{C}) \setminus \left\{ \pm \frac{1}{k} \right\}$.
	First, consider the case where $x^{*} \ne \infty$.
	The relation $\Phi(B, C, D) = x^{*}$ is written as
	\[
		(B C - x^{*} k^2 B + x^{*} k^2 C - k^2) D + (- x^{*} B C + B - C + x^{*} k^2) = 0.
	\]
	Let $f \in \mathbb{C}[B, C, D]$ be the left-hand side and let $f_0$ and $f_1$ be the coefficients of $D^0$ and $D^1$, respectively, i.e.,
	\begin{align*}
		f_0 &= - x^{*} B C + B - C + x^{*} k^2, \\
		f_1 &= B C - x^{*} k^2 B + x^{*} k^2 C - k^2.
	\end{align*}
	Since it is degree $1$ with respect to $D$, $f$ is irreducible unless $f_0$ and $f_1$ share an irreducible factor.
	Seeking for a contradiction, assume that $f_1$ and $f_2$ share an irreducible factor $g$.
	Since
	\[
		f_0 + x^{*} f_1 = - (k x^{*} - 1) (k x^{*} + 1) (B - C)
	\]
	and
	\[
		(k x^{*} - 1) (k x^{*} + 1) \ne 0,
	\]
	we have
	\[
		g = B - C.
	\]
	However, $f_1$ is not divisible by $B - C$, which is a contradiction.
	Therefore, $f$ is irreducible and thus the equation satisfies the basic pattern condition for $x^{*}$.
	
	Next, we consider the case $x^{*} = \infty$.
	In this case, the relation $\Phi(B, C, D) = \infty$ is
	\[
		k^2 D (B - C) + B C - k^2 = 0.
	\]
	The left-hand side is irreducible since it is degree $1$ with respect to $D$ and the coefficients of $D^1$ and $D^0$ do not share an irreducible factor.
	Therefore, the equation satisfies the basic pattern condition for $\infty$, too.

	We calculate movable patterns starting with $\pm k$.
	A direct calculation shows that the basic patterns corresponding to $x^{\text{basic}}_{0, 0} = \pm k$ are
	\[
		\begin{matrix}
			{\color{red} \bullet} & \circ & \circ & \circ \\
			{\color{red} \bullet} & \pm \frac{1}{k} & \mp k & \circ \\
			{\color{red} \bullet} & {\color{red} x^{\text{basic}}_{0, 0} = \pm k} & \mp \frac{1}{k} & \circ \\
			& {\color{red} \bullet} & {\color{red} \bullet} & {\color{red} \bullet}
		\end{matrix},
	\]
	respectively,
	where the multiplicities of these four constant singularities are all $1$.
	Therefore, the corresponding movable patterns are
	\[
		\begin{matrix}
			x_{t, n + 1} = \pm \frac{1}{k} & x_{t + 1, n + 1} = \mp k \\
			x_{t, n} = \pm k & x_{t + 1, n} = \mp \frac{1}{k}
		\end{matrix}.
	\]

	Let us calculate the fixed patterns containing $\pm k$ or $\pm \frac{1}{k}$.
	Since any fixed pattern corresponding to $z = \gamma \notin \left\{ \pm k, \pm \frac{1}{k} \right\}$ is solitary, it is sufficient to calculate the fixed patterns corresponding to $z = \pm k$, $\pm \frac{1}{k}$.
	A calculation similar to that in Example~\ref{example:nonintegrable_confining} shows that the fixed patterns corresponding to $z = \pm \frac{1}{k}$ do not depend on the shape of $H$:
	\[
		\begin{matrix}
			& & & \iddots \\
			& & \pm \frac{1}{k} \\
			& \pm \frac{1}{k} \\
			{\color{red} z = \pm \frac{1}{k}}
		\end{matrix}.
	\]
	On the other hand, the fixed patterns corresponding to $z = \pm k$ vary depending on whether $(t_0 + 1, n_0)$ and $(t_0, n_0 + 1)$ belong to $H_0$, respectively:
	\begin{itemize}
		\item 
		$\begin{matrix}
			\pm \frac{1}{k} & \mp k \\
			{\color{red} z = \pm k} & \mp \frac{1}{k}
		\end{matrix}$
		\quad
		(confining),

		\item
		$\begin{matrix}
			& & & \iddots \\
			& & \pm \frac{1}{k} \\
			& \pm \frac{1}{k} \\
			\pm \frac{1}{k} \\
			{\color{red} z = \pm k} & {\color{red} x_{t_0 + 1, n_0}}
		\end{matrix}$
		\quad
		(non-confining),

		\item
		$\begin{matrix}
			& & & & \iddots \\
			& & & \mp \frac{1}{k} \\
			{\color{red} x_{t_0, n_0 + 1}} & & \mp \frac{1}{k} \\
			{\color{red} z = \pm k} & \mp \frac{1}{k}
		\end{matrix}$
		\quad
		(non-confining),

		\item
		$\begin{matrix}
			{\color{red} x_{t_0, n_0 + 1}} & \text{reg} \\
			{\color{red} z = \pm k} & {\color{red} x_{t_0 + 1, n_0}}
		\end{matrix}$
		\quad
		(solitary).

	\end{itemize}

	Let $H$ and $z$ be the same as in Example~\ref{example:nonintegrable_confining}.
	Then, the degree relations for $(t, n) \in H$ are
	\begin{align*}
		\deg_z x_{t, n}(z) &= Z_{t, n}(k) + Z_{t - 1, n - 1}( - k ) + \delta^{t, 1}_{n, 0} \\
		&= Z_{t, n}( - k ) + Z_{t - 1, n - 1}( k ) + \delta^{t, 1}_{n, 0} \\
		&= Z_{t - 1, n}( - k ) + Z_{t, n - 1}( k ) + \delta_{t - 1, n} + (\delta_{t, n} - \delta^{t, 0}_{n, 0}) \\
		&= Z_{t - 1, n}( k ) + Z_{t, n - 1}( - k ) + \delta_{t - 1, n} + (\delta_{t, n} - \delta^{t, 0}_{n, 0}),
	\end{align*}
	where $Z_{t, n}( \pm k ) = 0$ for $(t, n) \in \mathbb{Z}^2 \setminus (H \setminus H_0)$.
	Each line represents the degree calculated as the number of preimages for $k$, $- k$, $\frac{1}{k}$ and $- \frac{1}{k}$, respectively.
	It follows from the degree relations that
	\[
		Z_{t, n}(k) = Z_{t, n}( - k )
	\]
	for all $(t, n)$.
	Therefore, $Z_{t, n}(k)$ satisfies
	\[
		Z_{t, n}(k) = Z_{t - 1, n}(k) + Z_{t, n - 1}(k) - Z_{t - 1, n - 1}(k) + \delta_{t - 1, n} + (\delta_{t, n} - \delta^{t, 0}_{n, 0}) - \delta^{t, 1}_{n, 0}
	\]
	for $(t, n) \in H$, which implies polynomial degree growth.
	Solving this recursion, one obtains $Z_{t, n}( \pm k )$ and the degrees as
	\[
		Z_{t, n}( \pm k ) = \begin{cases}
			2 n & (t > n), \\
			2 t - 1 & (0 < t \le n), \\
			0 & (t = 0),
		\end{cases} \quad
		\deg_z x_{t, n}(z) = \begin{cases}
			4 n - 2 & (t > n; \ t \ge 2), \\
			4 t - 4 & (1 < t \le n), \\
			1 & (t = 1), \\
			0 & (t = 0), \\
			0 & (t \ge 2; \ n = 0).
		\end{cases}
	\]
	The following are the results of iterating the equation over a finite range:
	\[
		Z_{t, n}( \pm k ) \colon \quad
		\begin{matrix}
			{\textcolor{red} 0} & 1 & 3 & 5 & 7 & 9 & 11 & 13 & 15 & 17 & 19 \\
			{\textcolor{red} 0} & 1 & 3 & 5 & 7 & 9 & 11 & 13 & 15 & 17 & 18 \\
			{\textcolor{red} 0} & 1 & 3 & 5 & 7 & 9 & 11 & 13 & 15 & 16 & 16 \\
			{\textcolor{red} 0} & 1 & 3 & 5 & 7 & 9 & 11 & 13 & 14 & 14 & 14 \\
			{\textcolor{red} 0} & 1 & 3 & 5 & 7 & 9 & 11 & 12 & 12 & 12 & 12 \\
			{\textcolor{red} 0} & 1 & 3 & 5 & 7 & 9 & 10 & 10 & 10 & 10 & 10 \\
			{\textcolor{red} 0} & 1 & 3 & 5 & 7 & 8 & 8 & 8 & 8 & 8 & 8 \\
			{\textcolor{red} 0} & 1 & 3 & 5 & 6 & 6 & 6 & 6 & 6 & 6 & 6 \\
			{\textcolor{red} 0} & 1 & 3 & 4 & 4 & 4 & 4 & 4 & 4 & 4 & 4 \\
			{\textcolor{red} 0} & 1 & 2 & 2 & 2 & 2 & 2 & 2 & 2 & 2 & 2 \\
			{\textcolor{red} 0} & {\textcolor{red} 0} & {\textcolor{red} 0} & {\textcolor{red} 0} & {\textcolor{red} 0} & {\textcolor{red} 0} & {\textcolor{red} 0} & {\textcolor{red} 0} & {\textcolor{red} 0} & {\textcolor{red} 0} & {\textcolor{red} 0}
		\end{matrix},
	\]
	\[
		\deg_z x_{t, n}(z) \colon \quad
		\begin{matrix}
			{\textcolor{red} 0} & 1 & 4 & 8 & 12 & 16 & 20 & 24 & 28 & 32 & 36 \\
			{\textcolor{red} 0} & 1 & 4 & 8 & 12 & 16 & 20 & 24 & 28 & 32 & 34 \\
			{\textcolor{red} 0} & 1 & 4 & 8 & 12 & 16 & 20 & 24 & 28 & 30 & 30 \\
			{\textcolor{red} 0} & 1 & 4 & 8 & 12 & 16 & 20 & 24 & 26 & 26 & 26 \\
			{\textcolor{red} 0} & 1 & 4 & 8 & 12 & 16 & 20 & 22 & 22 & 22 & 22 \\
			{\textcolor{red} 0} & 1 & 4 & 8 & 12 & 16 & 18 & 18 & 18 & 18 & 18 \\
			{\textcolor{red} 0} & 1 & 4 & 8 & 12 & 14 & 14 & 14 & 14 & 14 & 14 \\
			{\textcolor{red} 0} & 1 & 4 & 8 & 10 & 10 & 10 & 10 & 10 & 10 & 10 \\
			{\textcolor{red} 0} & 1 & 4 & 6 & 6 & 6 & 6 & 6 & 6 & 6 & 6 \\
			{\textcolor{red} 0} & 1 & 2 & 2 & 2 & 2 & 2 & 2 & 2 & 2 & 2 \\
			{\textcolor{red} 0} & {\textcolor{red} 1} & {\textcolor{red} 0} & {\textcolor{red} 0} & {\textcolor{red} 0} & {\textcolor{red} 0} & {\textcolor{red} 0} & {\textcolor{red} 0} & {\textcolor{red} 0} & {\textcolor{red} 0} & {\textcolor{red} 0}
		\end{matrix}.
	\]
\end{example}

\begin{example}
	Let us calculate the degrees of the equation
	\[
		x_{t, n} = - x_{t - 1, n - 1} + \frac{a}{x^k_{t - 1, n}} + \frac{b}{x^k_{t, n - 1}},
	\]
	where $a$ and $b$ are non-zero parameters and $k$ is a positive odd integer greater than $1$.
	Let $z = x_{t_0, n_0}$ be one of the initial variables and let
	\[
		\Phi = \Phi(B, C, D) = - D + \frac{a}{B^k} + \frac{b}{C^k}.
	\]

	Most calculations are the same as in the case of Example~\ref{example:nonintegrable_confining}.
	The only difference is the basic pattern corresponding to $x^{\text{basic}}_{0, 0} = 0$, which is not confining anymore \cite{exkdv}:
	\[
		\begin{matrix}
			{\color{red} \vdots} & & & & \iddots & \iddots \\
			{\color{red} \bullet} & & & \infty^k & 0 & \iddots \\
			{\color{red} \bullet} & & \infty^k & 0 & \infty^k \\
			{\color{red} \bullet} & \infty^k & 0 & \infty^k \\
			{\color{red} \bullet} & {\color{red} x^{\text{basic}}_{0, 0} = 0} & \infty^k \\
			& {\color{red} \bullet} & {\color{red} \bullet} & {\color{red} \bullet} & {\color{red} \bullet} & {\color{red} \cdots}
		\end{matrix}.
	\]
	Therefore, the corresponding movable patterns are
	\[
		\begin{matrix}
			& & & \iddots & \iddots \\
			& & \infty^k & 0 & \iddots \\
			& \infty^k & 0 & \infty^k \\
			\infty^k & 0 & \infty^k \\
			x_{t, n} = 0 & \infty^k
		\end{matrix}.
	\]

	For example, take the same domain $H$ as in Example~\ref{example:nonintegrable_confining}.
	Then, the degree relations for $(t, n) \in H$ are
	\begin{align*}
		\deg_z x_{t, n}(z)
		&= \sum^{+ \infty}_{\ell = 0} Z_{t - \ell, n - \ell}(0) + \delta^{t, 1}_{n, 0} \\
		&= k \sum^{+ \infty}_{\ell = 0} \left( Z_{t - 1 - \ell, n - \ell}(0) + Z_{t - \ell, n - 1 - \ell}(0) \right)  + \delta_{t - 1, n} + k (\delta_{t, n} - \delta^{t, 0}_{n, 0}),
	\end{align*}
	where $Z_{t, n}(0) = 0$ for $(t, n) \in \mathbb{Z}^2 \setminus (H \setminus H_0)$.
	Taking the difference in the $(- 1, - 1)$-direction, we have
	\[
		Z_{t, n}(0) - \delta^{t - 1, 1}_{n - 1, 0}
		= k Z_{t - 1, n}(0) + k Z_{t, n - 1}(0) + k \delta^{t - 1, 0}_{n - 1, 0}
	\]
	for $(t, n) \in H \setminus H_0$, which implies exponential degree growth.
	If $k = 3$, for example, the following are the results of iterating the recursion over a finite range:
	\[
		Z_{t, n}(0) \colon \quad
		\begin{matrix}
			{\textcolor{red} 0} & 2187 & 46656 & 566433 & 5143824 & 38854242 & 257926032 & 1554996366 & \\ 
			{\textcolor{red} 0} & 729 & 13365 & 142155 & 1148175 & 7807590 & 47121102 & 260406090 \\ 
			{\textcolor{red} 0} & 243 & 3726 & 34020 & 240570 & 1454355 & 7899444 & 39680928 \\ 
			{\textcolor{red} 0} & 81 & 999 & 7614 & 46170 & 244215 & 1178793 & 5327532 \\ 
			{\textcolor{red} 0} & 27 & 252 & 1539 & 7776 & 35235 & 148716 & 597051 \\ 
			{\textcolor{red} 0} & 9 & 57 & 261 & 1053 & 3969 & 14337 & 50301 \\ 
			{\textcolor{red} 0} & 3 & 10 & 30 & 90 & 270 & 810 & 2430 \\ 
			{\textcolor{red} 0} & {\textcolor{red} 0} & {\textcolor{red} 0} & {\textcolor{red} 0} & {\textcolor{red} 0} & {\textcolor{red} 0} & {\textcolor{red} 0} & {\textcolor{red} 0} 
		\end{matrix},
	\]
	\[
		\deg_z x_{t, n}(z) \colon \quad
		\begin{matrix}
			{\textcolor{red} 0} & 2187 & 47385 & 580041 & 5289786 & 40037463 & 265982067 & 1603619592 \\ 
			{\textcolor{red} 0} & 729 & 13608 & 145962 & 1183221 & 8056035 & 48623226 & 268557796 \\ 
			{\textcolor{red} 0} & 243 & 3807 & 35046 & 248445 & 1502124 & 8151706 & 40896039 \\ 
			{\textcolor{red} 0} & 81 & 1026 & 7875 & 47769 & 252262 & 1215111 & 5480307 \\ 
			{\textcolor{red} 0} & 27 & 261 & 1599 & 8047 & 36318 & 152775 & 611658 \\ 
			{\textcolor{red} 0} & 9 & 60 & 271 & 1083 & 4059 & 14607 & 51111 \\ 
			{\textcolor{red} 0} & 3 & 10 & 30 & 90 & 270 & 810 & 2430 \\ 
			{\textcolor{red} 0} & {\textcolor{red} 1} & {\textcolor{red} 0} & {\textcolor{red} 0} & {\textcolor{red} 0} & {\textcolor{red} 0} & {\textcolor{red} 0} & {\textcolor{red} 0} 
		\end{matrix}.
	\]
\end{example}

Let us see what kind of problem arises if an equation does not satisfy the $\partial$-factor condition or the basic pattern condition.

\begin{example}\label{example:counterexample_1}
	Recall the equation from Example~\ref{example:polynomial_concrete}:
	\[
		x_{t, n} = \Phi \left( x_{t - 1, n}, x_{t, n - 1}, x_{t - 1, n - 1} \right), \quad
		\Phi(B, C, D) = C (B - D).
	\]
	This equation does not satisfy the $\partial$-factor condition since
	\[
		\partial_C \Phi(B, C, D) = B - D.
	\]
	Since this is a polynomial equation, $\partial$-factor condition is not relevant (Example~\ref{example:polynomial}), and its degrees were computed in Example~\ref{example:polynomial_concrete}.
	Here, we see that this equation is a counterexample to Lemma~\ref{lemma:key_lemma}.

	Let $H = \mathbb{Z}^2_{\ge (0, 0)}$, $z = x_{1, 0}$, $(t, n) = (1, 2)$ and $(s, m) = (2, 2)$.
	Suppose that $z^{*} \in \mathbb{P}^1 \left( \overline{\mathbb{K}_{1, 2}} \right)$ satisfies
	\[
		x_{1, 1}(z^{*}) \ne 0, \quad
		x_{1, 2}(z^{*}) - x_{1, 1}(z^{*}) = 0.
	\]
	Let $w = x_{2, 0}$ and think of $x_{2, 1}$ and $x_{2, 2}$ as rational functions in $z$ and $w$.
	Then, a direct calculation shows that
	$\partial_w x_{2, 1}(z^{*}, w) \ne 0$ and
	$\partial_w \frac{1}{x_{2, 1}(z^{*}, w)} \ne 0$.
	Therefore, except for the $\partial$-factor condition, all the assumptions of Lemma~\ref{lemma:key_lemma} are satisfied.
	However, it follows from $x_{1, 2}(z^{*}) - x_{1, 1}(z^{*}) = 0$ that $x_{2, 2}(z^{*}, w) = 0$.
	In particular, we have
	\[
		\partial_w x_{2, 2}(z^{*}, w) = 0.
	\]
	Hence, this equation does not satisfy Lemma~\ref{lemma:key_lemma}.
	Although $z^{*}$ belongs to $\mathbb{P}^1 \left( \overline{\mathbb{K}_{1, 2}} \right)$, the first constant singularity of this pattern is $x_{2, 2}(z^{*}) = 0$.
	Therefore, Proposition~\ref{proposition:p1_decomposition} does not hold, either.
\end{example}

\begin{example}\label{example:counterexample_2}
	Recall the discrete Liouville equation
	\[
		x_{t, n} = \Phi \left( x_{t - 1, n}, x_{t, n - 1}, x_{t - 1, n - 1} \right), \quad
		\Phi(B, C, D) = \frac{B C + 1}{D}.
	\]
	As discussed in Examples~\ref{example:laurent_property} and \ref{example:laurent_property_concrete}, one can compute the degrees by focusing on the value $\infty$ because the equation has the Laurent property.
	Here, we concentrate on the value $0$, instead of $\infty$.
	While this equation satisfies the $\partial$-factor condition, it does not satisfy the basic pattern condition for $0$ since the relation $\Phi(B, C, D) = 0$ contains the factor $B C + 1$.

	Let $H$ and $z$ be the same as in the previous example and let $(t, n) \in H \setminus H_0$.
	Suppose that $z^{*} \in \mathbb{P}^1 \left( \overline{\mathbb{K}_{t, n}} \right)$ satisfies
	\[
		x_{t - 1, n}(z^{*}) x_{t, n - 1}(z^{*}) + 1 = 0.
	\]
	Since $x_{t - 1, n}(z)$ (resp.\ $x_{t, n - 1}(z)$) depends on the initial variable $x_{0, n}$ (resp.\ $x_{t, 0}$), $z^{*}$ belongs to $\Delta_{t, n}$.
	Therefore, we have $x_{t - 1, n - 1}(z^{*}) \ne 0$ and thus $x_{t, n}(z^{*}) = 0$, which is the first singularity of this pattern.
	However, Lemma~\ref{lemma:algebraic_independence_hard} does not hold anymore because of the algebraic relation $x_{t - 1, n}(z^{*}) x_{t, n - 1}(z^{*}) + 1 = 0$.
	In this case, Theorem~\ref{theorem:main_theorem} does not hold, either.
	In fact, the basic pattern corresponding to $x^{\text{basic}}_{0, 0} = 0$ is
	\[
		\begin{matrix}
			{\color{red} \vdots} & \vdots & \vdots & \vdots & \iddots \\
			{\color{red} x^{\text{basic}}_{- 1, 2}} & \text{reg} & \infty & \infty & \ldots \\
			{\color{red} x^{\text{basic}}_{- 1, 1}} & \text{reg} & \infty & \infty & \ldots \\
			{\color{red} x^{\text{basic}}_{- 1, 0}} & {\color{red} x^{\text{basic}}_{0, 0} = 0} & \text{reg} & \text{reg} & \ldots \\
			& {\color{red} x^{\text{basic}}_{0, - 1}} & {\color{red} x^{\text{basic}}_{1, - 1}} & {\color{red} x^{\text{basic}}_{2, - 1}} & {\color{red} \ldots}
		\end{matrix}.
	\]
	This pattern was already calculated in Example~\ref{example:laurent_property_concrete}, up to a shift of indices, and is not confining.
	In contrast, the movable pattern starting with $x_{t, n}(z^{*}) = 0$ is solitary, as can be shown using the Laurent property and the coprimeness.
	Therefore, our strategy to use basic patterns to calculate movable patterns is not valid anymore if an equation does not satisfy the basic pattern condition.
\end{example}

\section{Conclusions and Discussion}\label{section:conclusion}

In this paper, we proposed a new method to rigorously calculate exact degrees for lattice equations.
Our basic idea was to extend Halburd's method to lattice equations and to give a framework to make the calculations rigorous.
If an equation satisfies some conditions, we can compute the degrees from its singularity patterns by solving a system of linear recurrences.

In \textsection\ref{section:introduction_calculation}, we extended Halburd's method to lattice equations without rigorous discussion.
The calculation part was not very different from the original method, but to guarantee the obtained degrees are exact and rigorous, we had many problems to resolve.

The main part of this paper was \textsection\ref{section:main} and \textsection\ref{section:proof}.
We proposed a framework to rigorously compute degrees in \textsection\ref{section:main} and gave proofs in \textsection\ref{section:proof}.
The key idea was to think of the initial values other than $z$ just as variables, instead of generic numerical values, and to clarify on which initial variables each substituted value for $z$ depends.
Such an approach is valid because an initial value problem for a lattice equation has infinitely many initial variables.
Using this approach, we showed that if an equation satisfies some conditions, each singularity pattern has only one starting point (Theorem~\ref{theorem:first_singularity}) and that all the movable patterns are determined by the corresponding basic patterns (Theorem~\ref{theorem:main_theorem}).
We also showed that the degree relations obtained by the singularity structure are exact (Theorem~\ref{theorem:main_degree}) and gave a condition for the degree relations to determine the degrees (Theorem~\ref{theorem:degree_solvable}).
Roughly speaking, we showed that if there exists a $\mathbb{P}^1(\mathbb{C})$-value that does not appear as the starting value of any movable pattern, one can obtain the degrees by solving the degree relations.
As long as our conditions are satisfied, the degrees can be computed rigorously in a unified manner, whereas previously such computations were possible only for equations with highly favorable structural properties.
As a corollary, it holds that the degrees are determined only by the singularity structure.

In \textsection\ref{section:examples}, we calculated the degrees for several lattice equations with our theory, focusing on how to apply our main theorems.
As expected in advance, the calculation part was similar to that of \textsection\ref{section:introduction_calculation}.
Thanks to our theorems and propositions, however, the calculations in this section are all exact and rigorous.
We studied how to detect and rigorously calculate constant singularity patterns as part of the method.
This aspect of our approach may also be useful for studies on singularity confinement in lattice equations, which could benefit from a more systematic treatment.

It should be mentioned that computing the degrees for a quad equation by deriving a linear partial difference equation satisfied by the degrees has already been proposed in previous studies \cite{tran1,tran2}.
There, one needs to assume some conjecture to ensure the calculations are rigorous.
Unfortunately, our theory does not give a direct answer to their conjecture since the degrees considered in their papers are total degrees while the degrees in this paper are individual degrees.
However, there is expected to be a strong connection between the strategy used in the previous studies and our theory.

Let us consider a future extension of our theory.

First, let us consider whether our theory can be applied to other types of lattice equations, such as non-quad equations on a two-dimensional lattice or equations on a higher-dimensional lattice.
Our main idea remains effective.
In fact, many lemmas and propositions, such as Propositions~\ref{proposition:solitary}, \ref{proposition:check_hidden_value} and Lemma~\ref{lemma:trdeg_two}, are still valid in general cases, with appropriate modifications.
However, the proofs of some lemmas, such as Lemmas~\ref{lemma:trdeg_at_least_two} and \ref{lemma:key_lemma}, are not valid for general lattice equations.
It seems that how difficult it is to show these lemmas varies depending on the configuration of the stencil of a lattice equation.

Consider an equation defined on a lattice $L$ using shift vectors $0, v_1, \ldots, v_N \in L$:
\[
	x_h = \Phi(x_{h + v_1}, \ldots, x_{h + v_N}) \quad
	(h \in L),
\]
where $\Phi$ depends on all of its arguments.
We refer to the set $\mathcal{S} = \{ 0, v_1, \ldots, v_N \}$ as the stencil of the equation \cite{venderkamp2009}.
With a slight modification of the definition in \cite{venderkamp2009}, we call the convex hull of the stencil \emph{$\mathcal{S}$-polytope}.
In the case of quad equations, the $\mathcal{S}$-polytope is a unit square $\operatorname{conv} \{ (0, 0), (- 1, 0), (0, - 1), (- 1, - 1) \}$.
This polytope has two edges attached to the origin $(0, 0)$, but neither of them contains a shift vector other than its endpoints, i.e., all shift vectors on each edge are its endpoints.
This property played an essential role in the proofs of Lemmas~\ref{lemma:trdeg_at_least_two} and \ref{lemma:key_lemma}.
Therefore, to prove these lemmas for general lattice equations would require some knowledge of convex geometry.

Next, let us consider the $\partial$-factor condition and the basic pattern condition (Definitions~\ref{definition:d_factor_condition} and \ref{definition:basic_pattern_condition}).
As already seen in Example~\ref{example:counterexample_1}, if an equation does not satisfy the $\partial$-factor condition, $z^{*} \in \Delta_{t, n}$ can generate a pattern whose first constant singularity occurs not at $(t, n)$ but somewhere else.
This problem affects the definition of $Z_{t, n}(\beta)$.
Moreover, showing $\partial_C \Phi(x_{s - 1, m}(z^{*}), x_{s, m - 1}(z^{*}, w), x_{s - 1, m - 1}(z^{*})) \ne 0$ in the proof of Lemma~\ref{lemma:key_lemma} would become more difficult since the numerator of the derivative has a nontrivial factor.

On the other hand, as seen in Example~\ref{example:counterexample_2}, if an equation does not satisfy the basic pattern condition, a movable pattern can be different from the corresponding basic pattern.
Therefore, we must consider some modification of basic patterns.
Suppose that for $x^{*} \in \mathbb{P}^1(\mathbb{C})$, the relation $\Phi(B, C, D) = x^{*}$ has a nontrivial irreducible factor $F = F(B, C)$.
Then, the corresponding ``semi-basic pattern'' should be
\[
	\begin{matrix}
		{\color{red} \vdots} & & & & \iddots \\
		{\color{red} \bullet} & \circ & \circ & \circ \\
		{\color{red} \bullet} & \circ & \circ & \circ \\
		{\color{red} x^{\text{semi}}_{- 1, 0}} & x^{*} & \circ & \circ \\
		{\color{red} \bullet} & {\color{red} x^{\text{semi}}_{0, - 1}} & {\color{red} \bullet} & {\color{red} \bullet} & {\color{red} \cdots}
	\end{matrix},
\]
where $x^{\text{semi}}_{- 1, 0}$ and $x^{\text{semi}}_{0, - 1}$ are taken as $F \left( x^{\text{semi}}_{- 1, 0}, x^{\text{semi}}_{0, - 1} \right) = 0$ and $\operatorname{trdeg}_{\mathbb{C}} \mathbb{C} \left( x^{\text{semi}}_{- 1, 0}, x^{\text{semi}}_{0, - 1} \right) = 1$, which generates $x^{\text{semi}}_{0, 0} = x^{*}$, and the other initial values are just variables.
For instance, in the case of Example~\ref{example:counterexample_2}, the non-trivial factor $F(B, C) = B C + 1$ of the relation $\Phi(B, C, D) = 0$ corresponds to the semi-basic pattern $x^{\text{semi}}_{- 1, 0} x^{\text{semi}}_{0, - 1} + 1 = 0$.
Note that such an idea has already been used in singularity analysis \cite{lattice_sc_classic}.
If one uses semi-basic patterns to calculate degrees, the definition of $Z_{t, n}(\beta)$ must be modified, too, such as
$Z_{t, n}(\beta) = Z^{\text{basic}}_{t, n}(\beta) + Z^{\text{semi}}_{t, n}(\beta)$
where $Z^{\text{basic}}_{t, n}(\beta)$ (resp.\ $Z^{\text{semi}}_{t, n}(\beta)$) is the number of spontaneous occurrences of $x_{t, n}(z^{*}) = \beta$ with $\operatorname{trdeg}_{\mathbb{C}} \mathbb{C} \left( x_{t - 1, n}(z^{*}), x_{t - 1, n}(z^{*}) \right) = 2$ (resp.\ $= 1$).
However, if there exists $\beta \in \mathbb{P}^1(\mathbb{C})$ such that neither $Z^{\text{basic}}_{t, n}(\beta)$ nor $Z^{\text{semi}}_{t, n}(\beta)$ vanishes, the degree relations tend to be insufficient to determine the degrees since $Z^{\text{basic}}_{t, n}(\beta)$ and $Z^{\text{semi}}_{t, n}(\beta)$ are both considered to be unknown variables in the proof of Theorem~\ref{theorem:degree_solvable}.

All the singularities we considered in this paper were constant singularities.
However, even in the case of the discrete KdV equations, other types of singularities have been reported in the context of singularity analysis \cite{sc_kdv1,sc_kdv2}.
Then, what should a rigorous definition of a non-constant singularity be?
A key to the question already appeared in \textsection\ref{section:proof}.
A natural definition would be that $z^{*} \in \mathbb{P}^1 \left( \overline{\mathbb{K}} \right)$ generates an algebraic singularity at $(t, n)$ if $\operatorname{trdeg}_{\mathbb{C}} \mathbb{C} \left( x_{t - 1, n}(z^{*}), x_{t, n - 1}(z^{*}), x_{t - 1, n - 1}(z^{*}) \right) \le 2$.
That is, an algebraic singularity occurs when $x_{t - 1, n}(z^{*})$, $x_{t, n - 1}(z^{*})$ and $x_{t - 1, n - 1}(z^{*})$ have a non-trivial algebraic relation over $\mathbb{C}$.
An algebraic singularity is called a constant singularity in our terminology if the algebraic relation is $\Phi \left( x_{t - 1, n}(z^{*}), x_{t, n - 1}(z^{*}), x_{t - 1, n - 1}(z^{*}) \right) = x^{*}$ for some $x^{*} \in \mathbb{P}^1 (\mathbb{C})$.
Actually, as for the semi-basic pattern considered in the above example, the starting point can be thought of as the algebraic singularity $F \left( x^{\text{semi}}_{- 1, 0}, x^{\text{semi}}_{0, - 1} \right) = 0$, instead of the constant singularity $x^{\text{semi}}_{0, 0} = x^{*}$.
Therefore, algebraic singularities could be useful when considering an equation without the basic pattern condition.

Lastly, let us consider whether our theory has some applications to mappings, i.e., equations on a one-dimensional lattice.
As noted in advance, it does not seem possible at all to apply our method to mappings since our theory, especially Lemma~\ref{lemma:key_lemma}, relies on the fact that an initial value problem for a lattice equation has infinitely many initial variables.
Although some ideas, such as Theorems~\ref{theorem:main_degree} and \ref{theorem:degree_solvable}, Lemma~\ref{lemma:main_lemma} and the strategy to use transcendental degrees, seem applicable, the proof of Theorem~\ref{theorem:main_theorem} is not valid anymore.
In other words, we do not know in the first place whether a movable pattern coincides with the corresponding basic pattern.
Therefore, to calculate exact degrees for mappings in a similar way to ours, some new idea is required.

\section*{Acknowledgement}

I wish to thank Prof.~R. Willox for discussion and comments.
I also thank the anonymous referee whose comments helped improve the presentation of the paper.
This work was supported by a Grant-in-Aid for Scientific Research of Japan Society for the Promotion of Science, JSPS KAKENHI Grant Number 23K12996.

\section*{Data availability}

All numerical calculations supplementing the results of this study are included in the article.

\appendix

\section{Essential choices of domain}\label{appendix:domain}

In this appendix, we show that there are only four essential choices of domain when one considers the degrees with respect to one of the initial variables.

\begin{theorem}\label{theorem:domain_choice}
	Suppose that a domain $H \subset \mathbb{Z}^2$ (see Definition~\ref{definition:domain}) and consider a quad equation \eqref{equation:quad_equation} on $H$.
	Let $(t_0, n_0) \in H_0$, $z = x_{t_0, n_0}$ and consider $\deg_z x_{t, n}(z)$, the degree with respect to $z$.
	Note that by the second condition of Definition~\ref{definition:domain}, $(t_0 + 1, n_0 + 1)$ must belong to $H \setminus H_0$.
	Define domains $H^{(1)}, H^{(2)}, H^{(3)}, H^{(4)} \subset \mathbb{Z}^2$ as
	\begin{itemize}
		\item 
		$H^{(1)} = \mathbb{Z}^2_{\ge (t_0, n_0)}$,

		\item 
		$H^{(2)} = \mathbb{Z}^2_{\ge (t_0 - 1, n_0)}$,

		\item 
		$H^{(3)} = \mathbb{Z}^2_{\ge (t_0, n_0 - 1)}$,

		\item 
		$H^{(4)} = \mathbb{Z}^2_{\ge (t_0 - 1, n_0 - 1)} \setminus \{ (t_0 - 1, n_0 - 1) \}$

	\end{itemize}
	and let $x^{(j)}_{t, n}$ be the solution to the equation \eqref{equation:quad_equation} on the domain $H^{(j)}$.
	\begin{enumerate}
		\item
		If $(t_0 + 1, n_0)$ and $(t_0, n_0 + 1)$ both belong to $H_0$, then for any $(t, n) \in H_{\ge (t_0, n_0)}$, the degree $\deg_z x_{t, n}$ coincides with $\deg_z x^{(1)}_{t, n}$.

		\item\label{enumerate:domain_case}
		If $(t_0 + 1, n_0)$ belongs to $H_0$ but $(t_0, n_0 + 1)$ does not, then for any $(t, n) \in H_{\ge (t_0, n_0)}$, the degree $\deg_z x_{t, n}$ coincides with $\deg_z x^{(2)}_{t, n}$.

		\item
		If $(t_0, n_0 + 1)$ belongs to $H_0$ but $(t_0 + 1, n_0)$ does not, then for any $(t, n) \in H_{\ge (t_0, n_0)}$, the degree $\deg_z x_{t, n}$ coincides with $\deg_z x^{(3)}_{t, n}$.

		\item
		If neither $(t_0 + 1, n_0)$ nor $(t_0, n_0 + 1)$ belongs to $H_0$, then for any $(t, n) \in H_{\ge (t_0, n_0)}$, the degree $\deg_z x_{t, n}$ coincides with $\deg_z x^{(4)}_{t, n}$.

	\end{enumerate}
\end{theorem}

\begin{proof}
	We check that $x_{s, m}(z)$ ($(s, m) \in H^{(j)}_0$), which are the solution on $H$, can be thought of as if they were initial variables (see Figure~\ref{figure:domain_choice}).
	It is sufficient to show the following two statements:
	\begin{itemize}
		\item 
		The set
		\[
			\{ x_{s, m}(z) \mid (s, m) \in H^{(j)}_0 \}
		\]
		is algebraically independent.

		\item
		For any $(s, m) \in H^{(j)}_0 \setminus \{ (t_0, n_0) \}$, $x_{s, m}(z)$ does not depend on $z$.

	\end{itemize}
	The first statement immediately follows from the shape of the set $H^{(j)}_0$.

	Let us check the second statement.
	From here on, we will only consider case \eqref{enumerate:domain_case} since the proof in the other cases is almost the same.
	Since $t_0 - 1 < t_0$,
	\[
		x_{t_0 - 1, n_0}(z), x_{t_0 - 1, n_0 + 1}(z), x_{t_0 - 1, n_0 + 2}(z), \cdots
	\]
	are independent of $z = x_{t_0, n_0}$.
	Since it is an initial variable, $x_{t_0 + 1, n_0}(z)$ does not depend on $z$.
	Let us consider $x_{t_0 + 2, n_0}(z)$.
	If $x_{t_0 + 2, n_0}(z)$ is an initial variable, then it clearly does not depend on $z$.
	Otherwise, $x_{t_0 + 2, n_0}(z)$ is calculated from $x_{t_0 + 1, n_0}(z)$, $x_{t_0 + 2, n_0 - 1}(z)$ and $x_{t_0 + 1, n_0 - 1}(z)$.
	However, since $x_{t_0 + 1, n_0}(z)$, $x_{t_0 + 2, n_0 - 1}(z)$ and $x_{t_0 + 1, n_0 - 1}(z)$ are all independent of $z$, $x_{t_0 + 2, n_0}(z)$ does not depend on $z$.
	Therefore, regardless of whether it is an initial variable, $x_{t_0 + 2, n_0}(z)$ does not depend on $z$.
	Repeating this procedure, one can show that $x_{t_0 + 3, n_0}(z), x_{t_0 + 4, n_0}(z), \cdots$ do not depend on $z$, either, which completes the proof.
\end{proof}

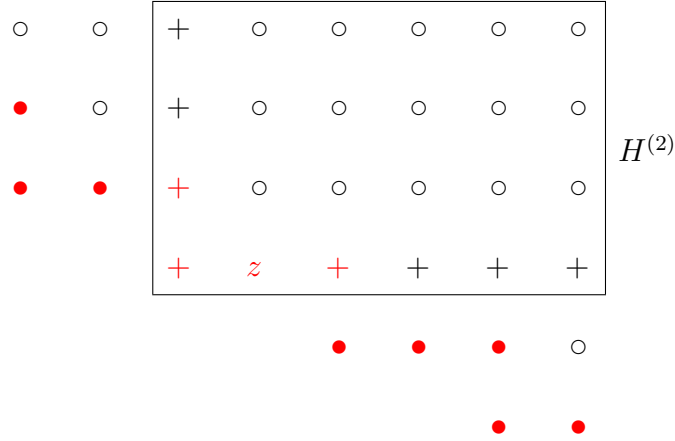
\begin{figure}
	\begin{center}
		\begin{picture}(250, 200)

			\multiput(100, 160)(30, 0){5}{\circle{5}}
			\multiput(100, 130)(30, 0){5}{\circle{5}}
			\multiput(100, 100)(30, 0){5}{\circle{5}}
			\multiput(220, 40)(30, 0){1}{\circle{5}}

			\multiput(10, 160)(30, 0){2}{\circle{5}}
			\multiput(40, 130)(30, 0){1}{\circle{5}}

			\multiput(155, 67)(30, 0){3}{$+$}
			\multiput(65, 127)(0, 30){2}{$+$}

			\put(60, 60){\framebox(170, 110)}
			\put(235, 110){$H^{(2)}$}

			\color{red}

			\put(95, 67){$z$}
			\multiput(10, 130)(30, 0){1}{\circle*{5}}
			\multiput(10, 100)(30, 0){2}{\circle*{5}}
			\multiput(65, 67)(0, 30){2}{$+$}
			\multiput(125, 67)(0, 30){1}{$+$}
			\multiput(130, 40)(30, 0){3}{\circle*{5}}
			\multiput(190, 10)(30, 0){2}{\circle*{5}}

			\color{black}

		\end{picture}
	\end{center}
	\caption{
		Situation in Theorem~\ref{theorem:domain_choice}.
		The points marked with ``$+$'' correspond to the initial boundary of $H^{(2)}$.
		For example, it follows from this figure that $x_{t_0 + 2, n_0}(z)$, which corresponds to a point marked with ``$+$'' but is not an initial variable in $H$, does not depend on $z = x_{t_0, n_0}$.
	}
	\label{figure:domain_choice}
\end{figure}

\begin{remark}
	With minor modifications, the theorem still holds even in the case of non-quad equations.
	For example, if an equation is defined on a lattice $L$ as
	\[
		x_h = \Phi(x_{h + h^{(1)}}, \ldots, x_{h + h^{(N)}})
		\quad
		(h \in L)
	\]
	with shifts $h^{(1)}, \ldots, h^{(N)} \in L$,
	then the number of essential choices of domain is at most $2^N$.
	If $z = x_{h_0}$ ($h_0 \in H_0$), the choices are whether $h_0 - h^{(j)}$ belongs to the initial boundary for each $j = 1, \ldots, N$.
	Some of the cases are excluded because of the conditions on a domain, such as the second condition in Definition~\ref{definition:domain}.
\end{remark}

\section{Substitution and derivatives}\label{appendix:substitution_and_derivatives}

In this appendix, let $K$ be an arbitrary field of characteristic $0$ and consider the rational function field $K(z, w)$.
We allow $\infty$ as the result of substitution throughout this paper.
This is why we need some unusual definitions and lemmas about substitution and derivatives.

\begin{definition}\label{definition:substitution}
	For $z^{*} \in \mathbb{P}^1 \left( \overline{K} \right)$, we define the substitution map
	\[
		K(z, w) \to \overline{K}(w) \cup \{ \infty \}; \quad
		f(z, w) \mapsto f(z^{*}, w)
	\]
	as follows.

	First, we consider the case $z^{*} \ne \infty$.
	Let $f(z, w) = \frac{p(z, w)}{q(z, w)}$ where $p, q \in K[z, w]$ are coprime.
	If $q(z^{*}, w) = 0$, i.e., $q$ is divisible by $z - z^{*}$, then we set $f(z^{*}, w) = \infty$.
	Note that in this case, $p(z^{*}, w)$ cannot become $0$.
	If $q(z^{*}, w) \ne 0$, then we define $f(z^{*}, w)$ as $f(z^{*}, w) = \frac{p(z^{*}, w)}{q(z^{*}, w)}$.

	Next, we consider the case $z^{*} = \infty$.
	In this case, we introduce $\widetilde{z} = \frac{1}{z}$ and substitute $\widetilde{z} = 0$.
\end{definition}

It should be noted that all the substitutions in this paper are codimension $1$.
In Definition~\ref{definition:substitution}, for example, we only substitute a single variable $z$.
More general settings are considered in Lemma~\ref{lemma:substitution_trdeg_codimension_one}.
Therefore, allowing $\infty$ as a value, substitutions are always well-defined.

The following properties immediately follow from Definition~\ref{definition:substitution}.

\begin{lemma}
	Let $f, g \in K(z, w)$ and let $z^{*} \in \mathbb{P}^1 \left( \overline{K} \right)$.
	\begin{enumerate}
		\item
		If $f(z^{*}, w) \ne \infty$ or $g(z^{*}, w) \ne \infty$, then
		\[
			(f + g)(z^{*}, w) = f(z^{*}, w) + g(z^{*}, w).
		\]

		\item
		If $(f(z^{*}, w), g(z^{*}, w)) \ne (0, \infty)$ or $(f(z^{*}, w), g(z^{*}, w)) \ne (\infty, 0)$, then
		\[
			(f g)(z^{*}, w) = f(z^{*}, w) g(z^{*}, w).
		\]

	\end{enumerate}
\end{lemma}

\begin{definition}
	We use $\partial_w$ to denote the derivative with respect to $w$.
	That is, $\partial_w \colon K(z, w) \to K(z, w)$ satisfies
	\[
		\partial_w z = 0, \quad
		\partial_w w = 1, \quad
		\partial_w a = 0 \ (a \in K)
	\]
	and the Leibniz rule.
\end{definition}

Since $K$ has characteristic $0$, it is easy to show that $\partial_w f(z, w) = 0$ if and only if $f \in K(z)$.

\begin{definition}
	For $f \in K(z, w)$ and $z^{*} \in \mathbb{P}^1 \left( \overline{K} \right)$, we define
	$\partial f(z^{*}, w) \in \overline{K}(w) \cup \{ \infty \}$ as
	\[
		\partial f(z^{*}, w) = \left( \partial_w f(z, w) \right) \big|_{z = z^{*}}.
	\]
\end{definition}

Note that the ordering of derivative and substitution can be interchanged if $f(z^{*}, w) \ne \infty$.

\begin{lemma}
	Let $f(z, w) \in K(z, w)$ and let $z^{*} \in \mathbb{P}^1 \left( \overline{K} \right)$.
	If $f(z^{*}, w)$ does not depend on $w$, then at least one of the following holds:
	\begin{enumerate}
		\item
		$f(z^{*}, w) = \infty$,

		\item
		$\partial_w f(z^{*}, w) = 0$.

	\end{enumerate}
\end{lemma}
\begin{proof}
	Assume that $f(z^{*}, w) \ne \infty$.
	Then, we have
	\[
		\partial_w f(z^{*}, w)
		= \partial_w \left( f(z^{*}, w) \right)
		= 0.
	\]
\end{proof}

\begin{example}
	Let $z^{*} = 0$ and let
	\[
		f(z, w) = \frac{1}{z} + w.
	\]
	Then,
	\[
		f(z^{*}, w) = \infty
	\]
	does not depend on $w$ but
	\[
		\partial_w f(z^{*}, w)
		= \left( \partial_w f(z, w) \right) \big|_{z = 0}
		= 1
		\ne 0.
	\]
	Therefore, in our settings, it is impossible to determine only by $\partial_w$ whether a rational function depends on $w$.
\end{example}

We use the following lemma many times in the main part of the paper.

\begin{lemma}
	For $f(z, w) \in K(z, w) \setminus \{ 0 \}$ and $z^{*} \in \mathbb{P}^1 \left( \overline{K} \right)$, the following two conditions are equivalent:
	\begin{enumerate}
		\item
		$f(z^{*}, w)$ does not depend on $w$,

		\item
		$\partial_w f(z^{*}, w) = 0$ or $\partial_w \frac{1}{f(z^{*}, w)} = 0$.

	\end{enumerate}
	In particular, $f(z^{*}, w)$ depends on $w$ if and only if $\partial_w f(z^{*}, w)$ and $\partial_w \frac{1}{f(z^{*}, w)}$ are both non-zero.
\end{lemma}
\begin{proof}
	Note that the definition of $\partial_w \frac{1}{f(z^{*}, w)}$ is
	\[
		\partial_w \frac{1}{f(z^{*}, w)} = \left( \partial_w \frac{1}{f(z, w)} \right) \big|_{z = z^{*}}.
	\]
	First, we suppose that $f(z^{*}, w)$ does not depend on $w$, i.e., $f(z^{*}, w) \in K \cup \{ \infty \}$.
	If $f(z^{*}, w) \ne \infty$, then we have
	$\partial_w f(z^{*}, w) = 0$.
	Let us consider the case $f(z^{*}, w) = \infty$.
	Introducing $g(z, w) = \frac{1}{f(z, w)}$, we have
	$g(z^{*}, w) = 0$, which implies that $\partial_w \frac{1}{f(z^{*}, w)} = 0$.

	Next, we show the converse.
	Let $\partial_w f(z^{*}, w) = 0$.
	If $f(z^{*}, w) = \infty$, then it is independent of $w$.
	Therefore, we may assume that $f(z^{*}, w) \ne \infty$.
	In this case, usual properties of the derivative hold and thus $f(z^{*}, w)$ does not depend on $w$.
	If $\partial_w \frac{1}{f(z^{*}, w)} = 0$, one can show in the same way that $g(z^{*}, w)$ does not depend on $w$ where $g(z, w) = \frac{1}{f(z, w)}$.
\end{proof}

\section{Algebraic lemmas}\label{appendix:algebraic_lemmas}

In this appendix, we show some purely algebraic lemmas, which we use in \textsection\ref{section:proof}.

\begin{lemma}\label{lemma:expansion_lemma}
	Let $F = F(X; Y_1, \ldots, Y_N) \in \mathbb{C}(X; Y_1, \ldots, Y_n)$
	and let $x^{*} \in \mathbb{C}$.
	Suppose that the expansion of $F(x^{*} + \varepsilon; Y_1, \ldots, Y_N)$ at $\varepsilon = 0$ is
	\[
		F(x^{*} + \varepsilon; Y_1, \ldots, Y_N) = F_{\ell}(Y_1, \ldots, Y_N) \varepsilon^{\ell} + O \left( \varepsilon^{\ell + 1} \right),
	\]
	where $\ell \in \mathbb{Z}$ and $F_{\ell}(Y_1, \ldots, Y_N)$ is non-zero as an element of $\mathbb{C}(Y_1, \ldots, Y_N)$.
	Let $\mathbb{C} \subset K$ be a field extension, $x(z), y_1(z), \ldots, y_N(z) \in K(z)$ and suppose that $z^{*} \in \overline{K}$ satisfies the following conditions:
	\begin{itemize}
		\item
		$x(z^{*} + \varepsilon) = x^{*} + a \varepsilon^r + O \left( \varepsilon^{r + 1} \right)$
		for some $a \in \overline{K} \setminus \{ 0 \}$ and $r \ge 1$,

		\item
		$y_1(z^{*}), \ldots, y_N(z^{*})$ are algebraically independent over $\mathbb{C}$.

	\end{itemize}
	Then, the expansion of $F(x(z^{*} + \varepsilon); y_1(z^{*} + \varepsilon), \ldots, y_N(z^{*} + \varepsilon))$ at $\varepsilon = 0$ is
	\[
		F(x(z^{*} + \varepsilon); y_1(z^{*} + \varepsilon), \ldots, y_N(z^{*} + \varepsilon)) = b \varepsilon^{r \ell} + O \left( \varepsilon^{r \ell + 1} \right)
	\]
	for some $b \in \overline{K} \setminus \{ 0 \}$.
\end{lemma}
\begin{proof}
	Since
	\begin{align*}
		F(x(z^{*} + \varepsilon); y_1(z^{*} + \varepsilon), \ldots, y_N(z^{*} + \varepsilon))
		&= F \left( x^{*} + a \varepsilon^r + O \left( \varepsilon^{r + 1} \right); y_1(z^{*} + \varepsilon), \ldots, y_N(z^{*} + \varepsilon) \right) \\
		&= F_{\ell}(y_1(z^{*} + \varepsilon), \ldots, y_N(z^{*} + \varepsilon)) a^{\ell} \varepsilon^{r \ell} + O \left( \varepsilon^{r \ell + 1} \right),
	\end{align*}
	it is sufficient to show that
	\[
		F_{\ell}(y_1(z^{*} + \varepsilon), \ldots, y_N(z^{*} + \varepsilon)) \sim \varepsilon^0.
	\]
	Let $F_{\ell} = \frac{p_{\ell}}{q_{\ell}}$ where $p_{\ell}, q_{\ell} \in \mathbb{C}[Y_1, \ldots, Y_N]$ are coprime.
	Since $F_{\ell}$ is not $0$ or $\infty$, neither $p_{\ell}$ nor $q_{\ell}$ is $0$ as a polynomial.
	Since $y_1(z^{*}), \ldots, y_N(z^{*})$ are algebraically independent over $\mathbb{C}$, we have
	\[
		p_{\ell}(y_1(z^{*}), \ldots, y_N(z^{*})), \
		q_{\ell}(y_1(z^{*}), \ldots, y_N(z^{*})) \ne 0.
	\]
	Using $y_j(z^{*} + \varepsilon) = y_j(z^{*}) + O(\varepsilon)$ for $j = 1, \ldots, N$, we have
	\begin{align*}
		F_{\ell}(y_1(z^{*} + \varepsilon), \ldots, y_N(z^{*} + \varepsilon))
		&= \frac{p_{\ell}(y_1(z^{*}), \ldots, y_N(z^{*})) + O(\varepsilon)}{q_{\ell}(y_1(z^{*}), \ldots, y_N(z^{*})) + O(\varepsilon)} \\
		&= \frac{p_{\ell}(y_1(z^{*}), \ldots, y_N(z^{*}))}{q_{\ell}(y_1(z^{*}), \ldots, y_N(z^{*}))} + O(\varepsilon) \\
		&\sim \varepsilon^0.
	\end{align*}
\end{proof}
\begin{remark}
	Lemma~\ref{lemma:expansion_lemma} holds even in the case of $z^{*} = \infty$ or $x^{*} = \infty$.
	If $z^{*} = \infty$, use a new variable $\widetilde{z} = \frac{1}{z}$ instead of $z$.
	If $x^{*} = \infty$, introduce $\widetilde{x}(z) = \frac{1}{x(z)}$ and use
	\[
		\widetilde{x}(z^{*} + \varepsilon) = \frac{1}{x(z^{*} + \varepsilon)} = a \varepsilon^r + O \left( \varepsilon^{r + 1} \right)
	\]
	instead of the expansion of $x(z^{*} + \varepsilon)$.
\end{remark}

\begin{lemma}\label{lemma:function_field}
	Let $\mathbb{C} \subset K$ be a field extension and suppose that $y_1, \ldots, y_N \in K$ satisfies
	\[
		\operatorname{trdeg}_{\mathbb{C}} \mathbb{C}(y_1, \ldots, y_N) = N - 1.
	\]
	Then, there exists an irreducible polynomial $G = G(Y_1, \ldots, Y_N) \in \mathbb{C}[Y_1, \ldots, Y_N]$ such that
	\[
		\{ f \in \mathbb{C}[Y_1, \ldots, Y_N] \mid f(y_1, \ldots, y_N) = 0 \} = (G),
	\]
	where the right-hand side is the ideal of $\mathbb{C}[Y_1, \ldots, Y_N]$ generated by $G$.
\end{lemma}
\begin{proof}
	Since this is a basic fact in algebraic geometry, we only give a sketch of the proof.

	We may assume that $y_1, \ldots, y_{N - 1}$ are algebraically independent over $\mathbb{C}$ and $y_N$ is algebraic over $\mathbb{C}(y_1, \ldots, y_{N - 1})$.
	Let $F = F(Y) \in \mathbb{C}(y_1, \ldots, y_{N - 1})[Y]$ be the minimal polynomial of $y_N$ over $\mathbb{C}(y_1, \ldots, y_{N - 1})$.
	Express $F$ as
	\[
		F(Y) = \sum_j \frac{p_j(y_1, \ldots, y_{N - 1})}{q_j(y_1, \ldots, y_{N - 1})} Y^j,
	\]
	where $p_j, q_j \in \mathbb{C}[Y_1, \ldots, Y_{N - 1}]$ are pairwise coprime for each $j$.
	Then, the polynomial
	\[
		G = \operatorname{LCM}_{\ell} (q_{\ell}(Y_1, \ldots, Y_{N - 1})) \sum_j \frac{p_j(Y_1, \ldots, Y_{N - 1})}{q_j(Y_1, \ldots, Y_{N - 1})} Y^j_N
	\]
	satisfies the condition.
\end{proof}

\begin{lemma}\label{lemma:polynomial_nonvanishing}
	Let $\mathbb{C} \subset K$ be a field extension and let $x, y_1, \ldots, y_N \in K$.
	Suppose that an irreducible polynomial $F = F(X; Y_1, \ldots, Y_N) \in \mathbb{C}[X; Y_1, \ldots, Y_N]$ satisfies $\partial_X F \ne 0$.
	If $\operatorname{trdeg}_{\mathbb{C}} \mathbb{C}(y_1, \ldots, y_N) \ge N - 1$ and $x$ is transcendental over $\mathbb{C}(y_1, \ldots, y_N)$, then
	\[
		F(x; y_1, \ldots, y_N) \ne 0.
	\]
\end{lemma}
\begin{proof}
	Since the statement is trivial in the case of $\operatorname{trdeg}_{\mathbb{C}} \mathbb{C}(y_1, \ldots, y_N) = N$, we may assume that the transcendental degree is $N - 1$.
	Let
	\[
		F = \sum_{\ell} F_{\ell}(Y_1, \ldots, Y_N) X^{\ell}.
	\]
	It is sufficient to show that $F_{\ell}(y_1, \ldots, y_N) \ne 0$ for some $\ell$.
	Let $G \in \mathbb{C}[Y_1, \ldots, Y_N]$ be the irreducible polynomial in Lemma~\ref{lemma:function_field}, i.e.,
	\[
		\{ f \in \mathbb{C}[Y_1, \ldots, Y_N] \mid f(y_1, \ldots, y_N) = 0 \} = (G).
	\]
	Since $F$ is irreducible and $\partial_X F \ne 0$, there exists $\ell$ such that $G$ does not divide $F_{\ell}$.
	Therefore, by the definition of $G$, we have $F_{\ell}(y_1, \ldots, y_n) \ne 0$.
\end{proof}

\begin{lemma}\label{lemma:substitution_trdeg_codimension_one}
	Let $\mathbb{C} \subset K$ be a field extension and suppose that $y_1, \ldots, y_N \in K$ satisfies
	\[
		\operatorname{trdeg}_{\mathbb{C}} \mathbb{C}(y_1, \ldots, y_N) \ge N - 1.
	\]
	If $F_1, F_2 \in \mathbb{C}[Y_1, \ldots, Y_N]$, $F_1, F_2 \ne 0$, are coprime, then at least one of $F_1(y_1, \ldots, y_N)$ and $F_2(y_1, \ldots, y_N)$ is non-zero.
	In particular, for a rational function $F \in \mathbb{C}(Y_1, \ldots, Y_N)$, the substitution $F(y_1, \ldots, y_N) \in K \cup \{ \infty \}$ is always well-defined, i.e., $F(y_1, \ldots, y_N)$ does not become $\frac{0}{0}$.
\end{lemma}
\begin{proof}
	If $\operatorname{trdeg}_{\mathbb{C}} \mathbb{C}(y_1, \ldots, y_N) = N$, then both $F_1(y_1, \ldots, y_N)$ and $F_2(y_1, \ldots, y_N)$ are non-zero.
	Seeking a contradiction, assume that the transcendental degree is $N - 1$ but $F_1(y_1, \ldots, y_N) = F_2(y_1, \ldots, y_N) = 0$.
	Let $G$ be the irreducible polynomial in Lemma~\ref{lemma:function_field}.
	Then, $F_1$ and $F_2$ both belong to $(G)$, which contradicts the coprimeness of $F_1$ and $F_2$.
\end{proof}

\begin{remark}
	Geometrically speaking, $(y_1, \ldots, y_N) \in K^N$ ($\operatorname{trdeg} = N - 1$) corresponds to the hypersurface $\{ G = 0 \}$ of $\mathbb{C}^N$, where $G$ is defined in Lemma~\ref{lemma:function_field}.
	On the other hand, the set of indeterminacies of a rational function has at least codimension $2$.
	This is why the substitution of $(y_1, \ldots, y_N)$ is always well-defined.
\end{remark}

\section{Divisors on \texorpdfstring{$\mathbb{P}^1(K)$}{P1(K)}}\label{appendix:divisors}

For those who are not familiar with algebraic geometry, we briefly review the basic notion of divisors on $\mathbb{P}^1(K)$, where $K$ is an algebraically closed field.
In this paper, we use divisors only in order to exactly count the number of preimages with multiplicity.
Although the definition of divisors is the same as usual, some notations, such as $\operatorname{div} (f = \alpha \mid T)$, are introduced by the author.

\begin{definition}\label{definition:rational_function_multiplicity}
	Let $f \colon \mathbb{P}^1(K) \to \mathbb{P}^1(K)$ be a rational function, $P \in \mathbb{P}^1(K)$ and $\alpha \in \mathbb{P}^1(K)$.
	Let $z$ denote the inhomogeneous coordinate of $\mathbb{P}^1(K)$.
	We define the multiplicity of $f = \alpha$ at $P$, say $n_P$, as follows.
	\begin{enumerate}
		\item
		If $f(P) \ne \alpha$, then we define $n_P = 0$.

		\item
		If $f$ is the constant function $\alpha$, then we define $n_P = 0$.

		For the remaining cases, we assume that $f(P) = \alpha$ and $f$ is not a constant function.

		\item
		If neither $P$ nor $\alpha$ is $\infty$, then $f(z) - \alpha$ is expressed as
		\[
			f(z) - \alpha = (z - P)^m \frac{g(z)}{h(z)},
		\]
		where $g(P) h(P) \ne 0$ and $m > 0$.
		We define $n_P = m$.

		\item
		If $P = \infty$ and $\alpha \ne \infty$, then $f(z) - \alpha$ is expressed as
		\[
			f(z) - \alpha = \frac{g(z)}{h(z)},
		\]
		where $g(z)$ and $h(z)$ are coprime.
		We define $n_P = \deg_z h(z) - \deg_z g(z)$.

		\item
		If $P \ne \infty$ and $\alpha = \infty$, then $f(z)$ is expressed as
		\[
			f(z) = (z - P)^{- m} \frac{g(z)}{h(z)},
		\]
		where $g(P) h(P) \ne 0$ and $m > 0$.
		We define $n_P = m$.

		\item
		If both $P$ and $\alpha$ are $\infty$, then $f(z)$ is expressed as
		\[
			f(z) = \frac{g(z)}{h(z)},
		\]
		where $g(z)$ and $h(z)$ are coprime.
		We define $n_P = \deg_z g(z) - \deg_z h(z)$.

	\end{enumerate}
\end{definition}

The reason we define $n_P = 0$, rather than $\infty$, for the constant function $f = \alpha$ is that this choice ensures Proposition~\ref{proposition:degree_preimages}.

\begin{definition}
	A divisor on $\mathbb{P}^1(K)$ is a finite formal sum over the points of $\mathbb{P}^1(K)$ with $\mathbb{Z}$-coefficients.
	We use $\operatorname{Div} \mathbb{P}^1(K)$ to denote the set of divisors on $\mathbb{P}^1(K)$.
\end{definition}

\begin{definition}\label{definition:divisor}
	For a rational function $f \colon \mathbb{P}^1(K) \to \mathbb{P}^1(K)$, $\alpha \in \mathbb{P}^1(K)$ and a subset $T \subset \mathbb{P}^1(K)$, we define $\operatorname{div} (f = \alpha \mid T) \in \operatorname{Div} \mathbb{P}^1(K)$ as
	\[
		\operatorname{div} (f = \alpha \mid T) = \sum_{P \in T} n_P P,
	\]
	where $n_P$ is the multiplicity of $f = \alpha$ at $P$.
	We simply write
	\[
		\operatorname{div} (f = \alpha) = \operatorname{div} (f = \alpha \mid \mathbb{P}^1(K)).
	\]
	We sometimes use a condition that defines a set such as
	\[
		\operatorname{div} (f = \alpha \mid T, \ g \ne 0) = \operatorname{div} \left( f = \alpha \mid \{ z \in T \mid g(z) \ne 0 \} \right).
	\]
\end{definition}

\begin{lemma}
	Let $f \colon \mathbb{P}^1(K) \to \mathbb{P}^1(K)$ be a rational function and let $(T_{\lambda})_{\lambda \in \Lambda}$ be a family of disjoint subsets of $\mathbb{P}^1(K)$.
	Then, we have
	\[
		\operatorname{div} \left( f = \alpha \mid \bigsqcup_{\lambda} T_{\lambda} \right)
		= \sum_{\lambda} \operatorname{div} \left( f = \alpha \mid T_{\lambda} \right).
	\]
\end{lemma}
\begin{proof}
	The statement immediately follows from the definition.
\end{proof}

\begin{definition}\label{definition:divisor_degree_support}
	For a divisor $D = \sum_{P \in \mathbb{P}^1(K)} n_P P$, we define its degree and support as
	\[
		\deg D = \sum_{P \in \mathbb{P}^1(K)} n_P, \quad
		\operatorname{supp} D = \{ P \in \mathbb{P}^1(K) \mid n_P \ne 0 \}.
	\]
\end{definition}

The following, well-known proposition plays a basic role in our strategy.

\begin{proposition}\label{proposition:degree_preimages}
	For a rational function $f \colon \mathbb{P}^1(K) \to \mathbb{P}^1(K)$ and an arbitrary value $\alpha \in \mathbb{P}^1(K)$, the following equality holds:
	\[
		\deg \left( \operatorname{div} \left( f = \alpha \right) \right) = \deg f,
	\]
	where the right-hand side is the degree of $f$ as a rational function.
	In particular, the degree of a rational function coincides with the number of preimages counted with multiplicity.
\end{proposition}

That is, the degree of a rational function can be calculated as the number of the preimages of any value with multiplicity.




\begin{thebibliography}{99}


	\bibitem{alonso_suris_wei1}
	J. Alonso, Y. B. Suris, K. Wei,
	Dynamical degrees of birational maps from indices of polynomials with respect to blow-ups I. General theory and 2D examples,
	\textit{preprint},
	arXiv:2303.15864.
	\url{https://doi.org/10.48550/arXiv.2303.15864}


	\bibitem{alonso_suris_wei2}
	J. Alonso, Y. B. Suris, K. Wei,
	Dynamical degrees of birational maps from indices of polynomials with respect to blow-ups II. 3D examples,
	\textit{preprint},
	arXiv:2307.09939.
	\url{https://doi.org/10.48550/arXiv.2307.09939}


	\bibitem{bedford_kim_higher}
	E. Bedford, K. Kim,
	On the degree growth of birational mappings in higher dimension,
	\textit{The Journal of Geometric Analysis}
	14
	(2004):
	567--596.
	\url{https://doi.org/10.1007/BF02922170}


	\bibitem{bedford_kim}
	E. Bedford and K. Kim,
	Continuous families of rational surface automorphisms with positive entropy,
	\textit{Mathematische Annalen}
	348
	(2010):
	667--688.
	\url{https://doi.org/10.1007/s00208-010-0498-2}


	\bibitem{entropy}
	M. P. Bellon, C.-M. Viallet,
	Algebraic entropy,
	\textit{Communications in Mathematical Physics}
	204
	(1999):
	425--437.
	\url{https://doi.org/10.1007/s002200050652}


	\bibitem{cantat}
	S. Cantat.
	Dynamique des automorphismes des surfaces projectives complexes,
	\textit{Comptes Rendus de l'Acad{\'e}mie des Sciences.\ S{\'e}rie I. Math{\'e}matique},
	328
	(1999):
	901--906.
	\url{https://doi.org/10.1016/S0764-4442(99)80294-8}


	\bibitem{carstea_takenawa_4d}
	A. S. Carstea, T. Takenawa,
	Space of initial conditions and geometry of two 4-dimensional discrete Painlev{\'e} equations,
	\textit{Journal of Physics A: Mathematical and Theoretical}
	52
	(2019):
	275201.
	\url{https://doi.org/10.1088/1751-8121/ab2253}


	\bibitem{diller_favre}
	J. Diller, C. Favre,
	Dynamics of bimeromorphic maps of surfaces,
	\textit{American Journal of Mathematics}
	123
	(2001):
	1135--1169.
	\url{https://doi.org/10.1353/ajm.2001.0038}
	

	\bibitem{laurent_phenomenon}
	S. Fomin, A. Zelevinsky 2002,
	The Laurent phenomenon,
	\textit{Advances in Applied Mathematics}
	28
	(2002):
	119--144.
	\url{https://doi.org/10.1006/aama.2001.0770}
	

	\bibitem{gizatullin}
	M. K. Gizatullin,
	Rational $G$-surfaces,
	\textit{Izvestiya Rossiiskoi Akademii Nauk.\ Seriya Matematicheskaya}
	44
	(1980):
	110--144.
	\url{https://doi.org/10.1070/IM1981v016n01ABEH001279}


	\bibitem{singularity_confinement}
	B. Grammaticos, A. Ramani, V. Papageorgiou,
	Do integrable mappings have the Painlev\'e property?,
	\textit{Physical Review Letters}
	67
	(1991):
	1825--1828.
	\url{https://doi.org/10.1103/PhysRevLett.67.1825}


	\bibitem{redeeming}
	B. Grammaticos, A. Ramani, R. Willox, T. Mase, J. Satsuma,
	Singularity confinement and full-deautonomisation: a discrete integrability criterion,
	\textit{Physica D}
	313
	(2015):
	11--25.
	\url{https://dx.doi.org/10.1016/j.physd.2015.09.006}
	
	
	\bibitem{sc_mkdv1}
	B. Grammaticos, T. Thamizharasi, R. Willox,
	On the singularity structure of a discrete modified-Korteweg-deVries equation,
	\textit{Journal of Physics A: Mathematical and Theoretical}
	55
	(2022):
	\url{https://doi.org/10.1088/1751-8121/ac711b}


	\bibitem{gubbiotti_scimiterna_levi}
	G. Gubbiotti, C. Scimiterna, D. Levi,
	Algebraic entropy, symmetries and linearization of quad equations consistent on the cube,
	\textit{Journal of Nonlinear Mathematical Physics}
	23
	(2016)
	507--543.
	\url{https://doi.org/10.1080/14029251.2016.1237200}


	\bibitem{face_centered_quad}
	G. Gubbiotti, A. P. Kels.
	Algebraic entropy for face-centered quad equations,
	\textit{Journal of Physics A: Mathematical and Theoretical}
	54
	(2021):
	455201.
	\url{https://doi.org/10.1088/1751-8121/ac2aeb}


	\bibitem{lattice_system_entropy}
	G. Gubbiotti,
	Algebraic entropy for systems of quad equations,
	\textit{Open Communications in Nonlinear Mathematical Physics}
	Special Issue in Memory of Decio Levi
	(2024):
	11638.
	\url{https://doi.org/10.46298/ocnmp.11638}


	\bibitem{halburd}
	R. G. Halburd,
	Elementary exact calculations of degree growth and entropy for discrete equations,
	\textit{Proceedings of the Royal Society A: Mathematical, Physical and Engineering Sciences}
	473
	(2017):
	20160831.
	\url{https://doi.org/10.1098/rspa.2016.0831}


	\bibitem{hamad_kamp}
	K. Hamad, P. H. van der Kamp,
	From discrete integrable equations to Laurent recurrences,
	\textit{Journal of Difference Equations and Applications}
	22
	(2016):
	\url{https://doi.org/10.1080/10236198.2016.1142980}
	

	\bibitem{laurentification_qrt}
	K. Hamad, A. N. W. Hone, P. H. van der Kamp, G. R. W. Quispel,
	QRT maps and related Laurent systems,
	\textit{Advances in Applied Mathematics}
	96
	(2018):
	\url{https://doi.org/10.1016/j.aam.2017.12.006}


	\bibitem{hietarinta_viallet}
	J. Hietarinta, C. Viallet,
	Singularity confinement and chaos in discrete systems,
	\textit{Physical Review Letters}
	81
	(1998):
	325--328.
	\url{https://doi.org/10.1103/PhysRevLett.81.325}


	\bibitem{lattice_factorization}
	J. Hietarinta, C. Viallet,
	Searching for integrable lattice maps using factorization,
	\textit{Journal of Physics A: Mathematical and Theoretical}
	40
	(2007):
	12629.
	\url{https://doi.org/10.1088/1751-8113/40/42/S09}


	\bibitem{domain}
	J. Hietarinta, T. Mase, R. Willox,
	Algebraic entropy computations for lattice equations: why initial value problems do matter,
	\textit{Journal of Physics A: Mathematical and Theoretical}
	52
	(2019):
	49LT01.
	\url{https://doi.org/10.1088/1751-8121/ab5238}


	\bibitem{hietarinta_3x3}
	J. Hietarinta,
	Degree growth of lattice equations defined on a 3 $\times$ 3 stencil,
	\textit{Open Communications in Nonlinear Mathematical Physics}
	Special Issue in Memory of Decio Levi
	(2024):
	11589.
	\url{https://doi.org/10.46298/ocnmp.11589}


	\bibitem{hirota_dkdv}
	R. Hirota,
	Nonlinear partial difference equations.\ I: A difference analogue of the Korteweg-de Vries equation,
	\textit{Journal of the Physical Society of Japan}
	43
	(1977):
	1424--1433.
	\url{https://doi.org/10.1143/JPSJ.43.1424}


	\bibitem{hone2007}
	A. N. W. Hone,
	Laurent polynomials and superintegrable maps,
	\textit{SIGMA. Symmetry, Integrability and Geometry: Methods and Applications}
	3
	(2007):
	\url{https://doi.org/10.3842/SIGMA.2007.022}


	\bibitem{extended_laurent}
	R. Kamiya, M. Kanki, T. Mase, T. Tokihiro,
	Nonlinear forms of coprimeness preserving extensions to the Somos-$4$ recurrence and the two-dimensional Toda lattice equation -- investigation into their extended Laurent properties,
	\textit{Journal of Physics A: Mathematical and Theoretical}
	51
	(2018):
	\url{https://doi.org/10.1088/1751-8121/aad074}


	\bibitem{venderkamp2009}
	P. H. van der Kamp,
	Initial value problems for lattice equations,
	\textit{Journal of Physics A: Mathematical and Theoretical}
	42
	(2009):
	404019.
	\url{https://dx.doi.org/10.1088/1751-8113/42/40/404019}


	\bibitem{exkdv}
	R. Kamiya, M. Kanki, T. Mase, T. Tokihiro,
	Coprimeness-preserving discrete KdV type equation on an arbitrary dimensional lattice,
	\textit{Journal of Mathematical Physics}
	62
	(2021):
	102701.
	\url{https://doi.org/10.1063/5.0034581}


	\bibitem{coprimeness}
	M. Kanki, J. Mada, T. Mase, T. Tokihiro,
	Irreducibility and co-primeness as an integrability criterion for discrete equations,
	\textit{Journal of Physics A: Mathematical and Theoretical}
	47
	(2014):
	465204.
	\url{https://dx.doi.org/10.1088/1751-8113/47/46/465204}
	
	
	\bibitem{exhv}
	M. Kanki, T. Mase, T. Tokihiro,
	Algebraic entropy of an extended Hietarinta-Viallet equation,
	\textit{Journal of Physics A: Mathematical and Theoretical}
	48
	(2015):
	355202.
	\url{https://dx.doi.org/10.1088/1751-8113/48/35/355202}


	\bibitem{2dchaos}
	M. Kanki, T. Mase, T. Tokihiro,
	Singularity confinement and chaos in two-dimensional discrete systems,
	\textit{Journal of Physics A: Mathematical and Theoretical}
	49
	(2016):
	23LT01.
	\url{https://dx.doi.org/10.1088/1751-8113/49/23/23LT01}


	\bibitem{factorize}
	M. Kanki, T. Mase, T. Tokihiro,
	On the coprimeness property of discrete systems without the irreducibility condition,
	\textit{SIGMA. Symmetry, Integrability and Geometry: Methods and Applications}
	14
	(2018):
	065.
	\url{https://dx.doi.org/10.3842/SIGMA.2018.065}

	
	\bibitem{levi_yamilov_mkdv}
	D. Levi, R. I. Yamilov,
	On a nonlinear integrable difference equation on the square,
	\textit{Ufimski{\u{\i}} Matematicheski{\u{\i}} Zhurnal}
	1
	(2009):
	101--105.
	\url{http://mathnet.ru/eng/ufa/v1/i2/p101}


	\bibitem{investigation}
	T. Mase,
	Investigation into the role of the Laurent property in integrability,
	\textit{Journal of Mathematical Physics}
	57
	(2016):
	022703.
	\url{https://doi.org/10.1063/1.4941370}


	\bibitem{phd}
	T. Mase,
	Studies on spaces of initial conditions for nonautonomous mappings of the plane,
	\textit{Journal of Integrable Systems}
	3
	(2018):
	xyy010.
	\url{https://dx.doi.org/10.1093/integr/xyy010}


	\bibitem{express2}
	T. Mase, R. Willox, A. Ramani, B. Grammaticos,
	Singularity confinement as an integrability criterion,
	\textit{Journal of Physics A: Mathematical and Theoretical}
	52
	(2019):
	205201.
	\url{https://dx.doi.org/10.1088/1751-8121/ab1433}


	\bibitem{laurent_domain}
	T. Mase,
	A study on the domain independence of the Laurent property, the irreducibility and the coprimeness in lattice equations,
	\textit{Open Communications in Nonlinear Mathematical Physics}
	4
	(2024):
	79--113.
	\url{https://doi.org/10.46298/ocnmp.13960}

	\bibitem{lattice_sc_classic}
	A. Ramani, B. Grammaticos, J. Satsuma,
	Integrability of multidimensional discrete systems,
	\textit{Physics Letters.\ A}
	169
	(1992):
	323--328.
	\url{https://doi.org/10.1016/0375-9601(92)90235-E}


	\bibitem{sc_mkdv2}
	A. Ramani, B. Grammaticos, J. Satsuma, R. Willox,
	On two (not so) new integrable partial difference equations,
	\textit{Journal of Physics A: Mathematical and Theoretical}
	42
	(2009):
	282002.
	\url{https://doi.org/10.1088/1751-8113/42/28/282002}


	\bibitem{sc_linearizable}
	A. Ramani, B. Grammaticos, J. Satsuma, N. Mimura,
	Linearizable QRT mappings,
	\textit{Journal of Physics A: Mathematical and Theoretical}
	44
	(2011):
	425201.
	\url{https://dx.doi.org/10.1088/1751-8113/44/42/425201}
	

	\bibitem{redemption}
	A. Ramani, B. Grammaticos, R. Willox, T. Mase, M. Kanki,
	The redemption of singularity confinement,
	\textit{Journal of Physics A: Mathematical and Theoretical}
	48
	(2015):
	11FT02.
	\url{https://dx.doi.org/10.1088/1751-8113/48/11/11FT02}
	

	\bibitem{miura_transformations}
	A. Ramani, B. Grammaticos, R. Willox,
	Miura transformations for discrete Painlev\'{e} equations coming from the affine $E_8$ Weyl group,
	\textit{Journal of Mathematical Physics}
	58
	(2017):
	043502.
	\url{https://doi.org/10.1063/1.4979794}
	

	\bibitem{express1}
	A. Ramani, B. Grammaticos, R. Willox, T. Mase,
	Calculating algebraic entropies: an express method,
	\textit{Journal of Physics A: Mathematical and Theoretical}
	50
	(2017):
	185203.
	\url{https://dx.doi.org/10.1088/1751-8121/aa66d7}

	
	\bibitem{express_unconfined}
	A. Ramani, B. Grammaticos, R. Willox, T. Mase, J. Satsuma,
	Calculating the algebraic entropy of mappings with unconfined singularities,
	\textit{Journal of Integrable Systems}
	3
	(2018):
	xyy006.
	\url{https://doi.org/10.1093/integr/xyy006}


	\bibitem{tran1}
	J. A. G. Roberts, D. T. Tran, 
	Algebraic entropy of (integrable) lattice equations and their reductions,
	\textit{Nonlinearity}
	32
	(2019):
	622--653.
	\url{https://doi.org/10.1088/1361-6544/aaecda}


	\bibitem{tran2}
	J. A. G. Roberts, D. T. Tran, 
	Linear degree growth in lattice equations,
	\textit{Journal of Computational Dynamics}
	6
	(2019):
	449--467.
	\url{https://doi.org/10.3934/jcd.2019023}



	\bibitem{takenawa1}
	T. Takenawa,
	A geometric approach to singularity confinement and algebraic entropy,
	\textit{Journal of Physics A: Mathematical and General}
	34
	(2001):
	L95--L102.
	\url{https://doi.org/10.1088/0305-4470/34/10/103}


	\bibitem{tremblay}
	S. Tremblay, B. Grammaticos, A. Ramani,
	Integrable lattice equations and their growth properties,
	\textit{Physics Letters.\ A}
	278
	(2001):
	\url{https://doi.org/10.1016/S0375-9601(00)00806-9}
	

	\bibitem{sc_kdv1}
	D. Um, R. Willox, B. Grammaticos, A. Ramani,
	On the singularity structure of the discrete KdV equation,
	\textit{Journal of Physics A: Mathematical and Theoretical}
	53
	(2020):
	114001.
	\url{https://doi.org/10.1088/1751-8121/ab72af}


	\bibitem{sc_kdv2}
	D. Um, A. Ramani, B. Grammaticos, R. Willox, J. Satsuma,
	On the singularities of the discrete Korteweg-deVries equation,
	\textit{Journal of Physics A: Mathematical and Theoretical}
	54
	(2021):
	\url{https://doi.org/10.1088/1751-8121/abd8f4}


	\bibitem{lattice_algebraic_entropy}
	C.-M. Viallet,
	Algebraic Entropy for lattice equations,
	\textit{preprint},
	arXiv:math-ph/0609043.
	\url{https://doi.org/10.48550/arXiv.math-ph/0609043}


	\bibitem{viallet_2015}
	C. M. Viallet,
	On the algebraic structure of rational discrete dynamical systems,
	\textit{Journal of Physics A: Mathematical and Theoretical}
	48
	(2015):
	16FT01.
	\url{https://doi.org/10.1088/1751-8113/48/16/16FT01}


\end{thebibliography}
\end{document}